\documentclass[letterpaper,11pt]{article}

\usepackage{authblk}

\usepackage{amsmath,bm}
\usepackage{amsthm}
\usepackage{amssymb}

\usepackage{etoolbox}
\usepackage{mathtools}

\usepackage{algorithmic}
\usepackage{algorithm}

\usepackage{tikz}
\usetikzlibrary{shapes.geometric,arrows,decorations.markings}

\usepackage{xcolor}

\usepackage{xfrac}

\usepackage{colortbl}
\usepackage{array,arydshln,multirow}

\usepackage{paralist}

\usepackage{thmtools}

\usepackage{graphicx}

\usepackage[left=1.0in,top=1.0in,right=1.0in,bottom=1.0in,nohead]{geometry}

\usepackage{wrapfig}

\usepackage{cite}

\theoremstyle{definition}
\newtheorem{definition}{Definition}

\newtheorem*{theorem*}{Theorem}
\newtheorem{theorem}{Theorem}

\newtheorem{lemma}[theorem]{Lemma}

\newtheorem{prop}[theorem]{Proposition}

\newtheorem*{prop*}{Proposition}

\patchcmd{\subequations}{}%
{}{}{}

\graphicspath{{figs/},{new_figs/}}

\makeatletter
\newcommand{\leqnomode}{\tagsleft@true}
\newcommand{\reqnomode}{\tagsleft@false}
\makeatother

\newcommand{\BIGBP}[1]{\left\{ #1\right\}}

\newcommand{\MC}[1]{{\mathcal #1}}
\newcommand{\MB}[1]{{\mathbf #1}}

\long\def\longdelete#1{}

\title{\bf Improved LP-based Approximations for Facility Location with Hard Capacities}

\author[ ]{Mong-Jen Kao}

\affil[ ]{Department of Computer Science,}
\affil[ ]{National Yang-Ming Chiao-Tung University, Taiwan.}
\affil[ ]{\small Email: \texttt{mjkao@nycu.edu.tw}}

\date{}

\begin{document}

\begin{titlepage}

\maketitle
\thispagestyle{empty}

\begin{abstract}
The Capacitated Facility Location (CFL), a long-standing classic problem with intriguing approximability and literature dated back to the 90s, is considered.
Following the open question posted in~[Williamson and Shmoys, 2011] and the notable work due to~[An et al., FOCS~2014], 
we present an LP-based approximation algorithm
with a guarantee of $\left(10+\sqrt{67}\right)/2 \approx 9.0927$, a significant improvement upon the previous LP-based ratio of $288$ due to An et al. in~2014.
Our contribution for this part is a simple and elegant rounding algorithm that brings clear insights for the MFN relaxation and the CFL problem.

\smallskip

For CFL with cardinality facility cost (CFL-CFC), 
we present an LP-based $4$-approximation algorithm,
which improves upon the decades-old ratio of~$5$ due to Levi~et~al.~that ages up since 2004.
%
Prior to our work, it was not clear whether or not LP-based methods can be used to provide a guarantee better than~$5$ for the CFL problem, even for restricted versions of this problem, for which natural LPs are already known to have small integrality gaps.
Our rounding algorithm provides the first affirmative answer on the case with cadinality facility cost.
\end{abstract}




%
\end{titlepage}

\section{Introduction}

We consider the facility location problem with hard capacities (CFL), a long-standing problem with 
intriguing unsettled approximability and 
literature dated back to the 90s.
In this problem, we are given a set $\MC{F}$ of facilities, a set $\MC{D}$ of clients, and a distance metric $c$ defined over $\MC{F} \cup \MC{D}$.
Each $i \in \MC{F}$ is associated with an open cost $o_i$ and a capacity $u_i$,
which is the number of clients it can serve when opened up.
The cost of serving a client $j$ using a facility $i$ equals the distance 
between them.
The goal is to determine a set of facilities to open up and an assignment of the clients to the opened facilities that respects their capacity 
limits so as to minimize that the total cost.
%

\smallskip

The CFL problem
was first considered by Shmoys, Tardos, and Aardal in~\cite{10.1145/258533.258600}. 
%
For facilities with uniform capacities, Koropolu et al.~\cite{10.5555/314613.314616} showed that the local search heuristic proposed by Kuehn and Hamburger~\cite{10.1287/mnsc.9.4.643} yields a constant factor approximation.
Chudak and Williamson~\cite{10.5555/645589.659788} improved the analysis of Korupolu et al.~\cite{10.5555/314613.314616} and obtained a $(6,5)$-approximation, i.e., a solution whose cost is bounded by 6 times the facility cost plus 5 times the service cost of an optimal solution.
Aggarwal et al.~\cite{DBLP:journals/mp/AggarwalLBGGGJ13} introduced the idea of taking suitable linear combinations of inequalities which captures the local optimality and obtained a $3$-approximation.

\smallskip

For facilities with non-uniform capacities, i.e., the general CFL problem, Pal et al.~\cite{10.5555/874063.875600} presented a $(9,5)$-approximation based on local search algorithm.
%
The ratio was improved by Mahdian and Pal~\cite{DBLP:conf/esa/MahdianP03} to an $(8,7)$-approximation.
Zhang et al.~\cite{DBLP:conf/ipco/ZhangCY04} introduced the idea of multi-exchange local operations and further improved the ratio to $(6,5)$.
%
The algorithm was later modified by Bansal et al.~\cite{DBLP:conf/esa/BansalGG12} to achieve a $5$-approximation, which is the best ratio known for the CFL problem.

\smallskip

In contrast to the rich LP-based toolboxes developed for the uncapacitated facility location problem (UFL), the fact that no LP-based algorithms with constant approximation guarantee were known for CFL was intriguing and surprising.
In fact, devising an LP-based approximation with constant guarantee for CFL was listed as one of ten open problems in the textbook on approximation algorithm due to Williamson and Shmoys~\cite{10.5555/1971947}.
This problem was resolved by the notable work
of An, Singh, and Svensson~\cite{DBLP:conf/focs/AnSS14}, in which a strong multi-commodity flow network (MFN)
relaxation
with constant integrality gap is presented.
The approximation guarantee obtained in this work, however, is in the order of 
288, and it remains an open problem to devise a better LP-based guarantee for CFL or a better integrality gap for the MFN relaxation.
%

\smallskip

In the pursuit of settling down the approximability of CFL, an important variation between the general problem and the case with uniform capacities is when we have cardinality-type facility costs (CFL-CFC), i.e., uniform facility cost for which $o_i = 1$ for all $i \in \MC{F}$.
This was studied by Levi, Shmoys, and Swamy~\cite{DBLP:conf/ipco/LeviSS04}, in which an LP-based $5$-approximation was presented.
Interestingly, the ratio of $5$ remained to be the best known ratio for the next $17$ years since 2004.
%

%

\smallskip

The hard capacitated problems have drawn a wide attention in the past two decades, with new understandings and techniques blossomed.
While some of these problems are shown to share the same approximability with their uncapacitated versions~\cite{10.5555/3039686.3039860,DBLP:conf/soda/CheungGW14}, many of others appear to be of greater difficulty to deal with~\cite{DBLP:journals/mp/AnBCGMS15,DBLP:conf/focs/CyganHK12,DBLP:conf/soda/Li15,DBLP:conf/focs/AnSS14}.
One primary reason for this phenomenon is that, 
the hard capacity constraint renders most of existing techniques, in particular, LP-based techniques developed for the uncapacitated versions, not directly applicable in ensuring the feasibility, and 
complicated constructions with compromise are often made to deliver a solid guarantee.
The existing LP-based result for the CFL problem~\cite{DBLP:conf/focs/AnSS14} is one of such examples.
%
%
%
With the usage of reassignable partial assignments, the MFN relaxation provides a way to handle the CFL problem with a bounded integrality gap.
The right rounding methodology
for this category of problems, however, 
appears to be yet to be found, be it the CFL problem, or its restricted variations.


%
%
%
%
%
%
%


%

%

%

\subsection{Our Contribution}

%
Following the open question posted by Williamson and Shmoys~\cite{10.5555/1971947} and the notable work due to An~et~al.~\cite{DBLP:conf/focs/AnSS14}, 
we present a simple and elegant rounding-based approximation algorithms with a significantly improved LP-based guarantee for the CFL problem.
%
%
%
Our result for CFL is the following theorem.

\begin{theorem} \label{thm-approx-cfl}
There is an LP-based algorithm for CFL that produces a $\left(10+\sqrt{67}\right)/2 \approx 9.0927$-approximation in polynomial-time.
\end{theorem}

This significantly
improves upon the previous ratio of 288 due to An~et~al.~\cite{DBLP:conf/focs/AnSS14} in~2014.
Our algorithm is built on an iterative rounding scheme for the MFN relaxation that combines several new insights and novel ideas with a couple of techniques developed in the past~\cite{zakmsp02,DBLP:conf/ipco/LeviSS04,DBLP:conf/focs/AnSS14}.
In addition, it has the characteristic of being simple and elegant in that no sophisticated construction is involved.
%
%
We believe that such simplicity is essential and beneficial in the further development of this problem. 

\smallskip

In addition to the general CFL problem, we present an improved approximation algorithm for CFL with cardinality facility cost (CFL-CFC).
Our result for this part is the following theorem.
%

%

\begin{theorem} \label{thm-approx-cfl-ufc}
There is a rounding-based algorithm for 
CFL-CFC that produces a $4$-approximation in polynomial-time.
\end{theorem}

This result yields an improvement upon the decades-old ratio of $5$ due to Levi et al.~\cite{DBLP:conf/esa/BansalGG12} that ages up for 17 years
since 2004.
%
Prior to this work, it was not even clear whether LP-based methods can be used to provide a guarantee better than~$5$ for the CFL problem, even for restricted versions of this problem, for which simple natural LPs are already known to have small integrality gaps.
Our result provides the first affirmative answer on the case with cadinality facility cost.

\smallskip

%
%
%
%
The algorithm we propose for CFL-CFC
is a delicate coordination of a two-staged iterative clustering scheme which incorporates a set of novel ideas with techniques developed in the past for both facility location and capacitated covering problems~\cite{10.5555/3039686.3039860,DBLP:conf/soda/CheungGW14,DBLP:conf/ipco/LeviSS04}. 
%
%
%
%
We believe that, the rounding techniques we develop in this work are of independent interests and will lead to further insights and progress for related problems.

%


%
%
%

%
%

%

%

%


\paragraph{Overview of Algorithms and Techniques.}

The core part of our results can be seen as rounding procedures that handles small instances incurred in the LP relaxations,
i.e., the rounding decisions for the small facilities 
and the assignments made to them.
As was illustrated implicitly by An~et~al.~in~\cite{DBLP:conf/focs/AnSS14}, the true power of the MFN relaxation lies in its ability to remove the large facilities from consideration, using the assignments made to them as the extra price.
As a result, what remains is the rounding problem for a relaxation of the small instance.
%

\smallskip

Our procedures 
aim at fractionally serving the clients while making sure that the rounded facilities are reasonably sparsely-loaded by the assignments, so that a small final round-up on the assignments can be made to guarantee the feasibility.
%
The sparsity of the small facilities is guaranteed by default.
In our algorithm, we make the observation that, 
with proper construction, the large facilities are also sparsely-loaded by the flow sent to them, and a reasonable final rounding blow-up of $1/(1-\alpha)$ can be made when necessary.
This characteristic is distinguishable to~\cite{DBLP:conf/focs/AnSS14}, in which the small instance is created by showing that, there exists a feasible flow that sends a firm fraction of $1/2$ demand from each client of interests to the small facilities.
%
%
%

\smallskip

Our rounding procedure for the small facilities builds around the idea due to Abraham~et~al.~\cite{zakmsp02} and prior works developed for uncapacitated facility location.
In each iteration, the facility with the least \emph{per-unit-flow-rerouting-cost}
is selected to be rounded, and all the flow along with the facility value in the vicinity is rerouted {simultaneously} and {proportionally} to the selected facility.
%
%
%
%
%
%
%

\smallskip

To bound the extra rerouting cost incurred during the rounding process, one essential element is to guarantee a low assignment radius for each client.
In~\cite{DBLP:conf/focs/AnSS14,zakmsp02}, this is done by applying the so-called filtering technique,
%
which is basically to apply the Markov's inequality to cut-off long-range assignments followed by unconditional round-up.
This inevitably creates a tremendous blowup in the final guarantee.
In our algorithm, we rebalance the instance in each iteration with a carefully designed LP and use the primal-dual schema in an implicit way
to obtain an exact pricing on the cost, which in turn
%
bounds the assignment radius tightly for each client that gets reassigned.
The LP we design also enables, in a subtle yet crucial way, our rounding algorithm to evenly balance the facility cost and the assignment cost we spend in each iteration.
%
%

\medskip

Our rounding algorithm for CFL-CFC is a delicate coordination between rounding procedures for the large and the small facilities.
%
In contrast to our previous result, the large facilities can be tightly-loaded by the assignments made in the natural LP solutions.
Hence, they do not allow a final round-up of the assignments in general.
%
%
%
%

\smallskip

To overcome this issue, we introduce the concept of \emph{client redistribution}:
When the residue demand of a client drops below a target threshold, we discard the client and redistribute part of it to the large facilities in the vicinity, defined by the LP solution, to form the so-called ``\emph{outlier clients}.''
The outlier clients participate in the rounding process after created and act as normal clients except for that, there is no threshold for them to be discarded, and we guarantee that they will be fully-assigned for the final feasibility.
Moreover, the way the outlier clients are created also guarantees that, the resulting assignment cost does not increase too much.

\smallskip

The concept of client redistribution resolves the assignment of the clients. 
However, when an outlier client is inevitably selected to form a cluster, we are no longer able to guarantee the overall rounding cost of the facilities, since the total facility value can be arbitrarily small, rendering the rounding error unbounded.
To prevent this undesirable situation, we formulate the rounding decisions for the \emph{outlier clusters} as another instance of CFL-CFC, in which the large facilities act as the clients and the small facilities act as the normal facilities.
%
%
%
%
%
We use a carefully designed matching-yielding assignment LP, followed by an \emph{unconditional rounding} scheme, for this instance.
%

\smallskip

To bound the cost incurred, we deploy a technique, which was originally developed for the capacitated covering problems~\cite{DBLP:conf/soda/CheungGW14,10.5555/3039686.3039860}, to show that, basic feasible solutions of this simple LP corresponds naturally to a matching from the non-integral facilities to the large facilities at which the outlier clients reside.
Hence the rounding cost of these small facilities can be bounded.
%

%

%

%
%

%


\paragraph{Organization of this paper.}

The rest of this paper is organized as follows.
In the rest of this section we introduce the MFN relaxation for the CFL problem.
We present our approximation algorithm for CFL in Section~\ref{sec-approx-cfl} and our approximation algorithm for CFC-CFC in Section~\ref{sec-approx-cfl-cfc}.

\smallskip

The additional content is organized as follows.
We establish the approximation guarantee for CFL in Section~\ref{sec-proof-approx-cfl}, page~\pageref{sec-proof-approx-cfl}, and the guarantee for CFC-CFC in Section~\ref{sec-proof-approx-cfl-cfc}, page~\pageref{sec-proof-approx-cfl-cfc}.
%
%

%

%

%


\subsection{Preliminaries}

%

%
In the CFL problem, we are given a set $\MC{F}$ of facilities, a set $\MC{D}$ of clients, and a distance metric $c$ defined over $\MC{F} \cup \MC{D}$.
Each $i \in \MC{F}$ is associated with an open cost $o_i$ and a capacity $u_i$, which is the number of clients it can serve when opened up.
%
%
%
%
A feasible solution for CFL consists of a multiplicity function $y \colon \MC{F} \rightarrow \{0,1\}$ and an assignment function $x \colon \MC{F} \times \MC{D} \rightarrow \{0,1\}$ such that the following conditions are met:
%
%
		(a) $\sum_{i \in \MC{F}}x_{i,j} \ge 1$, for each $j \in \MC{D}$. 
%
		(b) $\sum_{j \in \MC{D}}x_{i,j} \le u_i \cdot y_i$, for each $i \in \MC{F}$. 
%
		(c) $x_{i,j} \le y_i$, for each $i \in \MC{F}, j \in \MC{D}$.
%
%
%
The cost of a solution $(x,y)$ is defined to be 
$\psi(x,y) \; := \; \sum_{i \in \MC{F}} o_i \cdot y_i \; + \; \sum_{i \in \MC{F}, \; j \in \MC{D}} c_{i,j} \cdot x_{i,j}.$
Given an instance $\Psi = (\MC{F},\MC{D}, \MB{c}, \MB{o}, \MB{u})$ of CFL,
the goal of this problem is to compute an integral solution $(x,y)$ such that $\psi(x,y)$ 
is minimized.

%

%




%


%

%
\begin{figure*}[h]
%
\vspace{4pt}
\begin{minipage}{\textwidth}
\fbox{
\enskip
\begin{minipage}{.39\textwidth}
\centering
\includegraphics[scale=1]{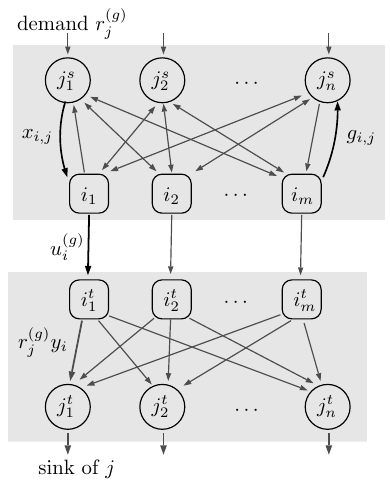}
\end{minipage}
\enskip
}
\enskip
\fbox{
\hspace{-10pt}
\begin{minipage}{.51\textwidth}
\begin{subequations}
\begin{align}
& \sum_{i \in \MC{F}, \; p \in P(i,j)} \hspace{-4pt} f_p \; \ge \; r^{(g)}_j, & & \forall j \in \MC{D}, \label{LP_MFN_1} \\[6pt]
& \sum_{p \in P \colon (j^s,i) \in p} \hspace{-6pt} f_p \; \le \; x_{i,j}, & & \hspace{-8pt} \forall i \in \MC{F}, j \in \MC{D}, \label{LP_MFN_2} \\[6pt]
& \sum_{p \in P \colon (i,j^s) \in p} \hspace{-6pt} f_p \; \le \; g_{i,j}, & & \hspace{-8pt} \forall i \in \MC{F}, j \in \MC{D}, \label{LP_MFN_3} \\[6pt]
& \sum_{j \in \MC{D}, \; p \in P(i,j)} \hspace{-10pt} f_p \; \le \; u^{(g)}_i \hspace{-4pt} \cdot y_i, & & \forall i \in \MC{F}, \label{LP_MFN_4} \\[6pt]
& \sum_{p \in P(i,j)} f_p \; \le \; r^{(g)}_j \hspace{-4pt} \cdot y_i, & & \hspace{-8pt} \forall i \in \MC{F}, j \in \MC{D},  \label{LP_MFN_5} \\[6pt]
& \hspace{10pt} f_p \; \ge \; 0, & & \forall p \in P. \label{LP_MFN_6}
\end{align}
\end{subequations}
\vspace{-8pt}
\end{minipage}
\enskip
}
\end{minipage}
\caption{The construction of $\mathbf{MFN}_\Psi(x,y,g)$ and the corresponding LP constraints.}
\label{cons-MFN-g}
\vspace{-8pt}
\end{figure*}
%

%


\paragraph{The MFN relaxation.}

%
As natural LP formulations for CFL are known to have an unbounded integrality gap even for simple settings,
%
%
%
An, Singh, and Svensson~\cite{DBLP:conf/focs/AnSS14} introduced a strong LP relaxation based on multicommodity flow networks (MFN).
%
The idea is to impose Knapsack-cover type constraints, formulated as reassignable \emph{partial assignments} given as free in each qualifying test.
%
%

\begin{definition}[Partial Assignments]
A partial assignment $g$ is 
a function $g\colon \MC{F} \times \MC{D} \rightarrow [0,1]$.
The partial assignment $g$ is said to be valid if
		\; (i) $\hspace{1pt} \sum_{i \in \MC{F}} g_{i,j} \le 1$, for each $j \in \MC{D}$, and
%
		\; (ii) $\hspace{1pt} \sum_{j \in \MC{D}} g_{i,j} \le u_i$, for each $i \in \MC{F}$.
%
%
\end{definition}

\smallskip

Given a candidate fractional solution $(x,y)$ and a valid partial assignment $g$, 
the multi-commodity flow network with respect to $(x,y,g)$, denoted $\mathbf{MFN}_\Psi(x,y,g)$ and to be defined in the following, gives a qualifying test on the validity of the candidate solution $(x,y)$.

\begin{definition}[Multi-commodity Flow Network]
For a valid partial assignment $g$ and a candidate solution $(x,y)$, $\mathbf{MFN}_\Psi(x,y,g)$ is a multicommodity flow network 
defined as follows.
\begin{itemize}
	\item
		Each client $j \in \MC{D}$ corresponds to two nodes $j^s$ and $j^t$ in the network and is associated with a commodity $j$ with source-sink pair $(j^s, j^t)$ and demand $r^{(g)}_j := 1-\sum_{i \in \MC{F}} g_{i,j}$.
		
	\item
		Each facility $i \in \MC{F}$ corresponds to two nodes $i$ and $i^t$ that are connected by an arc $(i,i^t)$ of capacity $u^{(g)}_i := y_i \cdot \left( u_i - \sum_{j \in \MC{D}} g_{i,j} \right)$.
		
	\item
		For each $j \in \MC{D}$ and each $i \in \MC{F}$, there is an arc $(j^s,i)$ of capacity $x_{i,j}$, an arc $(i,j^s)$ of capacity $g_{i,j}$, and an arc $(i^t,j^t)$ of capacity $r^{(g)}_j \cdot y_i$.
\end{itemize}
\end{definition}

See also Figure~\ref{cons-MFN-g} for an illustration on the construction of $\mathbf{MFN}_\Psi(x,y,g)$ and the corresponding constraints.
For any $i \in \MC{F}, j \in \MC{D}$, let $P^{(g)}_{(x,y)}(i,j)$ to denote the set of paths in $\mathbf{MFN}_\Psi(x,y,g)$ for commodity $j$ to sink via $i^t$. 
The superscript $(g)$ and the subscript $(x,y)$ is omitted when there is no confusion in the context.
Let $\vphantom{\text{\Large T}} P := \cup_{i \in \MC{F}, j \in \MC{D}} P(i,j)$ to denote the set of all possible paths.
%
%
%


%
%


%

%

%
%
%
\begin{lemma}[An, Singh, Svensson~\cite{DBLP:conf/focs/AnSS14}] \label{lemma-an-MFN-separation}
Given an instance $\Psi = (\MC{F},\MC{D}, \MB{c}, \MB{o}, \MB{u})$ of CFL,
the constraints defined by
\vspace{-8pt}
$$\mathbf{MFN}_\Psi(x,y) := \left\{ \; \vphantom{\text{\LARGE T}} \mathbf{MFN}_\Psi(x,y,g) \text{ feasible} \; \colon \; \text{$\forall$ valid $g$} \; \right\}$$
is a valid relaxation for integral solutions on $\Psi$.
Furthermore, 
the separation problem for the feasibility of $\mathbf{MFN}_\Psi(x,y,g)$ for any valid $g$ can be answered in weakly polynomial-time, and a basic feasible flow can be obtained.
\end{lemma}

\paragraph{On the need for facility-saturating partial assignments.}

To illustrate to what extent the MFN relaxation can constraint the fractional solutions to yield a bounded integrality gap for the CFL problem,
let us consider the following example.

\smallskip

Suppose that there are two facilities and $n+1$ clients, where $o_1 = 0$, $o_2 = 1$, $u_1 = u_2 = n$, and $c_{i,j} = 0$ for all $i,j$.
Then, for any $0\le \epsilon \le 1$, making an $\epsilon n$ amount of partial assignments to facility $1$ in the MFN relaxation, e.g., setting $g_{1,j} = \epsilon n / (n+1)$, only guarantees that $y_2 \ge 1/((1-\epsilon)n+1)$.
Hence, when partial assignments are needed to eliminate the low-cost facilities, saturating them with partial assignments is necessary.
%
The way how the clients are partially assigned, however, does not appear to affect the resulting integrality gap.
%

%

%

%

%


\section{LP-based Approximation for CFL}
\label{sec-approx-cfl}


%
In this section we describe our approximation algorithm for CFL and establish Theorem~\ref{thm-approx-cfl}.
%
%
%
%
%
%
The algorithm applies the Ellipsoid framework used in~\cite{DBLP:conf/focs/AnSS14}, which aims to either round candidate solutions with the claimed approximation guarantee
or to assert the infeasibility of the solution.
%

\smallskip

For completeness, we briefly sketch the framework.
%
Then we describe our construction for obtaining a sparsely-loading flow and the more interesting part of the iterative rounding process.
%
%



%
\paragraph{The Outer Framework.}
\begin{wrapfigure}[6]{r}{.31\textwidth}
\hspace{-1pt}
\begin{minipage}{.26\textwidth}
\vspace{-12pt}
\fbox{
\hspace{-4pt}
\begin{minipage}{\textwidth}
%
\vspace{-9pt}
\begin{align}
& \mathbf{MFN}_\Psi(x,y),  \label{LP-MFN} \tag*{LP-(1)} \\[3pt]
& \psi(\bm{x}, \bm{y}) \le \gamma,  \notag \\[4pt]
& \bm{x} \in [ \hspace{1pt} 0, 1 \hspace{1pt} ]^{\MC{F} \times \MC{D}}, \; \bm{y} \in [ \hspace{1pt} 0,1 \hspace{1pt} ]^{\MC{F}}. \notag
\end{align}
\vspace{-16pt}
\end{minipage}\hspace{-1pt}
}
\end{minipage}
\end{wrapfigure}
%
%
%
The framework starts by guessing the optimal cost using standard binary search.
For each guess, say $\gamma$, the Ellipsoid algorithm is applied on~\ref{LP-MFN}.
For each separation problem incurred during the Ellipsoid algorithm,
we apply Theorem~\ref{thm-MFN-relaxed-separation-oracle}, which is stated below, for either a separating hyperplane or an integral solution with the claimed approximation guarantee.
%
%
%
%
When an integral solution is successfully rounded, or when the Ellipsoid algorithm concludes the infeasibility of the guess $\gamma$, the framework continues to the next iteration of the binary search process until the desired precision is attained.
%
%
%
%
%
It suffices to prove the following theorem.

\begin{theorem}	\label{thm-MFN-relaxed-separation-oracle}
Given 
a candidate solution $(\bm{x}, \bm{y})$ for~LP-(MFN)
and a target parameter $\alpha$ with $0 < \alpha \le 1/3$, we can compute in polynomial-time either (i) a 
separating hyperplane for $(\bm{x}, \bm{y})$
and~LP-(MFN),
or (ii) an integral solution $(\bm{x}^*, \bm{y}^*)$, rounded from $(\bm{x}, \bm{y})$, for $\Psi$ with 
%
$$\psi(x^*, y^*) \le \max\left\{ \; \frac{3}{2\alpha} \; , \; \frac{7-4\alpha}{(1-\alpha)^2} \; \right\} \times \psi(x,y).$$
\end{theorem}

\noindent
Note that, 
the particular choice of $\alpha := \left(10-\sqrt{67}\right)/11$ completes the statement of Theorem~\ref{thm-approx-cfl}.
%
%
%
%
%
In the rest of this section, we describe our rounding algorithm that establishes the statement of Theorem~\ref{thm-MFN-relaxed-separation-oracle}.
We provide the analysis and 
finish the proof 
in Section~\ref{sec-proof-approx-cfl}, page~\pageref{sec-proof-approx-cfl}.

%
%


%


\paragraph{Initial Classification.}

%
%

%

%


%
Classify the facilities as follows.
%
Let $I := \left\{ \; i \in \MC{F} \; \colon \; 0 < y_i < \alpha \; \right\}$
and $U := \left\{ \; i \in \MC{F} \; \colon \; y_i \ge \alpha \; \right\}$.
The facilities in $U$ are further classified into two categories. Let 
\vspace{-2pt}
$$\hspace{0.2cm} U^{(>)} := \left\{ \vphantom{\text{\huge T}} \right. \; i \in U \; \colon \; \sum_{j \in \MC{D}} x_{i,j} \; > \; (1-\alpha) \cdot u_i  \;\; \left. \vphantom{\text{\huge T}} \right\}
\quad \text{and} \quad 
U^{(\le)} := U \setminus U^{(>)}.$$
Provided that the facilities in $U$ are to be rounded up in the approximate solution, we know that the assignments to $U^{(\le)}$ are ready to be rounded up by a factor of $1/(1-\alpha)$.
We refer the facilities in $U^{(\le)}$ to as \emph{sparsely-loaded} by the assignments made in $\bm{x}$.
%


%


\paragraph{Obtaining an Initial Sparsely-Loading Flow.}

%
Consider the bipartite graph $G = \left(\MC{D},U^{(>)},E\right)$, 
where there exists an edge $(j,i)$ in $E$ with
%
edge capacity $ x_{i,j} / (1-\alpha)$ for each $j \in \MC{D}$ and $i \in U^{(>)}$.
Solve~\ref{LP-b-matching} for an optimal $\bm{h}$ for the maximum b-matching problem on $G$.

\smallskip

\begin{wrapfigure}[8]{r}{.35\textwidth}
\hspace{-2pt}
\begin{minipage}{.31\textwidth}
\vspace{-11pt}
\fbox{
\hspace{-6pt}
\begin{minipage}{\textwidth}
\vspace{-9pt}
\begin{align}
& \text{max}  \hspace{-4pt} \sum_{\;\; i\in U^{(>)}, \; j \in \MC{D}} \hspace{-6pt}  h_{i,j}   \label{LP-b-matching} \tag*{LP-(2)} \\
& \hspace{-8pt} \sum_{\;\; i \in U^{(>)}} \hspace{-1pt} h_{i,j} \le 1,  &  \hspace{-22pt}  \forall j \in \MC{D},  \notag \\
& \sum_{j \in \MC{D}} \; h_{i,j} \le u_i,  &  \hspace{-22pt}  \forall i \in U^{(>)},  \notag \\
& \hspace{2pt} 0 \le \bm{h} \le \bm{x} / (1-\alpha). \notag
\end{align}
\vspace{-18pt}
\end{minipage}\hspace{1pt}
}
\end{minipage}
\end{wrapfigure}
For any $j \in \MC{D}$, the client $j$ is said to be 
partially-assigned
if $\vphantom{\text{\LARGE T}_{\text{\Large T}} } \sum_{i \in U^{(>)}} h_{i,j} < 1$ and fully-assigned otherwise.
%
%
We say that 
a path $P$ in $G$ is an \emph{augmenting path} if the following holds.
\begin{itemize}
	\item
		$P$ starts at a partially-assigned client $j \in \MC{D}$.
	\item
		For each $(j',i') \in P$ with $j' \in \MC{D}, i'\in U^{(>)}$, \newline
		we have $h_{i',j'} < x_{i',j'} / (1-\alpha)$.
		
	\item
		For each $(i',j') \in P$ with $i' \in U^{(>)}, j' \in \MC{D}$, \newline
		we have $h_{i',j'} > 0$.
\end{itemize}
%
\noindent
Intuitively, an augmenting path 
is a way to increasing the assignment of a partially-assigned client $j$ without altering the optimality of $\bm{h}$.
%
%
%
%
We say that a facility $i \in U^{(>)}$ is {tightly-occupied} if 
it is reachable in $G$ from a partially-assigned client via an augmenting path.
%

\smallskip

Let $U^{(\phi)}$ denote the set of all tightly-occupied facilities.
For each $i \in \MC{F}, \; j \in \MC{D}$, define 
%
\vspace{-2pt}
$$g_{i,j} := \begin{cases}
\; h_{i,j}, & \text{if $i \in U^{(\phi)}$,} \\
\; 0, & \text{otherwise,}
\end{cases}
\quad \text{and} \quad
y'_i := \begin{cases}
\; 1, & \text{if $i \in U$,} \\
\; y_i, & \text{otherwise.}
\end{cases}$$
%
%
%
%
%
%
Apply Lemma~\ref{lemma-an-MFN-separation} for either a basic feasible flow $\bm{f}$ or a 
separating hyperplane for $\mathbf{MFN}_\Psi(\bm{x},\bm{y}',\bm{g})$.
%
If a separating hyperplane is found, the algorithm reports it and terminates.
%
%

%

%
\begin{figure*}[h]
\centering
\fbox{
\hspace{-10pt}
\begin{minipage}{.8\textwidth}
\vspace{-6pt}
\begin{subequations}
\begin{align}
\text{min} \;\; & \;\; \sum_{i \in I'} \; o_i \cdot y_i \; + \; \sum_{i \in I, \; j \in D'} \; c_{i,j} \cdot x_{i,j} & & \label{LP_ITR} \tag*{LP-(M)} \\[4pt]
\text{s.t.} \;\; & \;\; \sum_{i \in I'} \; x_{i,j} \; \ge \; r'_j, & & \forall j \in D', \label{LP_ITR_1} \\[4pt]
& \;\; \hspace{-1pt} \sum_{j \in D'} \; x_{i,j} \; \le \; u_i \cdot y_i, & & \forall i \in I', \label{LP_ITR_2} \\
& \;\; x_{i,j} \; \le \; \frac{2\alpha}{1-\alpha} \cdot r^{(g)}_j \cdot y_i, & & \forall i \in I, j \in D',  \label{LP_ITR_3} \\[2pt]
& \;\; y_i \; \le \frac{1-\alpha}{2}, & & \forall i \in I',  \label{LP_ITR_4} \\[2pt]
& \;\; x_{i,j} \; \ge \; 0, \; y_i \ge 0, & & \forall i \in I', \; j \in D'. \label{LP_ITR_5}
\end{align}
\end{subequations}
\vspace{-12pt}
\end{minipage}
\enskip
}
\vspace{-4pt}
\end{figure*}


\paragraph{Iterative Rounding for the Small Facilities.}

In the following we describe our iterative rounding process for the small facilities in $I$ 
using the information computed in the previous stage.

\smallskip

During the rounding process, the algorithm maintains a parameter tuple $\Psi' = (I', D', \bm{r}')$ which corresponding to the remaining instance to be processed.
Initially, $I' := I$, 
$$r'_j := \sum_{i \in I, \; p \in P(i,j)} f_p, \;\; \text{for all $j \in \MC{D}$,} \quad\; \text{and} \quad\; 
D' := \left\{ \; j \in \MC{D} \; \colon \; r'_j \; > \; \alpha \cdot r^{(g)}_j \; \right\},$$ 
where $r^{(g)}_j = 1 - \sum_{i \in U^{(\phi)}} g_{i,j}$ is the total demand of $j$.
%
%
The rounding algorithm updates the tuple $\Psi'$ in iterations until $D'$ becomes empty.

\smallskip

%
%
In each iteration, the algorithm solves~\ref{LP_ITR} on the current tuple $\Psi'$ for an optimal $(\overline{\bm{x}}, \overline{\bm{y}})$.
%
%
%
Depending on the status of $\overline{\bm{y}}$, the algorithm selects a facility $i \in I'$ to form a cluster and defines the scaling factor $\sigma^{(i)}_k$ and $\sigma^{(i)}_{k,j}$ for all $k \in I'$ and $j \in D'$.
%
%
\begin{itemize}
	\item
		If $\overline{y}_i = (1-\alpha)/2$ for some $i \in I'$, 
		%
		then the algorithm picks such an $i$ with $\overline{y}_i = (1-\alpha)/2$ from $I'$ and sets $\sigma^{(i)}_k = \sigma^{(i)}_{k,j} = 0$ for all $k \in I' \setminus \{i\}$ and $j \in D'$.

	\item
		If $\overline{y}_i < (1-\alpha)/2$ for all $i \in I'$,
		%
		%
		then the algorithm selects among the facilities $i \in I'$ with $\overline{y}_i > 0$ the particular $i$ with the minimum $\theta(i)$, 
		where $\theta(i)$ is 
		defined as
		\begin{align*}
		\theta(i) \; := \; \frac{1}{\sum_{j \in D'} \overline{x}_{i,j}} \cdot \left( \vphantom{\text{\huge T}} \right. \;\; 3\cdot o_i\cdot \overline{y}_i \; + \; 2\cdot \sum_{j \in D'} c_{i,j} \cdot \overline{x}_{i,j} \;\; \left. \vphantom{\text{\huge T}} \right).
		\end{align*}
%
Intuitively, the selected facility $i$ has the lowest per-unit-assignment rerouting cost, and any other facility in $I'$ can afford the rerouting cost within the cluster centered at $i$.

\smallskip

%
%
For each $j \in D'$, define 
\vspace{-6pt}
$$\delta^{(i)}_j := \left( \; \frac{1-\alpha}{2} \cdot \frac{1}{\overline{y}_i} - 1 \; \right) \cdot \overline{x}_{i,j}$$ 
to be the amount of assignments to be gathered from the vicinity of facility $i$ to $i$ via commodity $j$.
%
%
%
%
For each $k \in I' \setminus \{i\}$, define the fraction of $k$, to be sent to $i$, as
$$\sigma^{(i)}_k \; := \; \frac{1}{\sum_{\ell \in D'} \overline{x}_{k,\ell}} \cdot \sum_{j \in D'} \sigma^{(i)}_{k,j}, 
\quad \text{where} \enskip 
\sigma^{(i)}_{k,j} \; := \; \frac{\overline{x}_{k,j}}{\sum_{\ell \in I' \setminus \{i\}} \overline{x}_{\ell,j}} \cdot \delta^{(i)}_j$$
is the amount of assignments to be sent via commodity $j$.
Note that, from the definition it follows that $\sum_{k \in I' \setminus \{i\}} \sigma^{(i)}_{k,j} = \delta^{(i)}_j$, and the amount to be gathered via $j$ can be fulfilled.
\end{itemize}

\smallskip

\noindent
For consistency, also define $\sigma^{(i)}_i = 1$ and $\sigma^{(i)}_{i,j} = \overline{x}_{i,j}$ for all $j \in D'$.

\bigskip

The algorithm 
updates the parameter tuple $\Psi'$ as follows.
\begin{itemize}
	\item
		For each $j \in D'$, the algorithm decreases $r'_j$ by $\sum_{k \in I'} \sigma^{(i)}_k \cdot \overline{x}_{k,j}$.
%
%
	\item
		The algorithm removes $i$ from $I'$ and all $j\in D'$ with $r'_j < \alpha \cdot r^{(g)}_j$ from $D'$.
\end{itemize}
Then the algorithm proceeds to the next iteration until $D'$ becomes empty.

%

%


\paragraph{Final Output.}

\begin{wrapfigure}[8]{r}{.34\textwidth}
\begin{minipage}{.32\textwidth}
\vspace{-6pt}
\fbox{
\hspace{-12pt}
\begin{minipage}{\textwidth}
\vspace{-9pt}
\begin{align}
& \text{min}  \hspace{-8pt} \sum_{\;\; i\in F^*, \; j \in \MC{D}} \hspace{-6pt}  c_{i,j} \cdot x_{i,j}   \label{LP-mc-assignment} \tag*{LP-(3)} \\
& \hspace{-4pt} \sum_{\;\; i \in F^*} \hspace{-1pt} x_{i,j} \ge 1,  &  \hspace{-22pt}  \forall j \in \MC{D},  \notag \\
& \sum_{j \in \MC{D}} \; x_{i,j} \le u_i,  &  \hspace{-22pt}  \forall i \in F^*,  \notag \\
& \hspace{2pt} \bm{x} \ge 0. \notag
\end{align}
\vspace{-18pt}
\end{minipage}\hspace{-2pt}
}
\end{minipage}
\end{wrapfigure}
When $D'$ becomes empty, define 
$$y^*_i \; := \; \begin{cases}
\; 1, & \text{if $i \in U \cup \left( I \setminus I' \right)$,} \\
\; 0, & \text{otherwise.}
\end{cases}$$
%
The algorithm solves 
the min-cost assignment problem on $\MC{D}$ and $\MC{F}^* := \left\{ \hspace{1pt} \vphantom{\text{\Large T}} i \in \MC{F} \; \colon \; y^*_i = 1 \hspace{1pt} \right\}$ for an optimal integral assignment $\bm{x}^*$
%
and outputs $(\bm{x}^*, \bm{y}^*)$ as the approximate solution.

\medskip

This completes the description for our rounding algorithm.
We provide the analysis and establish the statements of Theorem~\ref{thm-MFN-relaxed-separation-oracle} in Section~\ref{sec-proof-approx-cfl}, page~\pageref{sec-proof-approx-cfl}.

%


%

%
%

%

%

%

%


%


%


\section{$4$-Approximation for CFL-CFC}
\label{sec-approx-cfl-cfc}

Let $\Psi = (\MC{F},\MC{D}, \MB{c}, \MB{u})$ be an instance of CFL-CFC.
%
%
In this section, we describe our algorithm that establishes the statement of Theorem~\ref{thm-approx-cfl-ufc}.
We begin with the following natural LP relaxation.
%


%


%
\begin{figure*}[h]
\centering
\begin{minipage}{.44\textwidth}
\fbox{\hspace{-6pt}
\begin{minipage}{\textwidth}
\vspace{-6pt}
%
\begin{align}
& \; \text{min} \;\; \sum_{i \in \MC{F}} \; y_i + \hspace{-2pt} \sum_{i \in \MC{F}, j \in \MC{D}} c_{i,j} \cdot x_{i,j} & & \label{LP-natural-CFL} \tag*{LP-(N)}  \\[4pt]
& \; \sum_{i\in \MC{F}} \; x_{i,j} \; \ge \; 1, & & \hspace{-0.6cm} \forall j\in \MC{D}, 
\notag \\[2pt]
& \; \sum_{j \in \MC{D}} \; x_{i,j} \; \le \; u_i \cdot y_i, & & \hspace{-0.6cm}  \forall i \in \MC{F},  
\notag \\[1pt]
& \;\; 0 \; \le \; x_{i,j} \le y_i,   & & \hspace{-1.6cm}   \forall i \in \MC{F}, j \in \MC{D}, 
\notag \\[6pt]
& \;\; 0 \; \le \; y_i \; \le \; 1,   & & \hspace{-0.6cm}   \forall i \in \MC{F}. \notag
\end{align}
\vspace{-19pt}
\end{minipage}\hspace{3pt}
}
\end{minipage}
\quad\hspace{-3pt}
\begin{minipage}{.48\textwidth}
\fbox{\hspace{-10pt}
\begin{minipage}{\textwidth}
%
\begin{align}
\; & \; \text{max} \enskip \sum_{j \in \MC{D}} \alpha_j \; - \; \sum_{i \in \MC{F}} \; \eta_i  & & \label{LP-natural-dual} \tag*{LP-(DN)} \\[4pt]
\; & \; \alpha_j \; \le \; \beta_i \; + \; \Gamma_{i,j} \; + \; c_{i,j},  & & \hspace{-0.6cm} \forall i \in \MC{F}, j\in \MC{D},  \notag \\[7pt]
\; & \; u_i \cdot \beta_i \; + \; \sum_{j \in \MC{D}} \Gamma_{i,j} \; \le \; 1 \; + \; \eta_i,  & & \forall i \in \MC{F},  \notag \\[2pt]
\; & \; \alpha_j, \; \beta_i, \; \Gamma_{i,j}, \; \eta_i \; \ge \; 0,  & & \hspace{-0.6cm} \forall i \in \MC{F}, j \in \MC{D}.  \notag
\end{align}
\vspace{-12pt}
\end{minipage}\hspace{1pt}
}
\end{minipage}
%
%
\end{figure*}
%

%

%
We use the following notion 
of \emph{vicinity} with respect to any assignment function of interests, say, $\bm{x}$.
%
%
For any $A \subseteq \MC{F}$
and any $j \in \MC{D}$, we 
use $N_{(A,x)}(j) := \left\{ \; \vphantom{\text{\Large T}} i \in A \; \colon \; x_{i,j} > 0 \; \right\}$ to denote the set of facilities in $A$ to which $j$ is assigned to in $x$.
Similarly, for any $B \subseteq \MC{D}$ and any $i \in \MC{F}$, we 
use $N_{(B,x)}(i) := \left\{ \; \vphantom{\text{\Large T}} j \in B \; \colon \; x_{i,j} > 0 \; \right\}$ to denote the set of clients in $B$ that is assigned to $i$ in $x$.


\medskip

%
%

%

Let $(\bm{x}', \bm{y}')$, $(\bm{\alpha}, \bm{\beta}, \bm{\Gamma}, \bm{\eta})$ be optimal solutions for~\ref{LP-natural-CFL} and its dual~\ref{LP-natural-dual} on $\Psi$.
%
In the rest of this section, we describe our rounding algorithm for $(\bm{x}', \bm{y}')$.
%


%



%
%
%

%


%


\paragraph*{Initial Classification.}

Let $I := \left\{ \; i \in \MC{F} \; \colon \; 0 < y'_i < \frac{1}{2} \; \right\}$
and
$U := \left\{ \; i \in \MC{F} \; \colon \; y'_i \ge \frac{1}{2} \; \right\}.$
%
%
%
%
%
The clients in $\MC{D}$ are divided into three categories, namely, those that are served merely by $I$, those that are served jointly by $I$ and $U$, and those that are served merely by $U$.
Let $$J^{(I)} := \left\{ \vphantom{\text{\LARGE T}} \; j \in \MC{D} \hspace{1pt} \colon \hspace{1pt} x'_{i,j} = 0 \enskip \text{for all $i \in U$} \; \right\}, \enskip J^{(\leftrightarrow)} := \left\{ \; j \in \MC{D} \hspace{1pt} \colon \min\left( \max_{i \in I}x'_{i,j} \hspace{1pt} , \hspace{1pt} \max_{i \in U}x'_{i,j} \right) > 0 \; \right\},$$
and $J^{(U)} := \MC{D} \setminus \left( J^{(I)} \cup J^{(\leftrightarrow)} \right)$ denote the clients in the three categories, respectively.
%
%



\paragraph*{The Rounding Process.}

%
%
Let $F'$ and $D'$ be the set of facilities and the set of clients remained to be processed, and $x^*$ be the rounded assignment function our algorithm will maintain during the process.
Initially, $F' := I$ and $D' := J^{(I)} \cup J^{(\leftrightarrow)}$, and $x^* := 0$. 
%


%
%


%
%

\smallskip

The rounding algorithm consists of two stages.
In the first stage, it proceeds in iterations to form clusters.
In each of such iterations, the algorithm checks if 
$\sum_{i \in F'} x'_{i,j} \ge 1/2$ holds for all $j \in D'$.
%
If not, the algorithm makes it so by repeatedly removing small clients from $D' \cap J^{(\leftrightarrow)}$ and redistributing their demand to facilities in $U$ to form a set of \emph{outlier clients}.
%
We use $H$ to denote the set of outlier clients created in this step and $H' \subseteq H$ to denote those that are created but not yet processed by the rounding algorithm. Initially $H := \emptyset$ and $H' := \emptyset$.

\smallskip

When $\sum_{i \in F'} x'_{i,j} \ge 1/2$ holds for all $j \in D'$,
a cluster centered at a client is formed and possibly rounded, depending on whether or not the client forming the cluster is outlier, and the corresponding parts of the cluster are removed from $F'$, $D'$, and $H'$, respectively.
%
%
The cluster forming process repeats until $D' \cup H'$ becomes empty.

\smallskip

In the second stage, the algorithm handles the rounding decisions for the remaining clusters centered at the outlier clients, using a carefully-designed assignment LP and an unconditional rounding scheme, to form an integral multiplicity function.
%
%
%
In the following we describe the three components of our rounding algorithm in details.
%



\paragraph*{Creating the Outlier Clients.}

When 
$\sum_{i \in F'}x'_{i,j} < 1/2$ for some $j \in J^{(\leftrightarrow)} \cap D'$, the algorithm discards $j$
and relocates part of the remaining demand to facilities in $N_{(U,x')}(j)$ to form outlier clients in a way as if the demand were originated from these facilities, as illustrated in Figure~\ref{fig-outlier-client-construction}.
%
%
%

\smallskip

\begin{wrapfigure}[12]{r}{.29\textwidth}
\vspace{0pt}
\centering
\begin{minipage}{.26\textwidth}
\fbox{
\hspace{-8pt}
\includegraphics[scale=1]{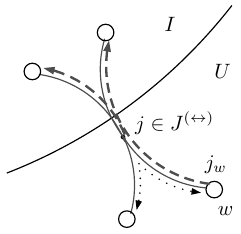}
\hspace{-4pt}
\vspace{-6pt}
}
\vspace{-16pt}
\caption{
Construction of the outlier clients.}
\label{fig-outlier-client-construction}
\end{minipage}
\end{wrapfigure}
%
Let $r'_j := \min\{ \; \sum_{i \in F'} x'_{i,j} \; , \; \sum_{i \in U} x'_{i,j} \; \}$ be the amount of residue demand of $\vphantom{\text{\Large T}} j$ to be redistributed.
For each $w \in N_{(U,x')}(j)$, we create a client $j_w$ at the facility $w$ with demand 
%
%
$$d_{j_w} := \frac{r'_j}{\sum_{i \in U}x'_{i,j}}\cdot x'_{w,j} \enskip\enskip \text{and set} \enskip  x'_{i,j_w} := \frac{d_{j_w}}{\sum_{k \in F'}x'_{k,j}} \cdot x'_{i,j}$$
for each $i \in N_{(F',x')}(j)$.
See also Figure~\ref{fig-outlier-client-construction} for an illustration.
%
%
%
We add $j_w$ to both $H$ and $H'$ and set $\vphantom{\text{\Large T}} \alpha_{j_w} := \alpha_j + c_{w,j}$.
%
%
%

\smallskip

After the outlier client $j_w$ is created for each $w \in N_{(U,x')}(j)$, the algorithm removes $j$ from $D'$ and set $x'_{i,j}$ to be zero for all $i \in F'$.
%
%
%
%
Note that, by construction, 
%
we have 
$$\sum_{w \in N_{(U,x')}(j)} d_{j_w} = r'_j \quad \text{and} \enskip \sum_{k \in N_{(F',x')}(j)} x'_{k,j_w} = d_{j_w}. \hspace{4.5cm}$$
Hence, the designated residue demand of $j$ is fully redistributed and
%
%
%
each $j_w$ is fully-assigned.
%
%
%
%
%
%
%
%
%

%

%


\paragraph*{Forming Clusters and Rounding.}

When $\sum_{i \in F'} x'_{i,j} \ge 1/2$ holds for all $j \in D'$,
the algorithm selects a client $j \in D' \cup H'$ that minimizes $\alpha_j$ to form a cluster.
%
%
%
%
%
Depending on the set to which $j$ belongs, the algorithm proceeds differently.
%
%
\begin{itemize}
\item
If $j \in H'$, then a cluster centered at $j$ with satellite facilities in $N_{(F',x')}(j)$ is formed.
We use $B(j) := N_{(F',x')}(j)$ to denote the set of satellite facilities at this moment.
%
%
%
%
The algorithm then removes $j$ from $H'$ and $B(j)$ from $F'$.
The rounding problem for this cluster is handled later in the second phase of the algorithm.
%


\item
If $j \in D'$, the algorithm further selects a facility $i \in N_{(F',x')}(j)$ with the maximum $u_i$.
The algorithm relocates the assignments and facility values 
from the facilities in $N_{(F',x')}(j)$ to $i$ as follows.
Let 
$$\vspace{-4pt} \delta_i := \left( \; \frac{1}{2} - y'_i \; \right) \cdot \frac{1}{\sum_{k \in N_{(F',x')}(j) \setminus \{i\}}y'_k}$$
be the factor 
to relocate for each facility in $\vphantom{\text{\LARGE T}} N_{(F',x')}(j) \setminus \{i\}$.
%
%
%
%
%
%
%
%

\smallskip

For each facility $\ell \in N_{(F', x')}(j) \setminus \{i\}$, the algorithm 
scales down $y'_\ell$ by
$(1-\delta_i)$. 
For each $k \in N_{(D'\cup H',x')}(\ell)$, the algorithm further scales down $x'_{\ell,k}$ by $(1-\delta_i)$ and increases 
$x^*_{i,k}$ by the same amount $x'_{\ell,k}$ has decreased in this step.


%
The algorithm increases $x^*_{i,k}$ by $x'_{i,k}$ for each $k \in D'$ and
%
%
%
then removes $i$ from $F'$.
%


%
%
%
%
\end{itemize}


\noindent
When the above updates 
are done, for each client $k \in J^{(I)} \cap D'$ with $\sum_{i \in F'} x'_{i,k} < 1/2$, the algorithm removes $k$ from $D'$ and sets $x'_{i,k}$ to be zero for all $i \in F'$.
%
Then the algorithm proceeds to the next iteration until $D' \cup H'$ becomes empty.

\paragraph*{Rounding the Outlier Clusters.}

When the cluster-forming process ends and 
$D' \cup H'$ becomes empty, the algorithm proceeds to 
processes the rounding decisions left for the outlier clusters, i.e., clusters centered at 
outlier clients in $H$.

\smallskip

\begin{wrapfigure}[11]{r}{.42\textwidth}
%
\centering
\vspace{-22pt}
\begin{minipage}{.38\textwidth}
\fbox{
\hspace{-6pt}
\begin{minipage}{\textwidth}
\vspace{-12pt}
\begin{align}
\label{LP-outliers} \tag*{LP-(O)}
\end{align}

\vspace{-16pt}
\noindent\hfill\rule{.98\textwidth}{.3pt}
\vspace{-6pt}
\begin{align}
& \text{min} \enskip \sum_{i \in G} \; y_i \; + \sum_{i \in G, \; j \in U} c_{i,j} \cdot x_{i,j} & &  \notag
\\
& \sum_{i\in G} \; x_{i,j} \; = \; d_j, & & \hspace{-32pt} \forall j \in U,  \notag \\[2pt]
& \sum_{j \in U} \; x_{i,j} \; \le \; u_i \cdot y_i, & & \hspace{-32pt} \forall i \in G,  \notag \\
& \; y_i \; \le \; 1, & & \hspace{-32pt} \forall i \in G,  \notag \\[1pt]
& \; x_{i,j}, \; y_i \; \ge 0, & & \hspace{-66pt} \forall i \in G, \; j \in U.  \notag
\end{align}
\vspace{-16pt}
\end{minipage}
\hspace{-2pt}
}
\end{minipage}
%
\end{wrapfigure}
%
%
%
We formulate the rounding problem as 
another instance of CFL-CFC with facility set $G := \bigcup_{j \in H} B(j)$ and client set $U$ as follows.

\smallskip

Each 
$w \in U$ is 
associated with a demand $d_w$, 
defined as 
\vspace{-4pt}
$$d_w := \sum_{\substack{k \in H, \\[1pt] k \text{ located at } w}} \; \sum_{i \in B(k), \; \ell \in \MC{D}\cup H} \; t'_\ell \cdot x'_{i,\ell},$$
%
where 
the scaling factor $t'_\ell$ is defined as
$$t'_\ell \; := \; \frac{1}{\sum_{ k \in I} x^*_{k,\ell} + \sum_{k \in G} x'_{k,\ell} } \cdot \left( \vphantom{\text{\huge T}} \right. \; 1 - \sum_{i \in U} x'_{i,\ell} - r'_\ell \; \left. \vphantom{\text{\huge T}} \right)$$
if $\ell \in \MC{D}$ and $\sum_{k \in I} x^*_{k,\ell} + \sum_{k \in G} x'_{k,\ell} > 0$,
and $t'_\ell := 1$ otherwise.
%
%

\medskip

Intuitively, in the above definition, for each $w \in U$, we consider all the outlier clusters centered at clients located at $w$ and collect the demand of these clusters to be the demand of $w$, and $t'_\ell$ is the factor for which the assignments made for the client $\vphantom{\text{\Large T}} \ell$ 
in these clusters should be scaled up in order to amend the amount lost in the previous stage.
%
%

\smallskip

We formulate the above instance 
with a carefully designed assignment LP, denoted~\ref{LP-outliers}.
%
The algorithm solves~\ref{LP-outliers}
%
for a \emph{basic optimal solution} $(\bm{x}'', \bm{y}'')$.
%
%

%


\paragraph*{Final Output.}

Define the integral multiplicity function
$$y^*_i \; := \; \begin{cases}
\; 1, & \text{if $i \in I \setminus F'$,} \\
\; \left\lceil y''_i \right\rceil, & \text{if $i \in G$, } \\
\; 0, & \text{otherwise}.
\end{cases}$$
%
%
The algorithm solves the min-cost assignment problem on $\MC{D}$ and $\MC{F}^* := \{ \hspace{1pt} \vphantom{\text{\Large T}} i \in \MC{F} \; \colon \; y^*_i = 1 \hspace{1pt} \}$ for an optimal integral assignment $\bm{x}^\dagger$,
%
and outputs $(\bm{x}^\dagger, \bm{y}^*)$ as the approximation solution.
%
%
%
%
%
%
%
%

\medskip

We have the following theorem. 
We provide the proof 
in Section~\ref{sec-proof-approx-cfl-cfc}, page~\pageref{sec-proof-approx-cfl-cfc}.

\begin{theorem} \label{theorem-cfl-ufc-approx}
Let $\Psi$ be an instance of CFL-CFC and $(\bm{x}', \bm{y}')$ be optimal for~\ref{LP-natural-CFL} on $\Psi$.
%
The rounding algorithm computes in polynomial time a feasible integral solution $(x^\dagger, y^*)$ for $\Psi$ with
$\psi(x^\dagger, y^*) \le 4\cdot \psi(x',y').$
\end{theorem}

%

%

%

%

%


\section{Proof of Theorem~\ref{thm-MFN-relaxed-separation-oracle}.}
\label{sec-proof-approx-cfl}

%


%
It suffices to prove the following two statements.

\begin{enumerate}
	\item
		The rounding algorithm is well-defined and terminates in polynomial time.

	\item
		Provided that $\mathbf{MFN}_\Psi(\bm{x},\bm{y}',\bm{g})$ is feasible, 
		the feasible region of the min-cost assignment problem on $\MC{D}$ and $F^*$ is non-empty, and $\psi(x^*, y^*) \le \max\left\{ \; \frac{3}{2\alpha} \; , \; \frac{7-4\alpha}{(1-\alpha)^2} \; \right\} \times \psi(x,y).$
\end{enumerate}

%

%
%
We prove the first statement in Section~\ref{subsec-proof-cfl-approx-algo-well-defined}.
To prove the second statement, we define an assignment $\bm{x}'''$ such that $(\bm{x}''', \bm{y}^*)$ is feasible for the min-cost assignment problem on $\MC{D}$ and $F^*$ with the claimed approximation guarantee.

\smallskip

%
This is done as follows.
In Section~\ref{subsec-proof-approx-cfl-p1} and Section~\ref{subsec-proof-approx-cfl-p2}, we define the assignment functions $\bm{x}'$ and $\bm{x}''$ for the assignments made to $U$ and $F^* \cap I$ separately and derive corresponding properties.
%
In Section~\ref{subsec-proof-approx-cfl-feasibility}, 
we define the overall assignment $\bm{x}'''$ and show that $(\bm{x}''', \bm{y}^*)$
forms a feasible solution for the input instance $\Psi$.
%
%
%
%
%
We bound the cost incurred by $(\bm{x}'', \bm{y}^*|_{F^* \cap I})$ in Section~\ref{subsec-cfl-cost-second-stage} and prove 
in Section~\ref{subsec-proof-approx-cfl-approx-guarantee}
that $(\bm{x}''', \bm{y}^*)$ satisfies the claimed approximation guarantee.
%
%
This completes the proof 
since $\bm{x}^*$ is the optimal min-cost assignment between $\MC{D}$ and $\MC{F}^*$.
%


%
\paragraph{Notations to use in the proof.}

To keep the notation precise, for each $i \in F^* \cap I$ that is selected in the second stage, we refer to the particular iteration for which facility $i$ is selected and removed from $I'$ as the $i^{th}$-iteration.
%
%
%
Let $\vphantom{\text{\Large T}} \Psi^{(i)} = \left( \; I'^{(i)}, D'^{(i)}, \bm{r}'^{(i)} \; \right)$ denote the parameter tuple the algorithm maintains in the beginning of the $i^{th}$-iteration.
%
%
We use $\left(\overline{\bm{x}}^{(i)}, \overline{\bm{y}}^{(i)} \right)$ to denote the optimal solution computed for~\ref{LP_ITR} on $\Psi^{(i)}$.

\medskip

We use $\Psi^{(0)} = \left( \; I'^{(0)}, D'^{(0)}, \bm{r}'^{(0)} \; \right)$ to denote the initial tuple 
the algorithm constructs in the beginning of the second stage.
We refer to $\left(\overline{\bm{x}}^{(0)}, \overline{\bm{y}}^{(0)}\right)$ the solution to be defined in Section~\ref{subsec-proof-cfl-approx-algo-well-defined} for~\ref{LP_ITR} on the initial parameter tuple $\Psi^{(0)}$.

%
%
%
%


%

%


\subsection{The feasibility of the algorithm}
\label{subsec-proof-cfl-approx-algo-well-defined}

Consider the first stage of the algorithm.
%
%
Since $\bm{y} \le \bm{y}'$, we have $\mathbf{MFN}_\Psi(\bm{x},\bm{y},\bm{g}) \subseteq \mathbf{MFN}_\Psi(\bm{x},\bm{y}',\bm{g})$.
Hence, should the algorithm outputs a separating hyperplane in the first stage, it must separate $(\bm{x},\bm{y})$ from $\mathbf{MFN}_\Psi(\bm{x},\bm{y},\bm{g})$ as well.

\smallskip

In the following, we assume that $\mathbf{MFN}_\Psi(\bm{x},\bm{y}',\bm{g})$ is feasible and prove that the iterative rounding process is well-defined and runs in polynomial time.
%
%
%
%
Since $\bm{f}$ is a basic solution, the number of paths with nonzero flow in $\bm{f}$ is polynomial in $\MC{F}$ and $\MC{D}$.
Hence, $\bm{r}'^{(0)}$ can be computed in polynomial time.
%
%
%
%
The following lemma, which is proved by verifying the corresponding constraints and the fact that $0 < \alpha \le 1/3$, shows that the feasible region of~\ref{LP_ITR} on $\Psi^{(0)}$ is nonempty.

\begin{lemma} \label{lemma-cfl-approx-iterative-rounding-initial-tuple-feasibility}
The solution $\left(\overline{\bm{x}}^{(0)}, \overline{\bm{y}}^{(0)}\right)$ defined by
\begin{align*}
\overline{x}^{(0)}_{i,j} := \sum_{p \in P(i,j)} f_p \quad \text{for all $i \in I'^{(0)}, j \in D'^{(0)}$} \quad\enskip \text{and} \quad \overline{y}^{(0)}_i := \frac{1-\alpha}{2\alpha} \cdot y_i \quad \text{for all $i \in I'^{(0)}$}
\end{align*}
is feasible for~\ref{LP_ITR} on the initial parameter tuple $\Psi^{(0)}$.
\end{lemma}

\begin{proof}
We prove this lemma by verifying the constraints of~\ref{LP_ITR} separately.
\begin{itemize}
	\item
		Constraint~(\ref{LP_ITR_1}) follows directly from the definition of $\overline{\bm{x}}^{(0)}$ and $\bm{r}'^{(0)}$, since for each $j \in D'^{(0)}$,
		\begin{align*}
		\sum_{i \in I'^{(0)}} \overline{x}^{(0)}_{i,j} \; = \; \sum_{i \in I, \; p \in P(i,j)} f_p \; = \; r'^{(0)}_j.
		\end{align*}
		
	\item
		For Constraint~(\ref{LP_ITR_2}), consider any $i \in I'^{(0)}$.
		Since $\bm{f}$ is feasible for $\mathbf{MFN}_\Psi(\bm{x},\bm{y}',\bm{g})$, we have
		\begin{align*}
		\sum_{j \in D'^{(0)}} \overline{x}^{(0)}_{i,j} \; = \; \sum_{j \in D'^{(0)}, \; p \in P(i,j)} f_p \; \le \; \sum_{j \in \MC{D}, \; p \in P(i,j)} f_p \; \le \; u^{(g)}_i \cdot y_i.
		\end{align*}
		%
		Since $i \in I'^{(0)} = I$, we have $\sum_{j \in \MC{D}} g_{i,j} = 0$ by the way $\bm{g}$ is defined and hence $u^{(g)}_i = u_i$.
		%
		Since $(1-\alpha)/(2\alpha)$ is strictly decreasing for $\alpha > 0$ and
		since $\vphantom{\text{\Large T}} 0 < \alpha \le 1/3$, it follows that $$\overline{y}^{(0)}_i \; := \; \frac{1-\alpha}{2\alpha} \cdot y_i \; \ge \; \frac{ 1- 1/3}{2/3} \cdot y_i \; = \; y_i.$$
		Combining the above, we obtain
		$\sum_{j \in D'^{(0)}} \overline{x}^{(0)}_{i,j} \; \le \; u_i \cdot \overline{y}^{(0)}_i$.
		
	\item
		For Constraint~(\ref{LP_ITR_3}), consider any $i \in I'^{(0)}$ and any $j \in D'^{(0)}$.
		%
		Since $\bm{f}$ is feasible for $\mathbf{MFN}_\Psi(\bm{x},\bm{y}',\bm{g})$, apply the definitions of $\overline{y}^{(0)}_i$ and $r'^{(0)}_j$ and we have
		\begin{align*}
		\overline{x}^{(0)}_{i,j} \; = \; \sum_{p \in P(i,j)} f_p \; \le \; r^{(g)}_j \cdot y_i \; = \; \frac{2\alpha}{1-\alpha} \cdot r'^{(0)}_j \cdot \overline{y}^{(0)}_i.
		\end{align*}
		
	\item
		For Constraint~(\ref{LP_ITR_4}), consider any $i \in I'^{(0)}$.
		%
		Since $y_i < \alpha$ by definition, it follows that
		$$\overline{y}^{(0)}_i \; = \; \frac{1-\alpha}{2\alpha} \cdot y_i \; \le \; \frac{1-\alpha}{2}.$$
\end{itemize}
This proves the lemma.
\end{proof}

Consider the $i^{th}$-iteration for any $i \in F^* \cap I$.
Suppose that~\ref{LP_ITR} on $\Psi^{(i)}$ is feasible, and recall that $\left(\overline{\bm{x}}^{(i)}, \overline{\bm{y}}^{(i)} \right)$ is the optimal solution computed for~\ref{LP_ITR} on $\Psi^{(i)}$.
%
The following lemma shows that the scale-down operation is well-defined.

\begin{lemma} \label{lemma-cfl-approx-iterative-rounding-well-defined}
For any $i \in F^* \cap I$ with $y'^{(i)}_i < (1-\alpha)/2$, the following holds.
\begin{itemize}
	\item
		For any $j \in D'^{(i)}$, 
		$$0 \; \le \; \delta^{(i)}_j \; \le \; \sum_{\ell \in I'^{(i)} \setminus \{i\}} \overline{x}^{(i)}_{\ell,j}.$$
		
	\item
		For any $k \in I'^{(i)} \setminus \{i\}$, we have $0 \le \sigma^{(i)}_k \; \le \; 1$.
\end{itemize}
\end{lemma}

\begin{proof}
Consider any $j \in D'^{(i)}$.
Since $\left( \overline{\bm{x}}^{(i)}, \overline{\bm{y}}^{(i)} \right)$ is feasible for~\ref{LP_ITR} on $\Psi^{(i)}$, 
we have $\overline{y}^{(i)}_i \le (1-\alpha)/2$, which implies that $(1-\alpha) / (2 \cdot \overline{y}^{(i)}_i) \ge 1$.
Hence $$\delta^{(i)}_j \; := \; \left( \; \frac{1-\alpha}{2} \cdot \frac{1}{\overline{y}^{(i)}_i} -1 \; \right) \cdot \overline{x}^{(i)}_{i,j} \; \ge \; 0.$$
On the other hand, applying constraint~(\ref{LP_ITR_3}) and then constraint~(\ref{LP_ITR_1}), we have
$$\delta^{(i)}_j \; := \; \frac{1-\alpha}{2} \cdot \frac{1}{\overline{y}^{(i)}_i} \cdot \overline{x}^{(i)}_{i,j} - \overline{x}^{(i)}_{i,j} \; \le \; \alpha\cdot r^{(g)}_j - \overline{x}^{(i)}_{i,j} \; \le \; r'^{(i)}_j - \overline{x}^{(i)}_{i,j} \; = \; \sum_{\ell \in I'^{(i)} \setminus \{i\}} \overline{x}^{(i)}_{\ell,j},$$
where in the second last inequality we apply the fact that $r'^{(i)}_j \ge \alpha \cdot r^{(g)}_j$ holds for all $j \in D'^{(i)}$ by the design of the algorithm.
%
This proves the first part of this lemma.

\smallskip

For the second part, consider any $k \in I'^{(i)} \setminus \{i\}$.
By the conclusion of the first part, for any $j \in D'^{(i)}$, we have 
$$\sigma^{(i)}_{k,j} \; := \; \frac{\delta^{(i)}_j}{\sum_{ \ell \in I'^{(i)} \setminus \{i\}} \overline{x}^{(i)}_{\ell,j}} \cdot \overline{x}^{(i)}_{k,j} \; \le \; \overline{x}^{(i)}_{k,j},
\quad \text{and} \quad
\sigma^{(i)}_k \; := \; \frac{1}{\sum_{\ell \in D'^{(i)}} \overline{x}^{(i)}_{k,\ell}} \cdot \sum_{j \in D'^{(i)}} \sigma^{(i)}_{k,j} \; \le \; 1.$$
\end{proof}


%
Let $\vphantom{\text{\Large T}} \Psi''^{(i)} = \left( \; I''^{(i)}, D''^{(i)}, \bm{r}''^{(i)} \; \right)$ denote the updated tuple the algorithm maintains at the end of the $i^{th}$-iteration.
The following lemma, which is proved by verifying the corresponding constraints, shows that the feasible region of~\ref{LP_ITR} on $\Psi''^{(i)}$ remains nonempty.
This shows that the iterative rounding process is well-defined.

\begin{lemma} \label{lemma-cfl-approx-iterative-rounding-successive-tuple-feasibility}
For any $i \in F^* \cap I$, 
%
the solution $\left(\overline{\bm{x}}''^{(i)}, \overline{\bm{y}}''^{(i)}\right)$
defined by
\begin{align*}
\overline{x}''^{(i)}_{i,j} := & \left(1-\sigma^{(i)}_k\right) \cdot \overline{x}^{(i)}_{i,j}, \enskip\quad \text{for any $k \in I''^{(i)}, j \in D''^{(i)}$}, \\[6pt]
\overline{y}''^{(i)}_k := & \left(1-\sigma^{(i)}_k\right) \cdot \overline{y}^{(i)}_k \enskip\quad \text{for any $k \in I''^{(i)}$},
\end{align*}
is feasible for~\ref{LP_ITR} on the updated tuple $\Psi''^{(i)}$.
\end{lemma}

\begin{proof}
We prove this lemma by verifying the constraints in~\ref{LP_ITR}.
First, by Lemma~\ref{lemma-cfl-approx-iterative-rounding-well-defined}, 
we know that
$0 \le \sigma^{(i)}_k \le 1$ holds for all $k \in I''^{(i)}$, and $\left(\overline{\bm{x}}''^{(i)}, \overline{\bm{y}}''^{(i)}\right)$ is nonnegative.
%
%
\begin{itemize}
	\item
		Consider the constraint~(\ref{LP_ITR_1}) for any $j \in D''^{(i)}$ and observe that it is obtained by subtracting $\sum_{k \in I'^{(i)}} \sigma^{(i)}_k \cdot \overline{x}^{(i)}_{k,j}$ from both sides of the same constraint w.r.t. the original tuple $\Psi'^{(i)}$.
		Hence the constraint remains valid.
		
	\item
		Consider the constraint~(\ref{LP_ITR_2}) for any $k \in I''^{(i)}$.
		Observe that it is obtained by first multiplying the factor $( 1 - \sigma^{(i)}_k )$ to both sides of the same constraint w.r.t. $\Psi'^{(i)}$, followed by subtracting from the L.H.S. $\sum_{j \in D'^{(i)} \setminus D''^{(i)}} ( 1-\sigma^{(i)}_k ) \cdot \overline{x}^{(i)}_{i,j}$.
		Hence the resulting inequality remains valid.
		
		\smallskip

		For the constraint~(\ref{LP_ITR_3}) for any $k \in I''^{(i)}$ and $j \in D''^{(i)}$, observe that it is obtained by multiplying $( 1 - \sigma^{(i)}_k )$ to both sides of the same constraint w.r.t. $\Psi'^{(i)}$ and remains valid.

	\item
		For the constraint~(\ref{LP_ITR_4}), consider any $k \in I''^{(i)}$ and observe that
		$$\overline{y}''^{(i)}_k \; = \; \left( 1-\sigma^{(i)}_k \right) \cdot \overline{y}^{(i)}_k \; \le \; \overline{y}^{(i)}_k \; \le \; \frac{1-\alpha}{2}.$$
\end{itemize}
This proves the lemma.
\end{proof}

\medskip

Lemma~\ref{lemma-cfl-approx-iterative-rounding-initial-tuple-feasibility}, Lemma~\ref{lemma-cfl-approx-iterative-rounding-well-defined}, and Lemma~\ref{lemma-cfl-approx-iterative-rounding-successive-tuple-feasibility} prove the feasibility of the iterative rounding process.
%
Since exactly one facility is removed from $I'$ in each iteration, the process repeats for at most $|I|$ iterations before $I'$ becomes empty.
Moreover, since the tuple remains feasible during all the iterations, it follows from the constraint~(\ref{LP_ITR_3}) of~\ref{LP_ITR} that, when $I'$ becomes empty, $D'$ must also be empty.
%
Hence, the iterative process terminates in polynomial time.
%

%

%

%


\subsection{The assignment function $\bm{x}'$ to $U$}
\label{subsec-proof-approx-cfl-p1}

%
%
For notational brevity, let $D^{(\phi)}$ denote the set of clients that are fully-assigned by $\bm{h}$ and unreachable from any partially-assigned clients.
%
%
%
%
For any $i \in U$, $j \in \MC{D}$, define the assignment $x'_{i,j}$ as
$$x'_{i,j} \; := \; \begin{cases}
\; g_{i,j}, & \text{if $i \in U^{(\phi)}$,} \\[4pt]
\; h_{i,j}, & \text{if $i \in U^{(>)} \setminus U^{(\phi)}, \; j \in D^{(\phi)}$,} \\[4pt]
\; \frac{1}{1-\alpha} \cdot \sum_{p \in P(i,j)} \; f_p, & \text{if $i \in U^{(>)} \setminus U^{(\phi)}$, $j \in \MC{D} \setminus D^{(\phi)}$}, \\[4pt]
\; \frac{1}{1-\alpha} \cdot \sum_{p \in P(i,j)} \; f_p, & \text{if $i \in U^{(\le)}$}. \\[2pt]
\end{cases}
$$
%


\noindent
Intuitively, in the above definition,
we keep the partial assignments made in $\bm{g}$ and the flow sent in $\bm{f}$ to $U$, except for those originated from $D^{(\phi)}$ and those sent to sink in $U^{(>)} \setminus U^{(\phi)}$ via $D^{(\phi)}$.
%
%
%

\smallskip

The following proposition 
shows that, in the flow $\bm{f}$, for any $i \in \MC{F}$ and $j \in D^{(\phi)}$, the arc $(j^s,i)$ carries only flow originated from the commodity $j$.
%
%


\begin{prop} \label{prop-MFN-flow-client-without-partial-assignment}
For any $j \in D^{(\phi)}$, $i \in \MC{F}$, and any $p \in P$ with $(j^s,i) \in p$, we have $f_p > 0$ implies that $p \in \bigcup_{i' \in \MC{F}} P(i',j)$, that is, $f_p > 0$ only when $p$ is a path for commodity $j$.
%
\end{prop}

\begin{proof}
By the definition of $\bm{g}$, we have $g_{k,j} = 0$ for all $k \in \MC{F}$ since $j \in D^{(\ell)}$.
%
Hence, $p$ must start from $j^s$ and is thereby a path for commodity $j$.
%
\end{proof}

\smallskip

The following lemma
shows that the facilities in $U \setminus U^{(\phi)}$ are reasonably sparsely-loaded in that, extra amount of assignments can be accommodated when the assignments from $\MC{D} \setminus D^{(\phi)}$ are scaled up by a factor of $1/(1-\alpha)$.
%

\begin{lemma} \label{lemma-approx-cfl-U-sparsely-loaded}
For any $i \in U \setminus U^{(\phi)}$, we have
%
$\sum_{j \in \MC{D}} \; x'_{i,j} \; \le \; u_i.$
%
\end{lemma}

\begin{proof}
%
%
%
Consider the category to which $i$ belongs.
We have the following two cases.
%
\begin{itemize}
	\item
		$i \in U^{(>)} \setminus U^{(\phi)}$.
		%
		%
		
		Since $i$ is not tightly-occupied, it follows by the optimality of $\bm{h}$ that, 
		\begin{equation}
		h_{i,j} \; = \; \frac{1}{ 1-\alpha} \cdot x_{i,j} \quad \text{holds for any $j \in \MC{D} \setminus D^{(\phi)}$,}
		\label{ieq-approx-cfl-U-sparsely-loaded-1}
		\end{equation}
		since any $j \in \MC{D} \setminus D^{(\phi)}$ is either partially-assigned, or fully-assigned but reachable from a partially-assigned client.
		By the construction of $\bm{x}'$, we have
		\begin{align*}
		\sum_{j \in \MC{D}} \; x'_{i,j} \;\; 
		= & \;\; \sum_{j \in D^{(\phi)}} h_{i,j} \; + \; \frac{1}{1-\alpha} \cdot \sum_{j \in \MC{D} \setminus D^{(\phi)}, \; p \in P(i,j)} f_p.
		\end{align*}
		By Proposition~\ref{prop-MFN-flow-client-without-partial-assignment}, flow originated from $\MC{D} \setminus D^{(\phi)}$ must sink via some client in $\MC{D} \setminus D^{(\phi)}$.
		Hence, by the feasibility of $\bm{f}$ for $\mathbf{MFN}_\Psi(\bm{x},\overline{\bm{y}},\bm{g})$, we have
		\begin{align*}
		\sum_{j \in \MC{D} \setminus D^{(\phi)}, \; p \in P(i,j)} f_p \;\; 
		\le \;\; \sum_{j \in \MC{D}, \; p \in P, \; (j^s,i) \in p} f_p \;\; 
		\le \;\; \sum_{j \in \MC{D} \setminus D^{(\phi)}} x_{i,j}.
		\end{align*}
		%
		%
		Combining the above with~(\ref{ieq-approx-cfl-U-sparsely-loaded-1}), we obtain
		\begin{align*}
		\sum_{j \in \MC{D}} \; x'_{i,j} \;\; 
		\le & \;\; \sum_{j \in D^{(\phi)}} h_{i,j} \; + \; \frac{1}{1-\alpha} \sum_{j \in \MC{D} \setminus D^{(\phi)}} x_{i,j}  \;\;
		= \;\; \sum_{j \in \MC{D}} \; h_{i,j} \;\; \le \;\; u_i,
		\end{align*}
		where the last inequality follows from the feasibility of $\bm{h}$.
		%
		%
				
	\item
		$i \in U^{(\le)}$.

		By the construction of $\bm{x}'$ and the feasibility of $\bm{f}$ for $\mathbf{MFN}_\Psi(\bm{x},\overline{\bm{y}},\bm{g})$, we have
		\begin{align*}
		\sum_{j \in \MC{D}} \; x'_{i,j} \;\; = \; \frac{1}{1-\alpha} \cdot \sum_{j \in \MC{D}, \; p \in P(i,j)} f_p \;\; \le \;\; \frac{1}{1-\alpha} \cdot \sum_{j \in \MC{D}} \; x_{i,j} \;\; \le \;\; u_i,
		\end{align*}
		where the last inequality follows from the definition of $U^{(\le)}$.
\end{itemize}
This completes the proof of this lemma.
\end{proof}

%

%

%
%


\subsection{The rounded assignment $\bm{x}''$ to $F^* \cap I$}
\label{subsec-proof-approx-cfl-p2}

Recall that, for any $i \in F^* \cap I$,
we use $\vphantom{\text{\Large T}} \Psi^{(i)} = \left( \; I'^{(i)}, D'^{(i)}, \bm{r}'^{(i)} \; \right)$ to denote the parameter tuple the algorithm maintains in the beginning of the $i^{th}$-iteration and
%
%
$\left(\overline{\bm{x}}^{(i)}, \overline{\bm{y}}^{(i)} \right)$ to denote the optimal solution computed for~\ref{LP_ITR} on $\Psi^{(i)}$.

\smallskip



%
For any $i \in F^* \cap I$ and $k \in I'^{(i)}$, define the reassignment function from $k$ to $i$ as $$z''_{i,k} \; := \; \sum_{j \in D'^{(i)}} \sigma^{(i)}_{k,j} \; = \; \sigma^{(i)}_k \cdot \sum_{\ell \in D'^{(i)}} \overline{x}^{(i)}_{k,\ell}.$$
%
Intuitively, $z''_{i,k}$ is the amount of assignment that gets reassigned from $k$ to $i$ in the $i^{th}$-iteration of the rounding process.
%
%
%
For any $i \in F^* \cap I$, $j \in \MC{D}$, define the rounded assignment function
%
$$x''_{i,j} \;\; := \;\; \frac{1}{1-\alpha} \cdot \sum_{k \in I'^{(i)}} \; \frac{\overline{x}^{(i)}_{k,j}}{\sum_{\ell \in D'^{(i)}} \overline{x}^{(i)}_{k,\ell}} \cdot z''_{i,k}.$$
We note that, in the definition we scale up the reassignment by $1/(1-\alpha)$. 
Furthermore, by definition $x''_{i,j} > 0$ only when $j \in D'^{(i)}$.
%
%
%
%
The following lemma shows that, for each $j \in \MC{D}$, the demand originally assigned to $I$ is reassigned by $\bm{x}''$ to facilities in $F^* \cap I$.
%

\begin{lemma} \label{lemma-approx-cfl-I-mostly-assigned}
For any $j \in \MC{D}$, we have 
$$\sum_{i \in F^* \cap I} x''_{i,j} \;\; \ge \;\; \frac{1}{1-\alpha} \cdot \left( \; r'^{(0)}_j \; - \; \alpha \cdot r^{(g)}_j \; \right).$$
\end{lemma}

\begin{proof}
By the way $\bm{z}''$ is defined, all the assignment reduced in the process due to the scale-down operation is reassigned by $\bm{z}''$ to facilities in $F^* \cap I$.
For any $j \in \MC{D}$, we have
\begin{align}
(1-\alpha) \cdot \sum_{i \in F^* \cap I} x''_{i,j} \;\; 
= & \;\; \sum_{i \in F^* \cap I} \; \sum_{k \in I'^{(i)}} \; \frac{ \overline{x}^{(i)}_{k,j} }{ \sum_{\ell \in D'^{(i)}} \overline{x}^{(i)}_{k,\ell} } \cdot z''_{i,k} \;\; 
= \;\; \sum_{i \in F^* \cap I} \; \sum_{k \in I'^{(i)}} \; \sigma^{(i)}_k \cdot \overline{x}^{(i)}_{k,j}, \notag
\end{align}
where in the last equality we apply the definition of $\bm{z}$.

\smallskip
%

Since the algorithm removes a client $j$ from $D'$ only when its residue assignment to $I'$ drops below $\alpha \cdot r^{(g)}_j$, and since the algorithm repeats until $D'$ becomes empty,
it follows that 
$$\sum_{i \in F^* \cap I} \sum_{k \in I'^{(i)}} \sigma^{(i)}_k \cdot \overline{x}^{(i)}_{k,j} \; \ge \; r'^{(0)}_j - \alpha \cdot r^{(g)}_j.$$
This proves the lemma.
\end{proof}

The following lemma shows that, the facilities in $F^* \cap I$ 
can accommodate the rounded assignments given by $\bm{x}''$.

\begin{lemma} \label{lemma-approx-cfl-I-sparsely-loaded}
For any $i \in F^* \cap I$, we have $\sum_{j \in \MC{D}} \; x''_{i,j} \; \le \; u_i$.
%
%
\end{lemma}

\begin{proof}
By the definition of $\bm{x}''$ and $\bm{z}''$,
we have
\begin{align}
(1-\alpha) \cdot \sum_{j \in \MC{D}} x''_{i,j} \;\; 
= & \;\; (1-\alpha) \cdot \sum_{j \in D'^{(i)}} x''_{i,j} \;\; 
= \;\; \sum_{j \in D'^{(i)}} \; \sum_{k \in I'^{(i)}} \; \frac{ \overline{x}^{(i)}_{k,j} }{ \sum_{\ell \in D'^{(i)}} \overline{x}^{(i)}_{k,\ell} } \cdot z''_{i,k} \;\;  \notag \\[8pt]
= & \;\; \sum_{j \in D'^{(i)}} \; \sum_{k \in I'^{(i)}} \; \frac{ \overline{x}^{(i)}_{k,j} }{ \sum_{\ell \in D'^{(i)}} \overline{x}^{(i)}_{k,\ell} } \cdot \sum_{\ell \in D'^{(i)}} \sigma^{(i)}_{k,\ell}. \;\;  \notag
\end{align}
Apply the definition of $\sigma^{(i)}_{k,\ell}$ for $i$ and any $k \in I'^{(i)} \setminus \{i\}$ and we have
\begin{align}
(1-\alpha) \cdot \sum_{j \in \MC{D}} x''_{i,j} \;\; 
= & \;\; \sum_{j \in D'^{(i)}} \; \frac{ \overline{x}^{(i)}_{i,j} }{ \sum_{\ell \in D'^{(i)}} \overline{x}^{(i)}_{i,\ell} } \cdot \sum_{\ell \in D'^{(i)}} \overline{x}^{(i)}_{i,\ell} \;\;  \notag \\[8pt]
& \qquad + \; \sum_{j \in D'^{(i)}} \; \sum_{k \in I'^{(i)} \setminus \{i\} } \; \frac{ \overline{x}^{(i)}_{k,j} }{ \sum_{\ell \in D'^{(i)}} \overline{x}^{(i)}_{k,\ell} } \cdot \sum_{\ell \in D'^{(i)}} \; \frac{ \overline{x}^{(i)}_{k,\ell} }{ \sum_{ m\in I'^{(i)} \setminus \{i\}} \overline{x}^{(i)}_{m,\ell} } \cdot \delta^{(i)}_\ell \;\;  \notag \\[8pt]
= & \;\; \sum_{j \in D'^{(i)}} \overline{x}^{(i)}_{i,j} \; + \; \sum_{k \in F'^{(i)} \setminus \{i\}} \; \sum_{\ell \in D'^{(i)}} \; \frac{ \overline{x}^{(i)}_{k,\ell} }{ \sum_{ m\in I'^{(i)} \setminus \{i\}} \overline{x}^{(i)}_{m,\ell} } \cdot \delta^{(i)}_\ell \;\; \notag \\[10pt]
= & \;\; \sum_{j \in D'^{(i)}} \overline{x}^{(i)}_{i,j} \; + \; \sum_{\ell \in D'^{(i)}} \delta^{(i)}_{\ell}. \;\; \notag
\end{align}
Further applying the definition of $\delta^{(i)}_\ell$ for any $\ell \in D'^{(i)}$, we obtain
\begin{align}
\sum_{j \in \MC{D}} x''_{i,j} \;\; 
= & \;
\; \frac{\alpha}{1-\alpha} \cdot \frac{1}{\overline{y}^{(i)}_i} \cdot \sum_{j \in D'^{(i)}} \overline{x}^{(i)}_{i,j} \;\; 
\le \;\; u_i,  \notag
\end{align}
where the last inequality follows from~(\ref{LP_ITR_2}) and the fact that $\alpha \le 1/2$.
\end{proof}

%

%
%


\subsection{The overall assignment $\bm{x}'''$ and the feasibility}
\label{subsec-proof-approx-cfl-feasibility}

%
%
\noindent
For any $i \in \MC{F}^*$ and $j \in \MC{D}$, define the overall assignment $x'''_{i,j}$ as
$$x'''_{i,j} \;\; := \;\; 
\begin{cases}
\;\; x'_{i,j}, & \text{if $i \in U$,} \\[4pt]
\;\; x''_{i,j}, & \text{if $i \in F^* \cap I$,} \\[4pt]
\;\; 0, & \text{otherwise.} \\[2pt]
\end{cases}$$
%
%
%
%
The following lemma shows that the feasible region of the min-cost assignment problem on $\MC{D}$ and $\MC{F}^*$ is nonempty, and proves the feasibility of our rounding algorithm.
%

\begin{lemma} \label{lemma-approx-cfl-flow-feasibility}
$\left( \bm{x}''', \bm{y}^* \right)$ is feasible for $\Psi$, i.e., 
\begin{align*}
\forall j \in \MC{D} \colon \;\; \sum_{i \in F^*} \; x'''_{i,j} \; \ge \; 1,  
\quad \text{and} \quad
\forall i \in F^* \colon \;\; \sum_{j \in \MC{D}} \; x'''_{i,j} \; \le \; u_i.
\end{align*}
\end{lemma}

\begin{proof}
Depending on the category to which $j$ belongs,
consider the following two cases.
\begin{itemize}
	\item
		$j \in D^{(\phi)}$, i.e., 
		$j$ is fully-assigned by $\bm{h}$ to facilities in $U^{(>)}$ and unreachable from any partially-assigned client via augmenting paths.
		%
		%
		It follows that $h_{i,j} = 0$ for all $i \in U^{(\phi)}$, and
		$$\sum_{i \in F^*} x'''_{i,j} \;\;  \ge \;\; \sum_{i \in U^{(>)} \setminus U^{(\phi)}} x'_{i,j} \;\; = \;\; \sum_{i \in U^{(>)} \setminus U^{(\phi)}} h_{i,j} \; = \; 1.$$
		%
		%
		
		%
		
		
	\item
		$j \in \MC{D} \setminus D^{(\phi)}$, i.e., $j$ is either partially-assigned, or fully-assigned but reachable from partially-assigned clients via augmenting paths.
		%
		By the definition of $\bm{x}'''$, we have
		\begin{align*}
		\sum_{i \in F^*} x'''_{i,j} \;\; 
		= & \;\; \sum_{i \in U^{(\phi)}} x'_{i,j} \; + \; \sum_{i \in U \setminus U^{(\phi)}} x'_{i,j} \; + \; \sum_{i \in F^* \cap I} x''_{i,j}  \\[8pt]
		= & \;\; \sum_{i \in U^{(\phi)}} g_{i,j} \; + \; \frac{1}{1-\alpha} \cdot \sum_{i \in U \setminus U^{(\phi)}, \; p \in P(i,j)} f_p \; + \; \sum_{i \in F^* \cap I} x''_{i,j}.
		\end{align*}
		Applying Lemma~\ref{lemma-approx-cfl-I-mostly-assigned} and the definition of $r'^{(0)}_j$, the above becomes
		\begin{align*}
		\sum_{i \in F^*} x'''_{i,j} \;\; 
		\ge & \;\; \sum_{i \in U^{(\phi)}} g_{i,j} \; + \; \frac{1}{1-\alpha} \cdot \left( \; \sum_{i \in U \setminus U^{(\phi)}, \; p \in P(i,j)} f_p \; + \; \sum_{i \in I, \; p \in P(i,j)} f_p \; - \; \alpha \cdot r^{(g)}_j \; \right).
		\end{align*}
		Since $\sum_{j \in \MC{D}} g_{i,j} = u_i$ for all $i \in U^{(\phi)}$, it follows that $\sum_{i \in U^{(\phi)}, \; p \in P(i,j)} f_p = 0$ and 
		\begin{align*}
		\sum_{i \in U \setminus U^{(\phi)}, \; p \in P(i,j)} f_p \; + \; \sum_{i \in I, \; p \in P(i,j)} f_p \; - \; \alpha \cdot r^{(g)}_j \;\;
		= \sum_{i \in \MC{F}, \; p \in P(i,j)} f_p \; - \; \alpha \cdot r^{(g)}_j \; \ge \; (1-\alpha) \cdot r^{(g)}_j
		\end{align*}
		by the constraint~(\ref{LP_MFN_1}) in $\mathbf{MFN}_\Psi(\bm{x},\bm{y}',\bm{g})$.
		%
		Hence, we obtain
		\begin{align*}
		\sum_{i \in F^*} x'''_{i,j} \;\; 
		\ge \; \sum_{i \in U^{(\phi)}} g_{i,j} \; + \; r^{(g)}_j \; = \; 1.
		\end{align*}

\end{itemize}
This proves the first half of this lemma.
For the second half, by Lemma~\ref{lemma-approx-cfl-U-sparsely-loaded} and Lemma~\ref{lemma-approx-cfl-I-sparsely-loaded}, it remains to consider 
the case for which $i \in U^{(\phi)}$.
In this case, $i$ is fully-matched by $\bm{h}$, and we have $\sum_{j \in \MC{D}} x'''_{i,j} = \sum_{j \in \MC{D}} \; g_{i,j} = u_i.$
This proves the second half of this lemma.
\end{proof}

%

%

%

%


%
\begin{figure*}[h]
\centering
\fbox{
\hspace{-10pt}
\begin{minipage}{.8\textwidth}
\begin{subequations}
\begin{align}
\text{max} \;\; & \;\; \sum_{j \in D'} r'_j \cdot \lambda_j \; - \; \sum_{i \in F'} \; \frac{1}{2} \cdot (1-\alpha) \cdot \eta_i & & \label{LP_ITR_dual} \tag*{LP-(DM)} \\[6pt]
\text{s.t.} \;\; & \;\; \lambda_j \; \le \; \beta_i \; + \; \Gamma_{i,j} \; + \; c_{i,j}, & & \forall i \in F', j\in D', \label{LP_ITR_dual_1} \\[6pt]
& \;\; u_i \cdot \beta_i \; + \; \frac{2\alpha}{1-\alpha} \cdot \sum_{j \in D'} r_j \cdot \Gamma_{i,j} \; \le \; o_i \; + \; \eta_i, & & \forall i \in F', \label{LP_ITR_dual_2} \\[4pt]
& \;\; \lambda_j, \; \beta_i, \; \Gamma_{i,j}, \; \eta_i \; \ge \; 0, & & \forall i \in F', j \in D'. \notag 
\end{align}
\end{subequations}
\vspace{-8pt}
\end{minipage}
\enskip
}
\vspace{-4pt}
\end{figure*}


\subsection{The cost incurred by $F^*\cap I$}
\label{subsec-cfl-cost-second-stage}

Recall that $(\overline{\bm{x}}^{(0)}, \overline{\bm{y}}^{(0)})$ is the solution defined in Section~\ref{subsec-proof-cfl-approx-algo-well-defined} for the initial tuple $\Psi^{(0)}$.
In this section, we prove that 
$$\psi\left( \bm{x}'',\bm{y}^*|_{F^* \cap I} \right) \; \le \; \frac{3}{1-\alpha} \cdot \psi\left(\overline{\bm{x}}^{(0)},\overline{\bm{y}}^{(0)} \right).$$
%

\smallskip

%
For any $i \in F^* \cap I$, consider the dual LP of~\ref{LP_ITR} on $\Psi^{(i)}$, which is given below as~\ref{LP_ITR_dual}.
Let $(\bm{\lambda}^{(i)},\bm{\beta}^{(i)},\bm{\Gamma}^{(i)},\bm{\eta}^{(i)})$ be an optimal solution for~\ref{LP_ITR_dual}.
It follows that $\overline{x}^{(i)}_{i,j} > 0$ implies that $\lambda^{(i)}_j \ge c_{i,j}$.
The following lemma
establishes the equivalence between the cost incurred by a non-extremal facility and the amount of dual values it receives.
%

\begin{lemma} \label{lemma-cfl-approx-primal-dual-for-I}
For any $i \in F^* \cap I$ and any $k \in I'^{(i)}$ with $0 < \overline{y}^{(i)}_k < (1-\alpha)/2$, we have
$$\sum_{j \in D'^{(i)}} \lambda^{(i)}_j \cdot \overline{x}^{(i)}_{k,j} \; = \; o_k \cdot \overline{y}^{(i)}_k \; + \; \sum_{j \in D'^{(i)}} c_{k,j} \cdot \overline{x}^{(i)}_{k,j}. $$
\end{lemma}

\begin{proof}
This lemma follows from the complementary slackness conditions between $(\overline{\bm{x}}^{(i)}, \overline{\bm{y}}^{(i)})$ and $(\bm{\lambda}^{(i)}, \bm{\beta}^{(i)}, \bm{\Gamma}^{(i)}, \bm{\eta}^{(i)})$.
Since $\overline{x}^{(i)}_{k,j} > 0$ implies that the corresponding dual inequality~(\ref{LP_ITR_dual_1}) must hold with equality,
we have
$$\sum_{j \in D'^{(i)}} \; \lambda^{(i)}_j \cdot \overline{x}^{(i)}_{k,j} \;\; = \;\; \beta^{(i)}_k \cdot \sum_{j\in D'^{(i)}} \overline{x}^{(i)}_{k,j} \; + \; \sum_{j \in D'^{(i)}} \Gamma^{(i)}_{k,j} \cdot \overline{x}^{(i)}_{k,j} \; + \; \sum_{j \in D'^{(i)}} c_{k,j} \cdot \overline{x}^{(i)}_{k,j}.$$
Similarly, $\beta^{(i)}_k > 0$ and $\Gamma^{(i)}_{k,j} > 0$ imply that the corresponding inequalities~(\ref{LP_ITR_2}) and~(\ref{LP_ITR_3}) hold with equality.
Hence the above equality becomes
$$\beta^{(i)}_k \cdot u_k \cdot \overline{y}^{(i)}_k \; + \; \frac{2\alpha}{1-\alpha} \sum_{j \in D'^{(i)}} r^{(g)}_j \cdot \overline{y}^{(i)}_k \cdot\Gamma^{(i)}_{k,j} \; + \; \sum_{j \in D'^{(i)}} c_{k,j} \cdot \overline{x}^{(i)}_{k,j}.$$
The assumption that $\overline{y}^{(i)}_k > 0$ and $\overline{y}^{(i)}_k < (1-\alpha)/2$ imply that the corresponding inequality~(\ref{LP_ITR_dual_2}) holds with equality and the dual variable $\eta^{(i)}_k$ must be zero.
The above equality becomes
$$\overline{y}^{(i)}_k \cdot \left( \; u_k \cdot \beta^{(i)}_k \; + \; \frac{2\alpha}{1-\alpha} \sum_{j \in D'^{(i)}} r^{(g)}_j \cdot \Gamma^{(i)}_{k,j} \; \right) \; + \; \sum_{j \in D'^{(i)}} c_{k,j} \cdot \overline{x}^{(i)}_{k,j} \; = \; o_k \cdot \overline{y}^{(i)}_k \; + \; \sum_{j\in D'^{(i)}} c_{k,j} \cdot \overline{x}^{(i)}_{k,j}, $$
and this lemma is proved.
\end{proof}

The following lemma bounds the facility cost plus the total rerouting cost incurred by each individual cluster $i \in F^* \cap I$.
%

%

\begin{lemma} \label{lemma-cfl-cost-aggregation-per-iteration}
For any $i \in F^* \cap I$, we have 
\begin{align*}
(1-\alpha) \cdot o_i \; + \; \sum_{k \in I'^{(i)}} c_{i,k} \cdot z''_{i,k} \;\; \le \;\; \sum_{k \in I'^{(i)}} \; \sigma^{(i)}_k \cdot \left( \; 3\cdot o_k \cdot \overline{y}^{(i)}_k \; + \; 2 \cdot \sum_{j \in D'^{(i)}} c_{k,j} \cdot \overline{x}^{(i)}_{k,j} \; \right).
\end{align*}
%
\end{lemma}

\begin{proof}
If $\overline{y}^{(i)}_i = (1-\alpha)/2$, then the statement of this lemma holds trivially.
In the following we assume the nontrivial case for which $0 < \overline{y}^{(i)}_i < (1-\alpha)/2$.

\smallskip

By the definition of $\bm{z}''$ and the triangle inequality, we have
\begin{align}
\sum_{k \in I'^{(i)}} c_{i,k} \cdot z''_{i,k} \;\; 
= \;\; \sum_{k \in I'^{(i)}} c_{i,k} \cdot \sum_{j \in D'^{(i)}} \sigma^{(i)}_{k,j} \;\;
\le \; \sum_{j \in D'^{(i)}, \; k \in I'^{(i)}} \left( \; c_{i,j} \; + \; c_{k,j} \; \right) \cdot \sigma^{(i)}_{k,j}.
\label{ieq-cfl-cost-aggretation-per-iteration-1}
\end{align}
By the definition of $\sigma^{(i)}_{k,j}$, we know that $\sigma^{(i)}_{k,j} > 0$ implies that $\overline{x}^{(i)}_{k,j} > 0$, which further implies that $c_{k,j} \le \lambda^{(i)}_j$.
Hence, Inequality~(\ref{ieq-cfl-cost-aggretation-per-iteration-1}) becomes
\begin{align}
\sum_{k \in I'^{(i)}} c_{i,k} \cdot z''_{i,k} \;\; 
& \le \; \sum_{j \in D'^{(i)}, \; k \in I'^{(i)}} \left( \; c_{i,j} \; + \; \lambda^{(j)}_j \; \right) \cdot \frac{\overline{x}^{(i)}_{k,j}}{\sum_{\ell \in I'^{(i)}\setminus \{i\}}\overline{x}^{(i)}_{\ell,j}} \cdot \delta^{(i)}_j \notag \\[8pt]
& = \; \sum_{j \in D'^{(i)}} \left( \; c_{i,j} \; + \; \lambda^{(j)}_j \; \right) \cdot \delta^{(i)}_j \;\;
\le \; \frac{1-\alpha}{2} \cdot \frac{1}{\overline{y}^{(i)}_i} \cdot \sum_{j \in D'^{(i)}} \left( \; c_{i,j} \; + \; \lambda^{(j)}_j \; \right) \cdot \overline{x}^{(i)}_{i,j},
\label{ieq-cfl-cost-aggretation-per-iteration-2}
\end{align}
where in the last inequality we apply the definition of $\delta^{(i)}_j$.

\smallskip

By Lemma~\ref{lemma-cfl-approx-primal-dual-for-I} and Inequality~(\ref{ieq-cfl-cost-aggretation-per-iteration-2}), we have
\begin{align}
(1-\alpha) \cdot o_i \; + \sum_{k \in I'^{(i)}} c_{i,k} \cdot z''_{i,k} \;  
& = \;\; \frac{1-\alpha}{2} \cdot \frac{1}{\overline{y}^{(i)}_i} \cdot \left( \; 3 \cdot o_i \cdot \overline{y}^{(i)}_i \; + 2\cdot \sum_{j \in D'^{(i)}} c_{i,j} \cdot \overline{x}^{(i)}_{i,j} \; \right)  \notag \\[6pt]
& = \;\; \frac{1-\alpha}{2} \cdot \frac{1}{\overline{y}^{(i)}_i} \cdot \theta^{(i)} (i) \cdot \sum_{j \in D'^{(i)}} \overline{x}^{(i)}_{i,j},
\label{ieq-cfl-cost-aggregation-per-iteration-4}
\end{align}
where $\theta^{(i)}$ refers to the $\theta$ function defined in the $i^{th}$-iteration.
%
%
%
%
By the definitions of $\sigma^{(i)}_k$ and $\sigma^{(i)}_{k,j}$ for any $k \in I'^{(i)} \setminus \{i\}$ and $j \in D'^{(i)}$, we have
\begin{align}
\sum_{\ell \in D'^{(i)}} \sigma^{(i)}_k \cdot \overline{x}^{(i)}_{k,\ell} \; = \; \sum_{\ell \in D'^{(i)}} \sigma^{(i)}_{k,\ell} \quad \text{and} \quad \sum_{\ell \in I' \setminus\{i\}} \sigma^{(i)}_{\ell,j} \; = \; \delta^{(i)}_j.
\notag
\end{align}
Hence, we have
\begin{align}
\frac{1-\alpha}{2} \frac{1}{y^{\dagger(i)}_i} \cdot \sum_{j \in D'^{(i)}} x^{\dagger(i)}_{i,j} \; 
& = \sum_{j \in D'^{(i)}} \left( \; \delta^{(i)}_j \; + \; \sigma^{(i)}_{i,j} \; \right) \;
= \sum_{j \in D'^{(i)}, \; k \in I'^{(i)}} \sigma^{(i)}_{k,j} \;\;
= \sum_{k \in I'^{(i)}, \; j \in D'^{(i)}} \sigma^{(i)}_k \cdot \overline{x}^{(i)}_{k,j}.
\label{ieq-cfl-cost-aggregation-per-iteration-5}
\end{align}
Moreover, by the design of the algorithm, we have $\theta^{(i)}(i) \; \le \; \theta^{(i)}(k)$ for any $k \in I'^{(i)}$.
Combining this property with Inequalities~(\ref{ieq-cfl-cost-aggregation-per-iteration-4}) and~(\ref{ieq-cfl-cost-aggregation-per-iteration-5}), we obtain
\begin{align}
(1-\alpha) \cdot o_i \; + \sum_{k \in I'^{(i)}} c_{i,k} \cdot z''_{i,k} \; 
& \le \;\; \sum_{k \in I'^{(i)}} \sigma^{(i)}_k \cdot \theta^{(i)}(k) \cdot \sum_{j \in D'^{(i)}} \overline{x}^{(i)}_{k,j}  \notag \\[6pt]
& = \;\; \sum_{k \in I'^{(i)}} \; \sigma^{(i)}_k \cdot \left( \;  3\cdot o_k \cdot \overline{y}^{(i)}_k \; + \; 2\cdot \sum_{j \in D'^{(i)}} c_{k,j} \cdot \overline{x}^{(i)}_{k,j} \; \right), 
\label{ieq-cfl-cost-aggregation-per-iteration-6}
\end{align}
by applying the definition of $\theta^{(i)}(k)$ 
for all $k \in I'^{(i)}$ with $\sum_{j \in D'^{(i)}} \overline{x}^{(i)}_{k,j} > 0$.
%
%
%
\end{proof}

%


%
By Lemma~\ref{lemma-cfl-approx-iterative-rounding-initial-tuple-feasibility}, Lemma~\ref{lemma-cfl-approx-iterative-rounding-successive-tuple-feasibility}, and Lemma~\ref{lemma-cfl-cost-aggregation-per-iteration}, we obtain the following lemma which establishes the overall guarantee for our iterative rounding process.

\begin{lemma} \label{lemma-cfl-cost-aggregation-overall}
We have
\begin{align}
\sum_{i \in F^* \cap I} \left( \; o_i \; + \; \sum_{j \in D'^{(i)}} c_{i,j} \cdot x''_{i,j} \;\; \right)  
\; \le \;\; \frac{3}{1-\alpha} \cdot \left( \; \sum_{i \in I} \; o_i \cdot \overline{y}^{(0)}_i \; + \; \sum_{i \in I, \; j \in \MC{D}} c_{i,j} \cdot \overline{x}^{(0)}_{i,j} \; \right). 
\label{ieq-cfl-cost-aggregation-overall-target} 
\end{align}
\end{lemma}

\begin{proof}
Recall that, for any $i \in F^* \cap I$, 
%
we use $\Psi''^{(i)} = (I''^{(i)}, D''^{(i)},\bm{r''}^{(i)})$ to denote the updated parameter tuple the algorithm maintains at the end of the $i^{th}$-iteration.
%
Let $( \overline{\bm{x}}''^{(i)}, \overline{\bm{y}}''^{(i)} )$ be an optimal solution for~\ref{LP_ITR} on $\Psi''^{(i)}$.
%
%
%
%
By Lemma~\ref{lemma-cfl-approx-iterative-rounding-successive-tuple-feasibility}, we have
\begin{align*}
& \sum_{k \in I''^{(i)}} o_k \cdot \overline{y}''^{(i)}_k + \sum_{ \substack{ k \in I''^{(i)}, \\[2pt]  j \in D''^{(i)} } } c_{k,j} \cdot \overline{x}''^{(i)}_{k,j} \;
\le \; \sum_{k \in I''^{(i)}} \left( 1-\sigma^{(i)}_k \right) \cdot \left( o_k \cdot \overline{y}^{(i)}_k + \sum_{j \in D''^{(i)}} c_{k,j} \cdot \overline{x}^{(i)}_{k,j} \right).
\end{align*}
%
%
%
Combining the above with Lemma~\ref{lemma-cfl-cost-aggregation-per-iteration} and apply the fact that $I''^{(i)} = I'^{(i)} \setminus \{i\}$,
we obtain
\begin{align}
& \left( \; (1-\alpha) \cdot o_i \; + \; \sum_{j \in D'^{(i)}} c_{i,j} \cdot x''_{i,j} \; \right) \; + \; 3\cdot \left( \; \sum_{i \in I''^{(i)}} o_i \cdot \overline{y}''^{(i)}_i \; + \sum_{ i \in I''^{(i)}, \; j \in D''^{(i)} } c_{i,j} \cdot \overline{x}''^{(i)}_{i,j} \; \right)  \notag \\[6pt]
& \hspace{2.2cm} \le \;\; 3 \cdot \sum_{k \in I'^{(i)}} \; \left( \; o_k \cdot \overline{y}^{(i)}_k \; + \; \sum_{j \in D'^{(i)}} c_{k,j} \cdot \overline{x}^{(i)}_{k,j} \; \right).
\label{ieq-cfl-cost-aggregation-overall-2}
\end{align}
Inequality~(\ref{ieq-cfl-cost-aggregation-overall-2}) shows that, the total cost incurred by $i$ can be bounded within three times the difference between the optimal values of the successive iterations.
%
%
%
%
Taking the summation over $i \in F^* \cap I$ and applying 
Inequality~(\ref{ieq-cfl-cost-aggregation-overall-2}), we obtain
\begin{align*}
\sum_{i \in F^* \cap I} \left( \; (1-\alpha) \cdot o_i \; + \; \sum_{j \in D'^{(i)}} c_{i,j} \cdot x''_{i,j} \;\; \right) \;\; 
& \le \;\; 3\cdot \left( \; \sum_{i \in I} o_i \cdot \overline{y}''^{(0)}_i \; + \; \sum_{i \in I, j \in \MC{D}} c_{i,j} \cdot \overline{x}''^{(0)}_{i,j} \; \right) \\[6pt]
& \le \;\; 3\cdot \left( \; \sum_{i \in I} o_i \cdot \overline{y}^{(0)}_i \; + \; \sum_{i \in I, j \in \MC{D}} c_{i,j} \cdot \overline{x}^{(0)}_{i,j} \; \right),
\end{align*}
where we use $(\overline{\bm{x}}''^{(0)}, \overline{\bm{y}}''^{(0)})$ to denote an optimal solution for~\ref{LP_ITR} on the initial parameter tuple $\Psi^{(0)}$ and the fact from Lemma~\ref{lemma-cfl-approx-iterative-rounding-initial-tuple-feasibility} that $(\overline{\bm{x}}^{(0)}, \overline{\bm{y}}^{(0)})$ is a feasible solution for~\ref{LP_ITR} on $\Psi^{(0)}$.
Multiplying the above by $1/(1-\alpha)$ completes the proof of this lemma.
\end{proof}

%


\subsection{The overall guarantee}
\label{subsec-proof-approx-cfl-approx-guarantee}

In this section we establish the guarantee for the solution $\psi\left( \bm{x}''', \bm{y}^* \right)$.
Recall that $(\bm{x}, \bm{y})$ is the initial candidate solution we have for $\mathbf{MFN}_\Psi(\bm{x},\bm{y},\bm{g})$ and $(\overline{\bm{x}}^{(0)},\overline{\bm{y}}^{(0)})$ is the solution defined in Lemma~\ref{lemma-cfl-approx-iterative-rounding-initial-tuple-feasibility} for~\ref{LP_ITR} on the initial tuple $\Psi^{(0)}$.
%
%
%
Apply Lemma~\ref{lemma-cfl-cost-aggregation-overall} and the definition of $\overline{\bm{y}}^{(0)}$,
we have
\begin{align}
\psi\left( \; \bm{x}''', \; \bm{y}^* \; \right) \;\; 
& = \;\; \sum_{i \in F^*} \; o_i \; + \; \sum_{i \in \MC{F}^*, \; j \in \MC{D}} c_{i,j} \cdot x'''_{i,j}  \notag \\[8pt]
& \hspace{-1cm} \le \;\; \frac{1}{\alpha} \cdot \sum_{i \in U} \; o_i \cdot y_i \; + \; \frac{3}{2\alpha} \cdot \sum_{i \in I} o_i \cdot y'_i 
\;\; + \sum_{i \in U} c_{i,j} \cdot x'_{i,j}
\; + \frac{3}{1-\alpha} \cdot \sum_{i \in I, \; j \in \MC{D}} c_{i,j} \cdot \overline{x}^{(0)}_{i,j}.
\label{ieq-cfl-overall-bound-1}
\end{align}

\smallskip

\noindent
The following lemma, which is proved by considering the contribution of each individual edge in the flow paths, bounds the assignment cost in the R.H.S. of~(\ref{ieq-cfl-overall-bound-1}).
%

\begin{lemma}
\label{lemma-cfl-cost-overall-assignment-cost}
We have
\begin{align*}
\sum_{i \in U,\; j \in \MC{D}} c_{i,j} \cdot x'_{i,j} \;\; + \;\; \frac{3}{1-\alpha} \cdot \sum_{i \in I, \; j \in \MC{D}} c_{i,j} \cdot \overline{x}^{(0)}_{i,j}  
\;\; \le \;\; \frac{7-4\alpha}{\left(1-\alpha\right)^2} \cdot \sum_{i \in \MC{F}, \; j \in \MC{D}} c_{i,j} \cdot x_{i,j}.
\end{align*}
\end{lemma}

\begin{proof}
For any $i \in U, j \in \MC{D}$ and any $p \in P(i,j)$, define the length of path $p$ as 
$$|p| \; := \; \sum_{\substack{i' \in \MC{F}, \; j' \in \MC{D}, \\[2pt] (j'^s,i') \in p}} c_{i',j'} \; + \; \sum_{\substack{ i' \in \MC{F}, \; j' \in \MC{D}, \\[2pt] (i',j'^s) \in p}} c_{i',j'}. $$
By triangle inequality we have $c_{i,j} \le |p|$.
Applying the definition for $\bm{x}'$ and $\overline{x}^{(0)}$, we obtain
\begin{align}
& \sum_{i \in U,\; j \in \MC{D}} c_{i,j} \cdot x'_{i,j} \;\; + \;\; \frac{3}{1-\alpha} \cdot \sum_{i \in I, \; j \in \MC{D}} c_{i,j} \cdot \overline{x}^{(0)}_{i,j} \;\; 
\le \;\; \sum_{ \substack{ i \in U^{(\phi)}, \\[2pt] j \in \MC{D} } } c_{i,j} \cdot g_{i,j} 
\; + \sum_{ \substack{ i \in U^{(>)} \setminus U^{(\phi)}, \\[2pt] j \in D^{(\phi)} } } c_{i,j} \cdot h_{i,j}  
\notag \\[4pt]
& \hspace{1.2cm}
\; + \; \frac{1}{1-\alpha} \cdot \sum_{ \substack{ i \in U^{(>)} \setminus U^{(\phi)}, \\[2pt]  j \in \MC{D} \setminus D^{(\phi)}, \\[2pt]  p \in P(i,j) } } |p| \cdot f'_p 
\; + \; \frac{1}{1-\alpha} \cdot \sum_{ \substack{ i \in U^{(\le)}, \\[2pt]  j \in \MC{D}, \\[2pt]  p \in P(i,j) } } |p| \cdot f_p 
\; + \; \frac{3}{1-\alpha} \cdot \sum_{ \substack{ i \in I, \; j \in \MC{D}, \\[2pt]  p \in P(i,j) } } |p| \cdot f_p.
\label{ieq-cfl-overall-assignment-cost-1}
\end{align}

\smallskip

To bound the R.H.S. of~(\ref{ieq-cfl-overall-assignment-cost-1}), we consider each of the items and the total amount of flow these items has accounted for along the two edges $(j^s,i)$ and $(i,j^s)$ for each $i \in \MC{F}$ and $j \in \MC{D}$.
%
%
%
To be precise, we rewrite the R.H.S. of~(\ref{ieq-cfl-overall-assignment-cost-1}) as
$$\sum_{i \in \MC{F}, j \in \MC{D}} \; \left( \; a^{(\Rightarrow)}_{i,j} \; + \; b^{(\Leftarrow)}_{i,j} \; \right),$$
where $a^{(\Rightarrow)}_{i,j}$ and $b^{(\Leftarrow)}_{i,j}$ denote the total amount of flow the items in the R.H.S. of~(\ref{ieq-cfl-overall-assignment-cost-1}) have accounted for along the two edges $(j^s,i)$ and $(i,j^s)$, respectively.

%
%

\smallskip

%
%
In the following, we bound $a^{(\Rightarrow)}_{i,j}$ and $b^{(\Leftarrow)}_{i,j}$ for each $i \in \MC{F}$ and $j \in \MC{D}$ separately.
%
%
Depending on the categories to which $i$ and $j$ belong, we have the following three cases.

\begin{itemize}
	\item
		$i \in U^{(\phi)}$ is tightly-occupied.
		
		In this case, we have 
		\begin{align*}
		& a^{(\Rightarrow)}_{i,j} \; \le \; \frac{3}{1-\alpha} \cdot \sum_{ \substack{ p \in P, \\[2pt] (j^s,i) \in p } } c_{i,j} \cdot f_p \;\; \le \frac{3}{1-\alpha} \cdot c_{i,j} \cdot x_{i,j}	\quad\; \text{and} \\[6pt]
		& b^{(\Leftarrow)}_{i,j} \; \le \; c_{i,j} \cdot g_{i,j} \; + \; \frac{3}{1-\alpha} \cdot \sum_{ \substack{ p \in P, \\[2pt]  (i,j^s) \in p } } c_{i,j} \cdot f_p  \;\; \le \;\; \left( 1 + \frac{3}{1-\alpha} \right) \cdot c_{i,j} \cdot g_{i,j}
		\end{align*}
		by constraints~(\ref{LP_MFN_2}) and~(\ref{LP_MFN_3}) in $\mathbf{MFN}_\Psi(x,y',g)$.
		Since $g_{i,j} \le x_{i,j} / (1-\alpha)$, we obtain $$a^{(\Rightarrow)}_{i,j} + b^{(\Leftarrow)}_{i,j} \;\; \le \;\; \frac{7 - 4\alpha}{\; (1-\alpha)^2 \;} \cdot c_{i,j} \cdot x_{i,j}.$$
		
	\item
		$i \in U^{(>)} \setminus U^{(\phi)}, \; j \in D^{(\phi)}$.
		In this case, we have $b^{(\Leftarrow)}_{i,j} = 0$.
		%
		By Proposition~\ref{prop-MFN-flow-client-without-partial-assignment}, none of commodities in $\MC{D} \setminus D^{(\phi)}$ has sent nonzero flow through the edge $(j^s,i)$, i.e., 
		$$\sum_{ \substack{ i' \in \MC{F}, \; j' \in \MC{D} \setminus D^{(\phi)}, \\ p \in P(i',j') \text{ s.t. } (j^s,i) \in p } } f_p = 0.
		\qquad \text{Hence,} \quad
		a^{(\Rightarrow)}_{i,j} \; + \; b^{(\Leftarrow)}_{i,j} \;\; \le \;\; c_{i,j} \cdot h_{i,j} \;\; \le \;\; \frac{1}{1-\alpha} \cdot c_{i,j} \cdot x_{i,j}.$$
		
	\item
		For the remaining cases, i.e., $j \in \MC{D} \setminus D^{(\phi)}$ or $i \in \MC{F} \setminus U^{(>)}$.
		We have $b^{(\Leftarrow)}_{i,j} = 0$ and by a similar argument,
		$$a^{(\Rightarrow)}_{i,j} \; + \; b^{(\Leftarrow)}_{i,j} \;\; \le \;\; \frac{3}{1-\alpha} \cdot \sum_{ \substack{p \in P, \\[2pt]  (j^s,i)\in p } } c_{i,j} \cdot f_p \;\; \le \;\; \frac{3}{1-\alpha} \cdot c_{i,j} \cdot x_{i,j}.$$
\end{itemize}
In all cases, $a^{(\Rightarrow)}_{i,j} + b^{(\Leftarrow)}_{i,j}$ is upper-bounded by $\frac{7-4\alpha}{\; (1-\alpha)^2 \;} \cdot c_{i,j} \cdot x_{i,j}$.
This proves the lemma.
\end{proof}

\smallskip

\noindent
Combining~(\ref{ieq-cfl-overall-bound-1}) with Lemma~\ref{lemma-cfl-cost-overall-assignment-cost}, we obtain
\begin{align*}
\psi\left( \; \bm{x}''', \; \bm{y}^* \; \right) \; 
& \le \;\; \frac{3}{2\alpha} \cdot \sum_{i \in \MC{F}} o_i \cdot y_i \; + \; \frac{7-4\alpha}{\left( 1-\alpha \right)^2} \cdot \sum_{i \in \MC{F}, \; j \in \MC{D}} c_{i,j} \cdot x_{i,j} \\[4pt]
& \le \;\; \max\left\{ \; \frac{3}{2\alpha} \; , \; \frac{7-4\alpha}{(1-\alpha)^2} \; \right\} \cdot \psi(\bm{x}, \bm{y}).
\end{align*}
This completes the proof for Theorem~\ref{thm-MFN-relaxed-separation-oracle}.

%



%

%


\section{Proof of Theorem~\ref{theorem-cfl-ufc-approx}.}
\label{sec-proof-approx-cfl-cfc}

We outline the proof 
as follows.
In Section~\ref{subsec-proof-approx-cfl-cfc-feasibility}, we show that the rounding algorithm is well-defined and terminates in polynomial time.
We define in the same section the rounded assignment function $\bm{x}^\circ$ and shows that $(\bm{x}^\circ, y^*)$ is feasible for~\ref{LP-natural-CFL} on $\Psi$.
This shows that the min-cost assignment problem for $(\MC{D}, \MC{F}^*)$ is feasible, and hence the integral assignment $\bm{x}^\dagger$ can be computed.
%
%
%
%
In Section~\ref{subsec-proof-approx-cfl-cfc-guarantee}, we establish the $4$-approximation guarantee for $(\bm{x}^\circ, y^*)$.
%
%
This completes the proof for Theorem~\ref{theorem-cfl-ufc-approx} since $\bm{x}^\dagger$ is the optimal solution for the min-cost assignment problem on $(\MC{D}, \MC{F}^*)$.
%


%




%
\paragraph*{Notations to use in the proof.}

In the following, we define notations and notions that help describe our 
rounding process with precision in the analysis.
As the notations can be subtle, we refer the readers to kindly check the definitions to prevent notational ambiguity in the proof.

\smallskip

Consider the cluster-forming process. 
Let $\MC{C}_{D'}$ and $\MC{C}_{H'}$ denote the sets of clusters centered at the non-outlier clients and outlier clients, respectively.
For each $q \in \MC{C}_{D'}$, we use $j(q)$ to denote the center client of $q$ 
%
and $i(q)$ denote the facility that is selected to be rounded up in the iteration when $q$ is formed.
Let $F^*_{D'} := \left\{ \; \vphantom{\text{\Large T}} i(q) \; \colon \; q \in \MC{C}_{D'} \; \right\}$ denote the set of facilities rounded up for the clusters in $\MC{C}_{D'}$.
Note that, $F^*_{D'}$ and $G$ are disjoint by the algorithm design.
Furthermore, the set of satellite facilities $B(j)$ for each $j \in H$ forms a partition of $G$.

\smallskip

For each $q \in \MC{C}_{D'}$, we 
use $D'^{(q)}$, $F'^{(q)}$, $H^{(q)}$, $H'^{(q)}$, $\bm{x}'^{(q)}$, and $\bm{y}'^{(q)}$
to denote the set $D'$, the set $F'$, the sets $H$, the set $H'$, the assignment $\bm{x}'$, and the multiplicity $\bm{y}'$
the algorithm maintains at the moment when the cluster $q$ was formed.
We use $B(q)$ to denote the set of satellite facilities at that moment, i.e., $B(q) := N_{(F'^{(q)},x'^{(q)})}(j(q))$.
%
%

\smallskip

For each outlier client $j \in H$, we use $w(j)$ to denote the facility in $U$ at which $j$ is located.
We use $p(j)$ to denote the specific parent client in $J^{(\leftrightarrow)}$ from which $j$ is created.
%
%
%
%
On the contrary, for any $j \in J^{(\leftrightarrow)}$, we use $H(j)$ to denote the set of outlier clients that are created from $j$.
For each $w \in U$, we use $H(w)$ to denote the set of outlier clients located at $w$.

\smallskip

%
To prevent ambiguity on the usage of $(\bm{x}',\bm{y}')$, we will specifically use $(\bm{x}'^{(0)}, \bm{y}'^{(0)})$ to denote the initial solution the algorithm has for $\Psi$.
For outlier clients $j \in H$ and any $i \in \MC{F}$, we use $\bm{x}'^{(0)}_{i,j}$ to denote the assignment made for $j$ to $i$ at the moment when $j$ is created.
We additionally use $\bm{x}'^{(\text{II})}$ and $\bm{y}'^{(\text{II})}$ to denote the assignment $\bm{x}'$ and the multiplicity $\bm{y}'$ the algorithm maintains when it enters the second phase.
\subsection{The Feasibility}
\label{subsec-proof-approx-cfl-cfc-feasibility}

%
%
In Section~\ref{par-proof-approx-cfl-cfc-feasibility-cd} we show that the first stage of the rounding process is well-defined and terminates in polynomial time.
%
%
In Section~\ref{par-proof-approx-cfl-cfc-feasibility-ch} we consider the second stage of the process and show that the feasible region of~\ref{LP-outliers} is nonempty.
%
We define the intermediate assignment $\bm{x}^\circ$ and prove the feasibility of $(\bm{x}^\circ, \bm{y}^*)$ for~\ref{LP-natural-CFL} on $\Psi$ in Section~\ref{par-proof-approx-cfl-cfc-feasibility-xc}.
%


%


\subsubsection{The first stage of the rounding process}
\label{par-proof-approx-cfl-cfc-feasibility-cd}

We show that the first stage of our rounding process is well-defined and terminates in polynomial time.
Consider any particular moment in the first stage, and let $(\bm{x}',\bm{y}')$, $F'$, $D'$, and $H'$ denote the parameters the algorithm maintains at that moment.
%

\smallskip

Consider the capacity constraints for facilities in $F'$ and the third constraint from~\ref{LP-natural-CFL} for facilities in $F'$ and clients in $D' \cup H'$, listed as follows.
\begin{center}
\begin{minipage}{.65\textwidth}
\fbox{
\begin{minipage}{\textwidth}
\vspace{-5pt}
\begin{align}
\sum_{j \in D' \cup H'} \; x'_{i,j} \; \le \; u_i \cdot y'_i, & & & \forall i \in F'. \label{cons-P-2-with-H} \tag*{(MN-2)} \\
x'_{i,j} \; \le \; y'_i, & & & \forall i \in F', j \in D' \cup H'.  \label{cons-P-3-with-H} \tag*{(MN-3)}
\end{align}
\vspace{-12pt}
\end{minipage}\quad\enskip
}
\end{minipage}
\end{center}

\smallskip

%
%

It is clear that~\ref{cons-P-2-with-H} and~\ref{cons-P-3-with-H} hold in the beginning of the rounding process, since initially $F' := I$, $D' := J^{(I)} \cup J^{(\leftrightarrow)}$, $H' := \emptyset$, and $(\bm{x}',\bm{y}') := (\bm{x}'^{(0)}, \bm{y}'^{(0)})$ is feasible for~\ref{LP-natural-CFL}.

\smallskip

The following two lemmas establish that, the processes of creating outlier clients and cluster-forming do not render the validity of~\ref{cons-P-2-with-H} and~\ref{cons-P-3-with-H}.

\begin{lemma} \label{lemma-feasibility-creating-outlier-client}
The process of creating outlier clients 
does not render~\ref{cons-P-2-with-H} nor~\ref{cons-P-3-with-H} invalid.
\end{lemma}

\begin{proof}
Consider the process of creating outlier clients from a client, say, $j \in J^{(\leftrightarrow)} \cap D'$.
By the algorithm design, this happens when $\sum_{i \in F'} x'_{i,j} < 1/2$.

\smallskip

For each $i \in N_{(F',x')}(j)$, consider the assignments the algorithm has made from $H(j)$ to $i$ after $H(j)$ is created.
The total amount of assignment $i$ receives from $H(j)$ is
%
\begin{align*}
\sum_{\ell \in H(j)} x'^{(0)}_{i,\ell} \; := \; \sum_{\ell \in H(j)} d_\ell \cdot \frac{ x'_{i,j} }{ \sum_{ k \in F'} x'_{k,j} } 
\;\; = \;\; \frac{ x'_{i,j} }{ \sum_{ k \in F'} x'_{k,j} } \cdot \sum_{\ell \in H(j)} r'_j \cdot \frac{ x'_{w(\ell),j} }{ \sum_{w \in U} x'_{w,j} },
\end{align*}
where in the above we apply the definition of $x'^{(0)}_{i,\ell}$ and $d_\ell$ for any $\ell \in H(j)$.
By the definition of $r'_j$, we have $r'_j \le \sum_{k \in F'} x'_{k,j}$.
Hence, the above becomes
\begin{align}
\sum_{\ell \in H(j)} x'^{(0)}_{i,\ell} \;\; \le \;\;  x'_{i,j} \cdot \sum_{\ell \in H(j)} \frac{ x'_{w(\ell),j} }{ \sum_{w \in U} x'_{w,j} } \;\; = \;\; x'_{i,j},
\label{ieq-lemma-feasibility-creating-outlier-client-1}
\end{align}
where in the last equality we apply the fact that $\sum_{\ell \in H(j)} x'_{w(\ell),j} = \sum_{w \in U} x'_{w,j}$.
Since the algorithm resets $\vphantom{\text{\Large T}} x_{i,j}$ to be zero after $H(j)$ is created, it follows from~(\ref{ieq-lemma-feasibility-creating-outlier-client-1}) that~\ref{cons-P-2-with-H} still holds
after $H(j)$ is created and the new assignments to $i$ are made.

\smallskip

The argument for~\ref{cons-P-3-with-H} follows analogously.
For any $\ell \in H(j)$, by the above equations, we have $x'^{(0)}_{i,\ell} \le \sum_{k \in H(j)} x'^{(0)}_{i,k} \le x'_{i,j}$.
Hence~\ref{cons-P-3-with-H} holds for any $i \in N_{(F',x')}(j)$ and any $\ell \in H(j)$.
\end{proof}

%

%
%


%

\begin{lemma} \label{lemma-feasibility-scaled-down-well-defined}
We have $0 < \delta_{i(q)} \le 1$ for any $q \in \MC{C}_{D'}$.
Furthermore, the scaled-down operation the algorithm performs when rounding cluster $q$ does not render~\ref{cons-P-2-with-H} nor~\ref{cons-P-3-with-H} invalid.
%
\end{lemma}

\begin{proof}
Consider any $q \in \MC{C}_{D'}$.
We will show that, provided that constraints~\ref{cons-P-2-with-H} and~\ref{cons-P-3-with-H} are valid in the beginning of the iteration for which $q$ is formed,
we have
\begin{itemize}
	\item
		$0 < \delta_{i(q)} \le 1$, and
	\item
		the rounding process for $q$ does not render~\ref{cons-P-2-with-H} and~\ref{cons-P-3-with-H} invalid.
\end{itemize}
Note that this proves the lemma.
%

\smallskip

Since $i(q) \in F'^{(q)}$, we know that $y'^{(q)}_{i(q)} < 1/2$.
Furthermore, since $j(q)$ is selected as the center client and since $j(q) \in D'^{(q)}$ by assumption, it follows that $\sum_{i \in F'^{(q)}} x'^{(q)}_{i,j(q)} \ge 1/2$.
This implies that
$$\sum_{i \in B(q) \setminus \{i(q)\}} y'^{(q)}_i \; \ge \; \sum_{i \in B(q) \setminus \{i(q)\}} x'^{(q)}_{i,j(q)} \; \ge \; \frac{1}{2} - x'^{(q)}_{i(q),j(q)} \; \ge \; \frac{1}{2} - y'^{(q)}_{i(q)} \; > \; 0,$$
where in the first and the last inequalities we apply constraint~\ref{cons-P-3-with-H} for $i \in B(q)$ and $j(q)$.
This shows that $$\delta_{i(q)} \; := \; \left( \; \frac{1}{2} - y'^{(q)}_{i(q)} \; \right) \cdot \frac{1}{ \sum_{i \in B(q)\setminus \{i(q)\}} y'^{(q)}_i } \; > \; 0.$$
On the contrary, by~\ref{cons-P-3-with-H} we have $\sum_{i \in B(q)} y'^{(q)}_i \; \ge \; \sum_{i \in B(q)} x'^{(q)}_{i,j(q)} \; \ge \; 1/2$.
This implies that $1/2 - y'^{(q)}_{i(q)} \le \sum_{i \in B(q) \setminus \{i\}} y'^{(q)}_i$ and $\delta_{i(q)} \le 1$.
%

\medskip

To see that constraints~\ref{cons-P-2-with-H} and~\ref{cons-P-3-with-H} remain valid at the end of this iteration, observe that for each $i \in B(q) \setminus \{i(q)\}$ and any $j \in D'^{(q)}$, both $x'^{(q)}_{i,j}$ and $y'^{(q)}_i$ are scaled down simultaneously by the constant $\left( 1-\delta_{i(q)}\right)$.
%
%
This completes the proof of this lemma.
\end{proof}

By Lemma~\ref{lemma-feasibility-creating-outlier-client} and Lemma~\ref{lemma-feasibility-scaled-down-well-defined},
we obtain the following.

\begin{restatable}{corollary}{corfeasibilitycons}
\label{cor-feasibility-cons}
\ref{cons-P-2-with-H} and~\ref{cons-P-3-with-H} hold throughout the first stage of the rounding process.
\end{restatable}

%
%

%
%
%
It follows that, at any particular moment in the first stage, 
\begin{align}
\sum_{i \in F'} x'_{i,j} \;\; \le \;\; \sum_{i \in F'}y'_i \quad \text{holds for any $j \in D' \cup H'$.}
\label{ieq-feasibility-cd-1}
\end{align}
By the design of the algorithm, we know that at least one facility is removed from $F'$ after each iteration in the first phase. 
Therefore, the rounding process repeats for at most $|I|$ iterations before $F'$ becomes empty.
By~(\ref{ieq-feasibility-cd-1}), this implies that $\sum_{i \in F'} x'_{i,j} = 0$ for all $j \in D' \cup H'$, and $D' \cup H'$ will become empty in at most $|H|$ iterations after that.
This shows that the rounding algorithm terminates in polynomial time.

\smallskip

The following lemma, which shows that the rounded facility is sparsely-loaded by the rerouted assignments, is straightforward to verify.

\begin{restatable}{lemma}{lemmafeasibilityabsorbingprocess}
\label{lemma-feasibility-absorbing-process}
We have $\sum_{j \in \MC{D} \cup H} x^*_{i(q),j} \le u_{i(q)} / 2$ for any $q \in \MC{C}_{D'}$.
%
\end{restatable}

\begin{proof}
Consider any $q \in \MC{C}_{D'}$.
%
By the design of the algorithm and the fact that constraint~\ref{cons-P-3-with-H}
holds throughout the process, we have
\begin{align*}
\sum_{j \in \MC{D} \cup H} x^*_{i(q),j} \;\; 
& = \;\; \sum_{j \in \MC{D} \cup H} x'^{(q)}_{i(q),j} \; + \; \sum_{i \in B(q) \setminus \{i(q)\}} \; \sum_{j \in \MC{D} \cup H} \delta_{i(q)} \cdot x'^{(q)}_{i,j} \\[4pt]
& \le \;\; u_{i(q)} \cdot y'^{(q)}_{i(q)} \; + \; \sum_{i \in B(q) \setminus \{i(q)\}} \delta_{i(q)} \cdot u_i \cdot y'^{(q)}_i \\[2pt]
& \le \;\; u_{i(q)} \cdot y'^{(q)}_{i(q)} \; + \; u_{i(q)} \cdot \left( \; \frac{1}{2} - y'^{(q)}_{i(q)} \; \right) \;\; = \;\; \frac{1}{2} \cdot u_{i(q)},
\end{align*}
where in the third inequality we use apply fact that $u_{i(q)} \ge u_i$ for all $i \in B(q)$ by the way $i(q)$ is selected and the definition of $\delta_{i(q)}$.
\end{proof}

\smallskip
\smallskip

%


\subsubsection{The second stage of the rounding process}
\label{par-proof-approx-cfl-cfc-feasibility-ch}

In this section, we consider the second stage of the rounding process and the clusters in $\MC{C}_{H'}$.
%
The following lemma summarizes the status of the 
facilities and the clients.

\begin{restatable}{lemma}{lemmaassignmentabsorbingprocess}
\label{lemma-assignment-absorbing-process}
When the algorithm enters the second stage, the following holds.
\begin{itemize}
	\item
		For any $i \in G$, \hspace{1pt} $\sum_{j \in \MC{D} \cup H} x'^{(\text{II})}_{i,j} \le u_i \cdot y'^{(\text{II})}_i$.

	\item
		For any $j \in J^{(I)}$, \hspace{1pt} $\sum_{i \in I} x^*_{i,j} + \sum_{i \in G} x'^{(\text{II})}_{i,j} \; > \; 1/2$.
		
	\item
		For any $j \in J^{(\leftrightarrow)}$, \hspace{1pt} $\sum_{i \in I} x^*_{i,j} + \sum_{i \in G} x'^{(\text{II})}_{i,j} + \sum_{i \in U} x'^{(0)}_{i,j} \; > \; 1/2$.
		
	\item
		For any $j \in H$, \hspace{1pt} $\sum_{i \in F^*_{D'}} x^*_{i,j} + \sum_{i \in G} x'^{(\text{II})}_{i,j} \; = \; \sum_{i \in I} x'^{(0)}_{i,j}$.
\end{itemize}
\end{restatable}

\begin{proof}
The first statement of this lemma follows directly from Corollary~\ref{cor-feasibility-cons} and the definition of $(\bm{x}'^{(\text{II})}, \bm{y}'^{(\text{II})})$.
%
%
%
%
The remaining 
of this lemma follows from the way how the algorithm handles the residue demand of each client.
Consider the moment for which each $j \in J^{(I)} \cup J^{(\leftrightarrow)} \cup H$ is removed from consideration in the first phase, and the fact that $G := \bigcup_{j \in H} B(j)$.
%

\smallskip

For $j \in J^{(I)} \cup J^{(\leftrightarrow)}$, it is removed when $\sum_{i \in F'} x'_{i,j} < 1/2$. 
It follows that the assignments rerouted to clusters in $\MC{C}_{D'}$, the assignments taken into clusters in $\MC{C}_{H'}$, and possibly the original assignments to facilities in $U$, account for at least $1/2$.
%
%

\smallskip

For $j \in H$, it is removed when selected as the center of a cluster, possibly an empty cluster.
When this happens, all of the remaining assignments for $j$ are taken into this cluster.
\end{proof}

%


%
Lemma~\ref{lemma-assignment-absorbing-process} leads to the following corollary on the scaling factor $t'_j$ for all $j \in \MC{D}$.

\begin{restatable}{corollary}{corscalingfactor}
\label{cor-scaling-factor}
$0 \le t'_j \le 2$ for all $j \in \MC{D}$.
\end{restatable}

\begin{proof}
%
%
It suffices to prove the statement for $j \in \MC{D}$ with $\sum_{i \in I} x^*_{i,j} + \sum_{i \in G} x'^{(\text{II})}_{i,j} > 0$.

\smallskip

Since $\sum_{i \in I}x^*_{i,j}  + \sum_{i \in G} x'^{(\text{II})}_{i,j} > 0$ implies that $j \in J^{(I)} \cup J^{(\leftrightarrow)}$, by the definition of $r'_j$, we have
$$1-\sum_{i \in U} x'^{(0)}_{i,j} - r'_j \;\; \ge \;\; 1-\sum_{i \in U} x'^{(0)}_{i,j} - \sum_{i \in I} x'^{(0)}_{i,j} \;\; \ge \;\; 0, \enskip \text{which implies that} \enskip t'_j > 0.$$
%
%
%

\smallskip

In the following we show that $t'_j \le 2$.
%
Since $j \in J^{(I)} \cup J^{(\leftrightarrow)}$, it suffices to prove the statement for the following three cases.
\begin{itemize}
	\item
		If $j \in J^{(I)}$, then $t'_j < 2$ directly from the conclusion of Lemma~\ref{lemma-assignment-absorbing-process} since 
		$$\sum_{i \in I} x^*_{i,j} + \sum_{i \in G} x'^{(\text{II})}_{i,j} > 1/2 \quad \text{and} \quad (1-\sum_{i \in U}x'^{(0)}_{i,j} - r'_j) \le 1.$$
		
		%
		
		
	\item
		If $j \in J^{(\leftrightarrow)}$ and $r'_j \neq \sum_{i \in U} x'^{(0)}_{i,j}$,
		then all of the residue demand of $j$ has been redistributed as outlier clients when $j$ is to be removed from $D'$.
		It follows that $$1 \; - \; \sum_{i \in U} x'^{(0)}_{i,j} \; - \; r'_j \;\; = \;\; \sum_{i \in I}x^*_{i,j} \; + \; \sum_{i \in G} x^{(\text{II})}_{i,j} \qquad \text{and} \quad t'_j = 1.$$
		
		\smallskip
		
	\item
		If $j \in J^{(\leftrightarrow)}$ and $r'_j := \sum_{i \in U} x'^{(0)}_{i,j}$, 
		then by the conclusion of Lemma~\ref{lemma-assignment-absorbing-process} 
		we have $\sum_{i \in I} x^*_{i,j} + \sum_{i \in G} x'^{(\text{II})}_{i,j} + \sum_{i \in U} x'^{(0)}_{i,j} > 1/2$, which implies that
		$$1 \; - \; \sum_{i \in U} x'^{(0)}_{i,j} \; - \; r'_j \;\; = \;\; 1 \; - \; 2\cdot \sum_{i \in U} x'^{(0)}_{i,j} \;\; < \;\;  2 \cdot \left( \; \sum_{i \in I}x^*_{i,j} + \sum_{i \in G} x'^{(\text{II})}_{i,j} \; \right)$$
		and $t'_j < 2$.
\end{itemize}
\noindent
In all cases we have $t'_\ell \le 2$.
\end{proof}

%


%
%

%

\begin{figure*}[h]
\centering
\fbox{
\begin{minipage}{.8\textwidth}
\begin{align}
& \text{min} & & \sum_{i \in G} \; y_i \; + \; \sum_{i \in G, \hspace{1pt} j \in U} c_{i,j} \cdot x_{i,j}  & &  \label{LP-CFL-outliers} \tag*{LP-(O)} \\[3pt]
& \text{s.t.} & & \sum_{i\in G} \; x_{i,j} \; = \; d_j, & & \forall j\in U, \tag*{(O-1)} \label{LP-CFL-outliers-M1} \\[2pt]
& & & \sum_{j \in U} \; x_{i,j} \; \le \; u_i \cdot y_i, & & \forall i \in G, \tag*{(O-2)} \label{LP-CFL-outliers-M2} \\[3pt]
& & & y_i \; \le \; 1, & & \forall i \in G, \tag*{(O-3)} \label{LP-CFL-outliers-M3} \\[4pt]
& & & x_{i,j} \; \ge \; 0, \;\; y_i \; \ge \; 0, & & \forall i \in G, \; j \in U. \tag*{(O-4)} \label{LP-CFL-outliers-M4}
\end{align}
\vspace{-10pt}
\end{minipage}\quad
}
\caption{(Restate) The assignment LP for the outlier clusters in $\MC{C}_{H'}$.}
\label{fig-LP-CFL-outliers}
\end{figure*}


\paragraph*{The bundled assignment $\bm{g}$ for~\ref{LP-outliers}.}

%
%
For any $w \in U$ and $i \in G$ such that $i \in B(k)$ for some $k \in H(w)$, i.e., $i$ belongs to the clusters centered at some $k \in H(w)$, define the bundled assignment $g_{i,w}$ as
$$g_{i,w} \; := \; \sum_{\ell \in \MC{D} \cup H} t'_\ell \cdot x'^{(\text{II})}_{i,\ell}. $$
Intuitively, $g_{i,w}$ is the total amount of scaled assignments $i$ has when the algorithm enters the second stage.
The following lemma shows that the feasible region of~\ref{LP-outliers} is nonempty, and the basic optimal solution $(\bm{x}'', \bm{y}'')$ exists.
%


\begin{restatable}{lemma}{lemmaassignmentlpfeasibility}
\label{lemma-assignment-lp-feasibility}
%
$\left( \; \bm{g}, \; 2 \hspace{1pt} \bm{y}'^{(\text{II})} \; \right)$ is feasible for~\ref{LP-outliers}.
\end{restatable}

\begin{proof}
We prove by verifying the constraints of~\ref{LP-CFL-outliers}. 
Also refer to Figure~\ref{fig-LP-CFL-outliers} for the numbering of the constraints.
%
\begin{itemize}
	\item
		Consider the constraint~\ref{LP-CFL-outliers-M1}.
		
		For any $w \in U$, apply the definition of $\bm{g}$ and the definition of $d_w$, we have 
		\begin{align*}
		\sum_{i \in G} \; g_{i,w} 
		\; = \; \sum_{ k \in H(w), \; i' \in B(k)} \; \sum_{\ell \in \MC{D} \cup H} \; t'_\ell \cdot x'^{(\text{II})}_{i',\ell} \; = \; d_w.
		\end{align*}
		
	\item
		Consider the constraint~\ref{LP-CFL-outliers-M3}.
		
		For any $i \in G$, we have $y'^{(\text{II})}_i \le y'^{(0)}_i \le \; 1/2$ since $G \subseteq I$, and $2\cdot y'^{(\text{II})}_i \; \le \; 1$.
		
		\smallskip
		
	\item
		Consider the constraint~\ref{LP-CFL-outliers-M2}.
		
		For any $i \in G$, let $k \in H$ be the outlier client such that $i \in B(k)$. 
		Applying the definition of $\bm{g}$ and 
		Corollary~\ref{cor-scaling-factor}, we have 
		\begin{align*}
		\sum_{w \in U} g_{i,w} \; 
		& = \; g_{i,w(k)} \; = \; \sum_{\ell \in \MC{D} \cup H} t'_\ell \cdot x'^{(\text{II})}_{i,\ell} \; \le \; \sum_{\ell \in \MC{D} \cup H} 2\cdot x'^{(\text{II})}_{i,\ell} \; \le \; 2\cdot u_i \cdot y'^{(\text{II})}_i,
		\end{align*}
		where in the last inequality we apply the conclusion of Corollary~\ref{cor-feasibility-cons} which states that constraint~\ref{cons-P-2-with-H} holds for $i$ when $i$ is removed from $F'$ in the first stage of the rounding process.
		%
\end{itemize}
This proves the lemma.
\end{proof}

%


\paragraph*{The unbundled assignment $\bm{h}$ from $\bm{x}''$.}

Consider the basic optimal solution $(\bm{x}'', \bm{y}'')$ for~\ref{LP-outliers}.
In the following, we unbundle the assignment $\bm{x}''$ as assignment function $h$ for the original clients in $\MC{D} \cup H$.

\smallskip

For each $i \in G$ and $j \in \MC{D} \cup H$, define the unbundled assignment $h_{i,j}$ as
$$h_{i,j} \hspace{4pt} := \hspace{4pt} \sum_{w \in U} \; x''_{i,w} \cdot \frac{1}{d_w} \cdot \sum_{k \in H(w), \; i' \in B(k)} t'_j \cdot x'^{(\text{II})}_{i',j}.$$
%
Intuitively, in $h$ we redistribute the assignment $\bm{x}''$ back for the original clients in $\MC{D} \cup H$ proportionally.
It follows that for any $j \in \MC{D} \cup H$, 
\begin{align}
\sum_{i \in G} \; h_{i,j} \;\; 
& = \;\; \sum_{ i \in G, \; w \in U} \; x''_{i,w} \cdot \frac{1}{d_w} \cdot \sum_{k \in H(w), \; i' \in B(k)} t'_j \cdot x'^{(\text{II})}_{i',j}  \notag \\[6pt]
& = \;\; \sum_{w \in U} \; \sum_{k \in H(w), \; i' \in B(k)} t'_j \cdot x'^{(\text{II})}_{i',j} 
\;\; = \;\; \sum_{i \in G} \; t'_j \cdot x'^{(\text{II})}_{i,j},
\label{ieq-unbundle-fully-assigned}
\end{align}
where in the second equality we apply 
the first constraint of~\ref{LP-outliers} 
and in the last equality we use the fact that the set of satellite facilities for each $j \in H$ forms a partition of $G$.
%


%
%
%
%
%
%
%
%
%
%

\medskip

%

%

%


\subsubsection{The Rounded Assignment}
\label{par-proof-approx-cfl-cfc-feasibility-xc}

Provided the above, the intermediate assignment $\bm{x}^\circ$ for each $j \in \MC{D}$ is defined as
%
%
%
%
$$x^\circ_{i,j} \;\; := \;\; \begin{cases}
\; x'^{(0)}_{i,j}, & \text{if $i \in U$,} \\[3pt]
\; t'_j \cdot x^*_{i,j} \; + \; \sum_{k \in H(j)} x^*_{i,k}, \enskip & \text{if $i \in F^*_{D'}$,} \\[3pt]
\; h_{i,j} \; + \; \sum_{k \in H(j)} h_{i,k}, & \text{if $i \in G$,} \\[2pt]
\; 0, & \text{otherwise}.
\end{cases}$$
Intuitively, the assignment of each $j \in \MC{D}$ in $\bm{x}^\circ$ consists of its original assignments to $U$ and the rounded assignments for clients in $\{j\} \cup H(j)$ to facilities in $F^*_{D'} \cup G$.

\smallskip

The following lemma, which asserts the feasibility of $\bm{x}^\circ$, is 
straightforward to verify. 
%

\begin{lemma} \label{lemma-cfl-ufc-feasibility}
%
$(\bm{x}^\circ, \bm{y}^*)$ is feasible for~\ref{LP-natural-CFL} on the input instance $\Psi$.
\end{lemma}

\begin{proof}
%
%
Since $\bm{y^*}$ is already integral and takes values only from $\{0, 1\}$, it suffices to show that $\bm{x}^\circ$ fully-assigns each $j \in \MC{D}$ and respects the capacity constraints given by $\bm{y}^*$.

\smallskip

For the latter part, since $\bm{x}^\circ$ keeps the assignments of $\MC{D}$ to $U$ unchanged, it suffices to examine the assignments to $F^*_{D'} \cup G$.
By the definition of $\bm{x}^\circ$, Corollary~\ref{cor-scaling-factor}, and Lemma~\ref{lemma-feasibility-absorbing-process}, 
for any $i \in F^*_{D'}$, we have $$\sum_{j \in \MC{D}} x^\circ_{i,j} \; = \; \sum_{j \in \MC{D}} t'_j \cdot x^*_{i,j} \; + \; \sum_{j \in H} x^*_{i,j} \; \le \; \sum_{j \in \MC{D} \cup H} 2\cdot x^*_{i,j} \; \le \; u_i \; = \; u_i \cdot y^*_i.$$
%
Similarly, for any $i \in G$, applying the definition of $\bm{x}^\circ$ and $\bm{h}$, we have 
\begin{align*}
\sum_{j \in \MC{D}} \; x^\circ_{i,j} \; 
& = \; \sum_{j \in \MC{D} \cup H} h_{i,j} \; = \; \sum_{j \in \MC{D} \cup H} \; \sum_{w \in U} x''_{i,w} \cdot \frac{1}{d_w} \cdot \sum_{k \in H(w), \; i' \in B(k)} t'_j \cdot x'^{(\text{II})}_{i',j} \\[4pt]
& = \; \sum_{w \in U} \; x''_{i,w} \cdot \frac{1}{d_w} \cdot \sum_{k \in H(w), \; i' \in B(k)} \; \sum_{j \in \MC{D} \cup H} t'_j \cdot x'^{(\text{II})}_{i',j}
\; = \; \sum_{w \in U} x''_{i,w} \; \le \; u_i \cdot y''_i,
\end{align*}
where in the second last equality we apply the definition of $d_w$ for any $w \in U$ and in the last inequality we use the fact that $(\bm{x}'', \bm{y}'')$ is feasible for~\ref{LP-CFL-outliers}.

\smallskip

Next, we show that $\bm{x}^\circ$ fully-assigns each $j \in \MC{D}$.
It suffices to prove for the clients in $J^{(I)} \cup J^{(\leftrightarrow)}$.
For the former case, for any $j \in J^{(I)}$, we have $\sum_{i \in U} x'^{(0)}_{i,j} = 0$ and $H(j) = \emptyset$.
%
\begin{align*}
\text{Hence,} \quad \sum_{i \in \MC{F}} x^\circ_{i,j} \;
& = \; \sum_{i \in F^*_{D'}} t'_j \cdot x^*_{i,j} \; + \; \sum_{i \in G} \; h_{i,j} \; \\
& \hspace{-20pt} = \; \sum_{i \in F^*_{D'}} t'_j \cdot x^*_{i,j} \; + \; \sum_{i \in G} \; t'_j \cdot x'^{(\text{II})}_{i,j} \; = \; t'_j \cdot \left( \; \sum_{i \in F^*_{D'}} x^*_{i,j} \; + \; \sum_{i \in G} x'^{(\text{II})}_{i,j} \; \right) \; = \; 1,
\end{align*}
where 
in the second equality we apply Equality~(\ref{ieq-unbundle-fully-assigned}), and in the last equality we apply the definition of $t'_j$ with the fact that $\sum_{i \in U} x'^{(0)}_{i,j} = r'_j = 0$.

\smallskip

For $j \in J^{(\leftrightarrow)}$, we have
\begin{align}
\sum_{i \in \MC{F}} x^\circ_{i,j} \; = \; \sum_{i \in U} x'^{(0)}_{i,j} + \sum_{i \in F^*_{D'}} \left( \; t'_j \cdot x^*_{i,j} + \sum_{k \in H(j)} x^*_{i,k} \; \right) + \sum_{i \in G} \left( \; h_{i,j} + \sum_{k \in H(j)} h_{i,k} \; \right)
\label{ieq-j-lr-fully-assign-1}
\end{align}
Applying Equality~(\ref{ieq-unbundle-fully-assigned}) and the definition of $t'_k$ for $k \in H(j)$, we have
\begin{align}
\sum_{i \in G} \left( \; h_{i,j} \; + \; \sum_{k \in H(j)} h_{i,k} \; \right) \;\; = \;\; \sum_{i \in G} \; t'_j \cdot x'^{(\text{II})}_{i,j} \; + \; \sum_{k \in H(j), \; i \in G} x'^{(\text{II})}_{i,k}
\label{ieq-j-lr-fully-assign-2}
\end{align}
Combining~(\ref{ieq-j-lr-fully-assign-1}) and~(\ref{ieq-j-lr-fully-assign-2}), we have
\begin{align}
\sum_{i \in \MC{F}} \; x^\circ_{i,j} \; 
& = \; \sum_{i \in U} x'^{(0)}_{i,j} + \sum_{i \in F^*_{D'}} t'_j \cdot x^*_{i,j} + \sum_{i \in G} t'_j \cdot x'^{(\text{II})}_{i,j} + \sum_{k \in H(j)} \left( \; \sum_{i \in F^*_{D'}} x^*_{i,k} + \sum_{i \in G} x'^{(\text{II})}_{i,k} \; \right) \notag \\[6pt]
& = \;\; \sum_{i \in U} x'^{(0)}_{i,j} \; + \; \sum_{i \in F^*_{D'}} t'_j \cdot x^*_{i,j} \; + \; \sum_{i \in G} t'_j \cdot x'^{(\text{II})}_{i,j} \; + \; \sum_{k \in H(j)} \; \sum_{i \in I} x'^{(0)}_{i,k},
\label{ieq-j-lr-fully-assign-3}
\end{align}
where in the last equality we apply the conclusion of Lemma~\ref{lemma-assignment-absorbing-process}.
By the construction of outlier clients, each $k \in H(j)$ is fully-assigned to facilities in $I$ by $\bm{x}'$.
Hence, we have $\sum_{i \in I} x'^{(0)}_{i,k} = d_k$ for each $k \in H(j)$.
Further applying the definition of $d_k$, 
we have
\begin{align}
\sum_{k \in H(j)} \; \sum_{i \in I} x'^{(0)}_{i,k} \; = \; \sum_{k \in H(j)} d_k \; = \; \sum_{k \in H(j)} r'_j \cdot \frac{x'^{(0)}_{w(k),j}}{\sum_{i \in U} x'^{(0)}_{i,j}} \; = \; r'_j,
\label{ieq-j-lr-fully-assign-4}
\end{align}
where the last equality follows from the fact that exactly one outlier client is created for each $i \in U$ with $x'^{(0)}_{i,j} > 0$.
Combining~(\ref{ieq-j-lr-fully-assign-3}) and~(\ref{ieq-j-lr-fully-assign-4}) and applying the definition of $t'_j$, we have
$$\sum_{i \in \MC{F}} \; x^\circ_{i,j} \;\; 
= \;\; \sum_{i \in U} x'^{(0)}_{i,j} \; + \; \sum_{i \in F^*_{D'}} t'_j \cdot x^*_{i,j} \; + \; \sum_{i \in G} t'_j \cdot x'^{(\text{II})}_{i,j} \; + \; r'_j \;\; = \;\; 1.$$
This completes the proof of this lemma.
\end{proof}

%

%

%


\subsection{Approximation Guarantee}
\label{subsec-proof-approx-cfl-cfc-guarantee}

%


%
In the following we establish the $4$-approximation guarantee for $(\bm{x}^\circ, \bm{y}^*)$.
%
We consider the cost incurred by clusters in $\MC{C}_{H'}$ and $\MC{C}_{D'}$ separately in Section~\ref{par-proof-approx-cfl-cfc-cost-ch} and Section~\ref{par-proof-approx-cfl-cfc-cost-cd}.
In Section~\ref{par-proof-approx-cfl-cfc-overall-cost} we establish the overall guarantee.
%
%
%
%

\smallskip

Recall that we use $p(j)$ for $j \in H$ to denote the client in $\MC{D}$ from which $j$ is created.
In the following, we extend the definition and define $p(k) := k$ for any $k \in \MC{D}$ for brevity.
%

\smallskip

Moreover, for any assignment $\bm{x}$ of interest, we will use $\bm{x}|_{A,B}$ to denote the assignments made in $\bm{x}$ between $A \subseteq \MC{F}$ and $B \subseteq \MC{D} \cup H$.
Similarly, for any multiplicity function $\bm{y}$ of interest, we will use $\bm{y}|_A$ to denote the multiplicity of facilities in $A \subseteq \MC{F}$ in $\bm{y}$.
%
%
%
%
%
%
%

%


\subsubsection{The clusters in $\MC{C}_{H'}$}
\label{par-proof-approx-cfl-cfc-cost-ch}

%
%
The following lemma, which regards the assignment radius of the outlier clients in $H$, follows directly from the algorithm design and triangle inequality.

\begin{restatable}{lemma}{lemmaboundalphadist}
\label{lemma-bound-alpha-dist}
For any $j \in H$ and $i \in G$ such that $x'^{(\text{II})}_{i,j} > 0$, we have $c_{i,j} \; \le \; \alpha_j$.
\end{restatable}

\begin{proof}
By the algorithm design, the outlier client $j$ is created and assigned to $i$ in $\bm{x}'^{(0)}$ only when $x'^{(0)}_{i,p(j)} > 0$.
This implies that $c_{i,p(j)} \; \le \; \alpha_{p(j)}$ by complementary slackness condition.
By triangle inequality and the definition of $\alpha_j$, it follows that 
$$c_{i,j} \; = \; c_{i,w(j)} \; \le \; c_{w(j),p(j)} \; + \; c_{i,p(j)} \; \le \; c_{w(j),p(j)} \; + \; \alpha_{p(j)} \; = \; \alpha_j.$$
\end{proof}

In the following lemma, we bound the overall assignment cost in $\left.\bm{x}^\circ\right|_{G,\MC{D}}$ in terms of that in $\bm{x}''$ and $\vphantom{\text{\Large T}_\text{\Large T}} \left.\bm{x}'^{(\text{II})}\right|_{G,\MC{D}\cup H}$.
The proof follows from the way the clusters in $\MC{C}_{H'}$ are formed and the way the assignments are bundled.
%

\begin{restatable}{lemma}{lemmaboundoutlierlpreroutingcost}
\label{lemma-bound-outlier-lp-rerouting-cost}
%
$$\sum_{i \in G, \; j \in \MC{D}} c_{i,j} \cdot x^\circ_{i,j} \; \le \; \sum_{i \in G, \; j \in U} c_{i,j} \cdot x''_{i,j} \; + \; \sum_{i \in G} \; \sum_{j \in \MC{D} \cup H} t'_j \cdot \left( \; c_{i,p(j)} + \alpha_j \; \right) \cdot x'^{(\text{II})}_{i,j}.$$
\end{restatable}

\begin{proof}
%
%

\smallskip

By the definition of $\bm{x}^\circ$ and $\bm{h}$, we have 
\begin{align}
\sum_{i \in G, \; j \in \MC{D}} c_{i,j} \cdot x^\circ_{i,j} \;\; 
& = \;\; \sum_{i \in G, \; j \in \MC{D}} \; \sum_{k \in \{j\} \cup H(j)} c_{i,j} \cdot h_{i,k} \;\; = \;\; \sum_{i \in G} \; \sum_{j \in \MC{D} \cup H} c_{i,p(j)} \cdot h_{i,j}  \notag \\[6pt]
& = \;\; \sum_{i \in G} \; \sum_{j \in \MC{D} \cup H} \; c_{i,p(j)} \cdot \sum_{w \in U} \; x''_{i,w} \cdot \frac{1}{d_w} \cdot \sum_{k \in H(w), \; i' \in B(k)} t'_j \cdot x'^{(\text{II})}_{i',j}.
\label{ieq-bound-outlier-lp-rerouting-cost-1}
\end{align}
By triangle inequality, for any $i \in G$, $j \in \MC{D} \cup H$, $w \in U$, $k \in H(w)$, and $i' \in B(k)$ such that $x'^{(\text{II})}_{i',j} > 0$, we have
$$c_{i,p(j)} \; \le \; c_{i,w} \; + \; c_{i',w} \; + \; c_{i',p(j)} \; \le \; c_{i,w} \; + \; \alpha_k \; + \; c_{i',p(j)},$$
where the last inequality follows from Lemma~\ref{lemma-bound-alpha-dist} and the fact that $i' \in B(k)$ implies that $x'^{(\text{II})}_{i',k} > 0$.
See also Figure~\ref{fig-outlier-client-assignment} for an illustration.
%
%
%
%
%

\begin{figure*}[h]
\centering
\fbox{
\begin{minipage}{.4\textwidth}
\centering
\includegraphics[scale=1]{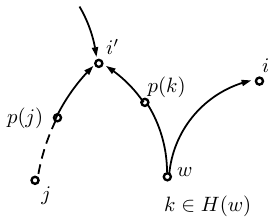}
\end{minipage}
}
\caption{An illustration on the bundled assignment from $w\in U$ to $i \in G$ and unbundled assignments for $k \in H(w)$, $i' \in B(k)$ such that $x'^{(\text{II})}_{i',j} > 0$.}
\label{fig-outlier-client-assignment}
\end{figure*}

\smallskip

Applying the above inequality on~(\ref{ieq-bound-outlier-lp-rerouting-cost-1}) with proper rearrangement, we obtain the following upper-bound for $\sum_{i \in G, \; j \in \MC{D}} c_{i,j} \cdot x^\circ_{i,j}$
\begin{align*}
& \sum_{i \in G, \; w \in U} c_{i,w} \cdot x''_{i,w} \cdot \frac{1}{d_w} \cdot \sum_{k \in H(w), \; i' \in B(k)} \; \sum_{j \in \MC{D}\cup H} t'_j \cdot x'^{(\text{II})}_{i',j} \\[4pt]
& \qquad\quad + \;\; \sum_{i \in G, \; w \in U} \; x''_{i,w} \cdot \frac{1}{d_w} \cdot \sum_{k \in H(w), \; i' \in B(k)} \; \sum_{j \in \MC{D} \cup H} \left( \; c_{i',p(j)} + \alpha_k \; \right) \cdot t'_j \cdot x'^{(\text{II})}_{i',j}.
\end{align*}
Applying the definition of $d_w$ on the former item and the fact that $\sum_{i \in G} x''_{i,w} = d_w$ 
on the latter item, the above becomes
\begin{align*}
\sum_{i \in G, \; w \in U} c_{i,w} \cdot x''_{i,w} \; + \; \sum_{w \in U} \; \sum_{k \in H(w), \; i' \in B(k)} \; \sum_{j \in \MC{D} \cup H} \left( \; c_{i',p(j)} + \alpha_k \; \right) \cdot t'_j \cdot x'^{(\text{II})}_{i',j}. 
\end{align*}
By the design of the algorithm, for any $w \in U, \; k \in H(w), \; i' \in B(k)$, and for any $j \in \MC{D} \cup H$ with $x'^{(\text{II})}_{i',j} > 0$, we have $\alpha_k \le \alpha_j$, since $k$ is selected as cluster center because of having the smallest $\alpha$ value.
Therefore, the above is further upper-bounded by
\begin{align*}
& \sum_{i \in G, \; w \in U} c_{i,w} \cdot x''_{i,w} \; + \; \sum_{w \in U} \; \sum_{k \in H(w), \; i' \in B(k)} \; \sum_{j \in \MC{D} \cup H} \left( \; c_{i',p(j)} + \alpha_j \; \right) \cdot t'_j \cdot x'^{(\text{II})}_{i',j}.
\end{align*}
Applying the fact that the satellite facilities of clusters in $\MC{C}_{H'}$ forms a partition of $G$, the above is exactly
\begin{align*}
\sum_{i \in G, \; j \in U} c_{i,j} \cdot x''_{i,j} \; + \; \sum_{i \in G} \; \sum_{j \in \MC{D} \cup H} t'_j \cdot \left( \; c_{i,p(j)} + \alpha_j \; \right) \cdot x'^{(\text{II})}_{i,j}.
\end{align*}
%
%
\end{proof}

%

%
%
%

\smallskip

\noindent
The following two lemmas
bound the overall cost incurred by $\left.\bm{y}^*\right|_{G}$ and $\bm{x}''$ by the cost of $\left.\bm{y}'^{(\text{II})}\right|_{G}$, $\left.\bm{y}'^{(0)}\right|_{U}$, and $\left.\bm{x}'^{(\text{II})}\right|_{G, \MC{D}\cup H}$.
%
%

\begin{lemma} \label{lemma-bound-outlier-lp-assignment-cost}
We have
$$\sum_{i \in G} \left\lceil y''_i \right\rceil \; + \sum_{i \in G, \; j \in U} c_{i,j} \cdot x''_{i,j} \hspace{4pt} \le \hspace{4pt} 2 \cdot \sum_{i \in G}y'^{(\text{II})}_i \; + \; |L| \; + \sum_{i \in G} \; \sum_{j \in \MC{D} \cup H} t'_j \cdot \alpha_j \cdot x'^{(\text{II})}_{i,j} ,$$
where $L := \left\{ \; \vphantom{\text{\Large T}} i \in G \; \colon \; 0 < y''_i < 1 \; \right\}$.
\end{lemma}

\begin{proof}
Since $(\bm{x}'', \bm{y}'')$ is optimal for~\ref{LP-CFL-outliers}, by Lemma~\ref{lemma-assignment-lp-feasibility}, the cost of $(\bm{x}'', \bm{y}'')$ is no more than that of $\left( \; \left.\bm{g}\right|_{G,U} \; , \; 2\left.\bm{y}'\right|_G \; \right)$.
Hence,
\begin{align*}
\sum_{i \in G} y''_i \; + \sum_{i \in G, \; j \in U} c_{i,j} \cdot x''_{i,j} \hspace{4pt} 
& \le \hspace{4pt} 2 \cdot \sum_{i \in G}y'^{(\text{II})}_i \; + \; \sum_{i \in G, \; j \in U} \; c_{i,j} \cdot g_{i,j} \\[6pt]
& = \hspace{4pt} 2 \cdot \sum_{i \in G}y'^{(\text{II})}_i \; + \; \sum_{\substack{i \in G, \; w \in U, \\[3pt] i \in B(k) \text{ for some } k \in H(w)}} \; c_{i,w} \cdot \sum_{\ell \in \MC{D} \cup H} t'_\ell \cdot x'^{(\text{II})}_{i,\ell}.
\end{align*}
By the algorithm setting and Lemma~\ref{lemma-bound-alpha-dist}, for any $i \in G$ and $w \in U$ such that $i \in B(k)$ for some $k \in H(w)$, i.e., facility $i$ belongs to some cluster centered at some $k \in H(w)$, we have $c_{i,w} \le \alpha_k \le \alpha_\ell$, for any $\ell \in \MC{D} \cup H$ with $x'^{(\text{II})}_{i,\ell} > 0$.
Therefore,
\begin{align*}
\sum_{i \in G} y''_i \; + \sum_{i \in G, \; j \in U} c_{i,j} \cdot x''_{i,j} \hspace{4pt} 
& \le \hspace{4pt} 2 \cdot \sum_{i \in G}y'^{(\text{II})}_i \; + \; \sum_{\substack{i \in G, \; w \in U, \\[3pt] i \in B(k) \text{ for some } k \in H(w)}} \; \sum_{\ell \in \MC{D} \cup H} t'_\ell \cdot \alpha_\ell \cdot x'^{(\text{II})}_{i,\ell} \\[6pt]
& = \hspace{4pt} 2 \cdot \sum_{i \in G}y'^{(\text{II})}_i \; + \; \sum_{i \in G, \; j \in \MC{D} \cup H} t'_j \cdot \alpha_j \cdot x'^{(\text{II})}_{i,j},
\end{align*}
where the last equality follows from the fact that the set of satellite facilities of clusters in $\MC{C}_{H'}$ forms a partition of $G$.
Applying the definition of $L$ competes the proof of this lemma.
\end{proof}

The following lemma follows from the fact that $(\bm{x}'', \bm{y}'')$ is a basic solution for~\ref{LP-CFL-outliers}.

\begin{restatable}{lemma}{lemmaboundoutlierlpbasicsolution}
\label{lemma-bound-outlier-lp-basic-solution}
%
$$|L| \; \le \; |U|, \quad \text{where } L := \left\{ \; \vphantom{\text{\LARGE T}} i \in G \; \colon \; 0 < y''_i < 1 \; \right\}.$$
\end{restatable}

\begin{proof}
%
%
%
Consider the set of constraints in~\ref{LP-CFL-outliers} that hold with equality at $(\bm{x}'', \bm{y}'')$, for which we denote by
$\MC{E}^{(=)}$ in the following.
Let $M_3$ and $M_4$ denote the set of constraints in $\MC{E}^{(=)}$ of the types~\ref{LP-CFL-outliers-M3} and~\ref{LP-CFL-outliers-M4}, respectively.
Formally,
$$M_3 := \BIGBP{ \vphantom{\dfrac{}{}}\hspace{3pt} i \hspace{2pt} \colon \hspace{2pt} \fbox{$\vphantom{{\dfrac{}{}}^{\bigcup}}$ $y_i \le 1$ } \in \MC{E}^{(=)} \hspace{3pt} }$$
and
$$M_4 := \BIGBP{ \vphantom{\dfrac{}{}}\hspace{3pt} (i,j) \hspace{2pt} \colon \hspace{2pt} \fbox{$\vphantom{{\dfrac{}{}}^{\bigcup}}$ $x_{i,j} \ge 0$ } \in \MC{E}^{(=)} \hspace{3pt} } \cup \BIGBP{ \vphantom{\dfrac{}{}}\hspace{3pt} i \hspace{2pt} \colon \hspace{2pt} \fbox{$\vphantom{{\dfrac{}{}}^{\bigcup}}$ $y_i \ge 0$ } \in \MC{E}^{(=)} \hspace{3pt} }.$$
Let $X$ be the number of variables in~\ref{LP-CFL-outliers}.
Since $(\bm{x}'', \bm{y}'')$ is a basic solution for~\ref{LP-CFL-outliers}, it follows that, the coefficient matrix of $\MC{E}^{(=)}$ is of full-rank, i.e., has rank $X$.
Since $M_3$ and $M_4$ are linearly independent, there exists a subset $\MC{E}' \subseteq \MC{E}^{(=)}$ of linearly independent constraints such that $M_3 \cup M_4 \subseteq \MC{E}'$ and $|\MC{E}'| = X$.

\smallskip

Let $\MC{E}'' := \MC{E}' \setminus ( M_3 \cup M_4)$ and modify the constraints in $\MC{E}''$ by setting the variable $y_i$ to be $1$ for all $i \in M_3$.
Similarly, modify $\MC{E}''$ by setting $x_{i,j}$ to be zero for all $(i,j) \in M_4$ and $y_i$ to be zero for all $i \in M_4$.
By doing so, we removed equally many constraints and variables from $\MC{E}'$.
Since $M_3$ and $M_4$ are linearly independent, it follows that $$\text{rank}(\MC{E}'') \; = \; \text{rank}(\MC{E}') - | M_3 \cup M_4|,$$
and the coefficient matrix of $\MC{E}''$ is still of full-rank.
%


%
$$\text{Let} \enskip M_1 := \BIGBP{ \vphantom{\dfrac{}{}}\hspace{3pt} j \hspace{2pt} \colon \hspace{2pt} \fbox{$\vphantom{{\dfrac{}{}}^{\bigcup}}$ $\sum_{i \in G} x_{i,j} = d_j$ } \hspace{3pt} } \enskip \text{and} \enskip M_2 := \BIGBP{ \vphantom{\dfrac{}{}}\hspace{3pt} i \hspace{2pt} \colon \hspace{2pt} \fbox{$\vphantom{{\dfrac{}{}}^{\bigcup}}$ $\sum_{j \in U} x_{i,j} \le u_i \cdot y_i$ } \in \MC{E}'' \hspace{3pt} } .$$
Also let $H := \left\{ \vphantom{\text{\Large T}} \; (i,j) \; \colon \; x''_{i,j} \neq 0 \; \right\}$.
It follows by the above setting that, $L \cup H$ corresponds exactly to the set of variables in $\MC{E}''$.
Since the coefficient matrix of $\MC{E}''$ has full rank,
the pivot in each row of the matrix defines a one-to-one mapping $\phi \colon L \cup H \rightarrow M_1 \cup M_2$ between the variables and the constraints.

\smallskip

Consider each $i \in L$.
Since the variable $y_i$ appears exactly in one constraint in $M_2$, the mapping must map $y_i$ to the constraint it corresponds to, i.e., $\phi(i) = i$.
Since the constraint $i$ corresponds to in $M_2$ contributes one rank, it is non-degenerated and contains at least one variable in $H$.
Let $x_{i,j}$ be one such variable.
Since $x_{i,j}$ appears in exactly two constraints, i.e., in the one $i$ corresponds to in $M_2$ and the one $j$ corresponds to in $M_1$, and since $\phi(i) = i$, it follows that $\phi\left( (i,j) \right) = j$.
Since the mapping $\phi$ is one-to-one, $j$ cannot be mapped to by other pairs.

\smallskip

Applying the above argument for each $i \in L$ results in a set consisting of distinct clients $j$ from $U$ with the same cardinality.
This shows that $|L| \le |U|$.
\end{proof}

\medskip

Applying Lemma~\ref{lemma-bound-outlier-lp-rerouting-cost}, Lemma~\ref{lemma-bound-outlier-lp-assignment-cost}, Lemma~\ref{lemma-bound-outlier-lp-basic-solution}, and the fact that $y'^{(0)}_i \ge 1/2$ for all $i \in U$, we obtain the following bound for the cost incurred by $\left( \left.\bm{x}^\circ\right|_{G,\MC{D}}, \left.\bm{y}^*\right|_{G} \right)$.
\begin{align}
\sum_{i \in G} y^*_i \; + \sum_{i \in G,\; j \in \MC{D}} c_{i,j} \cdot x^\circ_{i,j} \hspace{4pt} \;\;
& \le \;\; \hspace{4pt} 2\cdot \sum_{i \in G} y'^{(\text{II})}_i \; + \; \sum_{i \in G, \; j \in H} t'_j \cdot \left( \; c_{i,p(j)} + 2 \cdot \alpha_j \; \right) \cdot x'^{(\text{II})}_{i,j} \notag \\[6pt]
& \hspace{15pt} \; + \; 2 \cdot \sum_{i \in U} y'^{(0)}_i \; + \; \sum_{i \in G, \; j \in \MC{D}} t'_j \cdot \left( \; c_{i,j} + 2 \cdot \alpha_j \; \right) \cdot x'^{(\text{II})}_{i,j}  \notag \\[8pt]
& \hspace{-40pt} \le \;\; \hspace{4pt} 2\cdot \sum_{i \in G} y'^{(\text{II})}_i \; + \; \sum_{i \in G, \; j \in H} \left( \; c_{i,p(j)} + 2 \cdot \alpha_j \; \right) \cdot x'^{(\text{II})}_{i,j} \notag \\[6pt]
& \hspace{-40pt} \hspace{15pt} \; + \; 2 \cdot \sum_{i \in U} y'^{(0)}_i \; + \; \sum_{i \in G, \; j \in \MC{D}} \left( \; 2\cdot c_{i,j} + 2 \cdot t'_j \cdot \alpha_j \; \right) \cdot x'^{(\text{II})}_{i,j},
\label{ieq-bound-cluster-h}
\end{align}
where in the last inequality we apply Corollary~\ref{cor-scaling-factor} the fact that $t'_j \le 2$ for all $j \in \MC{D}$ and the definition that $t'_j = 1$ for all $j \in H$.


%
%

%
%


\medskip


\subsubsection{The clusters in $\MC{C}_{D'}$}
\label{par-proof-approx-cfl-cfc-cost-cd}

%
Consider the cost incurred by the clusters in $\MC{C}_{D'}$.
The following lemma bounds the cost incurred by the rounding process for each individual cluster $q$ in $\MC{C}_{D'}$.

\begin{restatable}{lemma}{lemmacostclusterdprime}
\label{lemma-cost-cluster-d-prime}
For any $q \in \MC{C}_{D'}$, we have
\begin{align*}
\text{(i)} \; & & y^*_{i(q)} \; & \le \;\; 2 y'^{(q)}_{i(q)} \; + \; 2\delta_{i(q)} \cdot \sum_{k \in B(q) \setminus \{i(q)\} } y'^{(q)}_k \enskip , \quad \text{and} \\[6pt]
\text{(ii)} \;\; & & \sum_{j \in \MC{D}} c_{i(q),j} \cdot x^\circ_{i(q),j} \; 
& \le \;\; \sum_{j \in D} 2\cdot c_{i(q),j} \cdot x'^{(q)}_{i(q),j} \; + \; \sum_{j \in H} c_{i(q),p(j)} \cdot x'^{(q)}_{i(q),j} \\[2pt]
& & & \hspace{-20pt} + \; \delta_{i(q)} \cdot \sum_{k \in B(q) \setminus \{i(q)\}} \; \left( \; \sum_{\; j \in \MC{D}} \; 2\cdot c_{k,j} \cdot x'^{(q)}_{k,j} \; + \; \sum_{\; j \in H} c_{k,p(j)} \cdot x'^{(q)}_{k,j} \; \right) \\[12pt]
& & & \hspace{40pt} + \; \sum_{j \in \MC{D}} \; 2 \cdot t'_j \cdot \alpha_j \cdot x^*_{i(q),j} \; + \; \sum_{j \in H} \; 2 \cdot \alpha_j \cdot x^*_{i(q),j}.
\end{align*}
\end{restatable}

\begin{proof}
The lemma follows directly from the rounding process for clusters in $\MC{C}_{D'}$.
%
%
%
%
Consider the iteration for which $q$ is formed and ready to be rounded.
%
%
By the algorithm design, the total facility value that has been removed from $F'$ due to the rounding process for $q$ is
\begin{align*}
y'^{(q)}_{i(q)} \; + \; \delta_{i(q)} \cdot \sum_{ k \in B(q) \setminus \{i(q)\}} y'^{(q)}_k \; = \; \frac{1}{2} \; = \; \frac{1}{2} \cdot y^*_{i(q)},
\end{align*}
where in the first equality we apply the definition of $\delta_{i(q)}$.
Attributing the cost of $y^*_{i(q)}$ to the facility value that is removed from $F'$ due to cluster $q$ proves the first part of this lemma.

\smallskip

For the second part, by the definition of $\bm{x}^\circ$ and the way how the algorithm reroutes the assignments from facilities in $B(q) \setminus \{i(q)\}$ to $i(q)$, we have
\begin{align}
\sum_{j \in \MC{D}} c_{i(q),j} \cdot x^\circ_{i(q),j} \;\;
& = \;\; \sum_{j \in \MC{D}} c_{i(q),j} \cdot \left( \; t'_j \cdot x^*_{i(q),j} \; + \; \sum_{\ell \in H(j)} x^*_{i(q),\ell} \; \right)  \notag \\[8pt]
& \hspace{-36pt} = \;\; \sum_{j \in \MC{D}} c_{i(q),j} \cdot t'_j \cdot x'^{(q)}_{i(q),j} \; + \; \sum_{\ell \in H} c_{i(q),p(\ell)} \cdot x'^{(q)}_{i(q),\ell}  \notag \\[2pt]
& \hspace{-30pt} \hspace{2pt}  + \; \delta_{i(q)} \cdot \sum_{k \in B(q) \setminus \{i(q)\}} \left( \; \sum_{j \in \MC{D}} c_{i(q),j} \cdot t'_j \cdot x'^{(q)}_{k,j} \; + \; \sum_{\ell \in H} c_{i(q),p(\ell)} \cdot x'^{(q)}_{k,\ell} \; \right).
\label{ieq-cost-cluster-d-prime-1}
\end{align}
By the algorithm setting, for any $k \in B(q) \setminus \{i(q)\}$ and any $\ell \in H$ with $x'^{(q)}_{k,\ell} > 0$, we have 
$$c_{i(q),p(\ell)} \;\; \le \;\; c_{k,p(\ell)} \; + \; c_{k,j(q)} \; + \; c_{i(q),j(q)} \;\; \le \;\; c_{k,p(\ell)} \; + \; 2\alpha_{j(q)} \;\; \le \;\; c_{k,p(\ell)} \; + \; 2\alpha_\ell,$$
where in the second inequality we apply the fact that $i(q)$ and $k$ are in $B(q)$, which implies that $x'^{(q)}_{i(q),j(q)} > 0$, $x'^{(q)}_{k,j(q)} > 0$, and $\max\left( c_{i(q),j(q)}, c_{k,j(q)} \right) \le \alpha_{j(q)}$ by complementary slackness, and in the last inequality we apply the assumption that $x'^{(q)}_{k,\ell} > 0$, which implies that $\ell \in H'^{(q)}$ and $\alpha_{j(q)} \le \alpha_\ell$ by the way $j(q)$ is selected.
%
%
By a similar argument, we have $c_{i(q),j} \; \le \; c_{k,j} \; + \; 2\alpha_j$ for any $k \in B(q) \setminus \{i(q)\}$ and any $j \in \MC{D}$ with $x'^{(q)}_{k,j} > 0$.

\smallskip

By the above conclusion, Corollary~\ref{cor-scaling-factor}, and the way the assignment $\bm{x}^*$ is formed during the rounding process, 
we have
\begin{align*}
& \delta_{i(q)} \cdot \hspace{-4pt} \sum_{k \in B(q) \setminus \{i(q)\}, \; j \in \MC{D}} \hspace{-4pt} c_{i(q),j} \cdot t'_j \cdot x'^{(q)}_{k,j} \;\; \le \;\; \delta_{i(q)} \cdot \hspace{-4pt} \sum_{k \in B(q) \setminus \{i(q)\}, \; j \in \MC{D}} \hspace{-4pt} \left( c_{k,j} + 2\cdot \alpha_j \right) \cdot t'_j \cdot x'^{(q)}_{k,j}  \\[8pt]
& \hspace{2.6cm} \le \;\; \delta_{i(q)} \cdot \sum_{k \in B(q) \setminus \{i(q)\}, \; j \in \MC{D}} \; 2 \cdot c_{k,j} \cdot x'^{(q)}_{k,j} \; + \; \sum_{j \in \MC{D}} \;  2\cdot t'_j \cdot \alpha_j \cdot x^*_{i(q),j}.
\end{align*}
Similarly, we have
\begin{align*}
\hspace{40pt} \delta_{i(q)} \cdot \hspace{-6pt} \sum_{k \in B(q) \setminus \{i(q)\}, \; \ell \in H} c_{i(q),p(\ell)} \cdot x'^{(q)}_{k,\ell} \;\; \\[6pt]
& \hspace{-100pt} \le \;\; \delta_{i(q)} \cdot \hspace{-6pt} \sum_{k \in B(q) \setminus \{i(q)\}, \; \ell \in H} c_{k,p(\ell)} \cdot x'^{(q)}_{k,\ell} \; + \; \sum_{\ell \in H} \;  2\cdot \alpha_\ell \cdot x^*_{i(q),\ell}.
\end{align*}
Combining the above two inequalities with~(\ref{ieq-cost-cluster-d-prime-1}) and further applying Corollary~\ref{cor-scaling-factor}, we have
%
%
%
\begin{align*}
& \sum_{j \in \MC{D}} c_{i(q),j} \cdot x^\circ_{i(q),j} \;\; \le \;\; \sum_{j \in D} \; 2\cdot c_{i(q),j} \cdot x'^{(q)}_{i(q),j} \; + \; \sum_{\ell \in H} \; c_{i(q),p(\ell)} \cdot x'^{(q)}_{i(q),\ell} \\[2pt]
& \hspace{2cm} + \; \delta_{i(q)} \cdot \sum_{k \in B(q) \setminus \{i(q)\}} \; \left( \; \sum_{\; j \in \MC{D}} \; 2\cdot c_{k,j} \cdot x'^{(q)}_{k,j} \; + \; \sum_{\; \ell \in H} c_{k,p(\ell)} \cdot x'^{(q)}_{k,\ell} \; \right) \\[12pt]
& \hspace{80pt} + \; \sum_{j \in \MC{D}} \; 2 \cdot t'_j \cdot \alpha_j \cdot x^*_{i(q),j} \; + \; \sum_{\ell \in H} \; 2 \cdot \alpha_\ell \cdot x^*_{i(q),\ell}.
\end{align*}
\end{proof}

%

%

%
%
The following lemma, which bounds the cost incurred by $\left.\bm{x}^\circ\right|_{F^*_{D'},\MC{D}}$ and $\left.\bm{y}^*\right|_{F^*_{D'}}$, is obtained by taking summation on the cost given in Lemma~\ref{lemma-cost-cluster-d-prime} over all clusters in $\MC{C}_{D'}$. 
%

\begin{restatable}{lemma}{lemmacostclusterdprimesum}
\label{lemma-cost-cluster-d-prime-sum}
\begin{align}
\text{(i)} \; & & \sum_{i \in F^*_{D'}} y^*_i \;\; & \le \;\; 2\cdot \sum_{i \in I \setminus G} y'^{(0)}_i \; + \; 2\cdot \sum_{i \in G} \left( \; y'^{(0)}_i - y'^{(\text{II})}_i \hspace{1pt} \right) . \label{ieq-bound-cluster-d-p-facility} \\[8pt]
\text{(ii)} \; & & \sum_{i \in F^*_{D'},\; j \in \MC{D}} c_{i,j} \cdot x^\circ_{i,j} \;\; 
& \le \;\; \sum_{i \in F^*_{D'}, \; j \in \MC{D}} 2\cdot t'_j \cdot \alpha_j \cdot x^*_{i,j} \; + \; \sum_{i \in F^*_{D'}, \; j \in H} 2\cdot \alpha_j \cdot x^*_{i,j} \notag \\[2pt]
& & & \hspace{-90pt} + \sum_{i \in G, \; j \in \MC{D}} \hspace{-2pt} 2\cdot c_{i,j} \cdot \left( x'^{(0)}_{i,j} - \sum_{\ell \in H(j)} x'^{(0)}_{i,\ell} - x'^{(\text{II})}_{i,j} \right) 
+ \hspace{-2pt} \sum_{i \in G, \; j \in H} \hspace{-2pt} c_{i,p(j)} \cdot \left( x'^{(0)}_{i,j} - x'^{(\text{II})}_{i,j} \right)  \notag \\[4pt]
& & & \hspace{-80pt} \; + \; \sum_{i \in I \setminus G, \; j \in \MC{D}} 2\cdot c_{i,j} \cdot \left( \; x'^{(0)}_{i,j} - \sum_{\ell \in H(j)} x'^{(0)}_{i,\ell} \; \right) \; + \; \sum_{i \in I \setminus G, \; j \in H} c_{i,p(j)} \cdot x'^{(0)}_{i,j}. \label{ieq-bound-cluster-d-p-assignment}
\end{align}
%
%
\end{restatable}

\begin{proof}
Consider the first part of Lemma~\ref{lemma-cost-cluster-d-prime}. We have
\begin{align}
\sum_{i \in F^*_{D'}} y^*_i \;\; = \;\; \sum_{q \in \MC{C}_{D'}} y^*_{i(q)} \;\; \le \;\; \sum_{q \in \MC{C}_{D'}} \left( \; 2 y'^{(q)}_{i(q)} \; + \; 2\delta_{i(q)} \cdot \sum_{k \in B(q) \setminus \{i(q)\} } y'^{(q)}_k \; \right).
\label{ieq-cost-cluster-d-prime-sum-1}
\end{align}
%
%
By the design of the scaled-down operation for rounding each $q \in \MC{C}_{D'}$, we know that, for each $k \in I \cap D'^{(q)}$, the facility value $y'^{(q)}_k$ decreases exactly by $\delta_{i(q)} \cdot y'^{(q)}_k$ if $k \neq i(q)$ and $y'^{(q)}_k$ otherwise. 
%
%
Hence, by summing up the 
RHS of~(\ref{ieq-cost-cluster-d-prime-sum-1}) in a backward manner, we obtain
$$\sum_{i \in F^*_{D'}} y^*_i \;\; \le \;\; 2\cdot \sum_{i \in I \setminus G} y'^{(0)}_i \; + \; 2\cdot \sum_{i \in G} \left( \; y'^{(0)}_i - y'^{(\text{II})}_i \hspace{1pt} \right).$$
%
%

\medskip

The second part of this lemma follows from an analogous 
argument.
By the second part of Lemma~\ref{lemma-cost-cluster-d-prime}, we have
\begin{align}
& \sum_{i \in F^*_{D'}, \; j \in \MC{D}} c_{i,j} \cdot x^\circ_{i,j} \;\; = \;\; \sum_{q \in \MC{C}_{D'}} \; \sum_{j \in \MC{D}} \; c_{i(q),j} \cdot x^\circ_{i(q),j} \;  \notag \\[4pt]
& \hspace{1.6cm} \le \; \sum_{q \in \MC{C}_{D'}} \left( \; \sum_{j \in D} 2\cdot c_{i(q),j} \cdot x'^{(q)}_{i(q),j} \; + \; \sum_{j \in H} c_{i(q),p(j)} \cdot x'^{(q)}_{i(q),j} \; \right)  \notag \\[6pt]
& \hspace{2.2cm} + \; \sum_{q \in \MC{C}_{D'}} \delta_{i(q)} \cdot \sum_{k \in B(q) \setminus \{i(q)\}} \; \left( \; \sum_{\; j \in \MC{D}} \; 2\cdot c_{k,j} \cdot x'^{(q)}_{k,j} \; + \; \sum_{\; j \in H} c_{k,p(j)} \cdot x'^{(q)}_{k,j} \; \right)  \notag \\[6pt]
& \hspace{2.2cm} + \; \sum_{q \in \MC{C}_{D'}} \left( \;  \sum_{j \in \MC{D}} \; 2 \cdot t'_j \cdot \alpha_j \cdot x^*_{i(q),j} \; + \; \sum_{j \in H} \; 2 \cdot \alpha_j \cdot x^*_{i(q),j} \; \right).
\label{ieq-cost-cluster-d-prime-sum-2}
\end{align}
For the last item in~(\ref{ieq-cost-cluster-d-prime-sum-2}), we have
\begin{align*}
& \sum_{q \in \MC{C}_{D'}} \left( \;  \sum_{j \in \MC{D}} \; 2 \cdot t'_j \cdot \alpha_j \cdot x^*_{i(q),j} \; + \; \sum_{j \in H} \; 2 \cdot \alpha_j \cdot x^*_{i(q),j} \; \right) \\[8pt]
& \hspace{3cm} = \; \sum_{i \in F^*_{D'}, \; j \in \MC{D}} 2\cdot t'_j \cdot \alpha_j \cdot x^*_{i,j} \; + \; \sum_{i \in F^*_{D'}, \; j \in H} 2\cdot \alpha_j \cdot x^*_{i,j}
\end{align*}
by definition.
For the remaining items in the RHS of~(\ref{ieq-cost-cluster-d-prime-sum-2}), we apply the same argument and charge the cost to the assignment values decreased due to the scaled-down operation when rounding each $q \in \MC{C}_{D'}$.
Hence, the remaining items can be bounded by
\begin{align*}
& \sum_{i \in G, \; j \in \MC{D}} \hspace{-2pt} 2\cdot c_{i,j} \cdot \left( x'^{(0)}_{i,j} - \sum_{\ell \in H(j)} x'^{(0)}_{i,\ell} - x'^{(\text{II})}_{i,j} \right) 
+ \hspace{-2pt} \sum_{i \in G, \; j \in H} \hspace{-2pt} c_{i,p(j)} \cdot \left( x'^{(0)}_{i,j} - x'^{(\text{II})}_{i,j} \right)  \\[4pt]
& \hspace{14pt} \; + \; \sum_{i \in I \setminus G, \; j \in \MC{D}} 2\cdot c_{i,j} \cdot \left( \; x'^{(0)}_{i,j} - \sum_{\ell \in H(j)} x'^{(0)}_{i,\ell} \; \right) \; + \; \sum_{i \in I \setminus G, \; j \in H} c_{i,p(j)} \cdot x'^{(0)}_{i,j}.
\end{align*}
%
%
%
%
%
%
This proves the lemma.
\end{proof}

%

%


%


\subsubsection{The overall guarantee}
\label{par-proof-approx-cfl-cfc-overall-cost}

%
%
Combining Inequality~(\ref{ieq-bound-cluster-h}), Inequality~(\ref{ieq-bound-cluster-d-p-facility}), and Inequality~(\ref{ieq-bound-cluster-d-p-assignment}) with proper rearrangement of the items, 
we obtain  
%
\begin{align}
& \hspace{-12pt} \psi( \bm{x}^\circ, \bm{y}^* ) \; 
 = \; \psi \left( \left.\bm{x}^\circ\right|_{U,\MC{D}}, \left.\bm{y}^*\right|_{U} \right) \; + \; \psi \left( \left.\bm{x}^\circ\right|_{F^*_{D'},\MC{D}}, \left.\bm{y}^*\right|_{F^*_{D'}} \right) \; + \; \psi \left( \left.\bm{x}^\circ\right|_{G,\MC{D}}, \left.\bm{y}^*\right|_{G} \right)  \notag \\[4pt]
& \hspace{-4pt} \le \;\; 4\cdot \sum_{i \in U} y'^{(0)}_i \; + \; 2\cdot \sum_{i \in I} y'^{(0)}_i \; + \; \sum_{i \in U, \; j \in \MC{D}} c_{i,j} \cdot x'^{(0)}_{i,j}  \notag \\[2pt]
& \hspace{-4pt} + \sum_{i \in I, \; j \in \MC{D}} 2\cdot c_{i,j} \cdot \left( x'^{(0)}_{i,j} - \sum_{\ell \in H(j)} x'^{(0)}_{i,\ell} \right) \; + \; \sum_{i \in I, \; j \in H} c_{i,p(j)} \cdot x'^{(0)}_{i,j}  \notag \\[3pt]
& \hspace{-4pt} + \; \sum_{j \in \MC{D}} 2 \cdot t'_j \cdot \alpha_j \cdot \left( \sum_{i \in F^*_{D'}} x^*_{i,j} + \sum_{i \in G} x'^{(\text{II})}_{i,j} \right) \; + \; \sum_{j \in H} 2\cdot \alpha_j \cdot \left( \sum_{i \in F^*_{D'}} x^*_{i,j} + \sum_{i \in G} x'^{(\text{II})}_{i,j} \right) .
\label{ieq-bound-overall-1}
\end{align}
%
%

\smallskip
\smallskip

\noindent
%
Consider the item $\sum_{i \in I, \; j \in H} c_{i,p(j)} \cdot x'^{(0)}_{i,j}$ in~(\ref{ieq-bound-overall-1}).
%
%
We have
\begin{align}
\sum_{i \in I, \; j \in H} c_{i,p(j)} \cdot x'^{(0)}_{i,j} \; = \; \sum_{i \in I, \; j \in \MC{D}, \; \ell \in H(j)} c_{i,j} \cdot x'^{(0)}_{i,\ell}.
\label{ieq-bound-overall-2-1}
\end{align}
%
%

\smallskip

\noindent
%
Consider the last two items in~(\ref{ieq-bound-overall-1}).
Applying the definition of $t'_j$ for each $j \in \MC{D}$ with $\sum_{i \in F^*_{D'}} x^*_{i,j} + \sum_{i \in G} x'^{(\text{II})}_{i,j} > 0$, we have 
\begin{align}
\sum_{j \in \MC{D}} \; 2 \cdot t'_j \cdot \alpha_j \cdot \left( \; \sum_{i \in F^*_{D'}} x^*_{i,j} + \sum_{i \in G} x^{(\text{II})}_{i,j} \; \right) \; = \; \sum_{j \in \MC{D}} \; 2 \cdot \alpha_j \cdot \left( \; 1- \sum_{i \in U} x'^{(0)}_{i,j} - r'_j \; \right).
\label{ieq-bound-overall-2-1-1}
\end{align}

\smallskip

By the algorithm design, we have $\alpha_j = \alpha_{p(j)} + c_{w(j),p(j)}$ for any $j \in H$.
Further applying the fact that the demand $d_j$ of any outlier client $j \in H$ is fully-assigned when created and remains fully-assigned during the rounding process,
it follows that
\begin{align}
\sum_{j \in H} \; 2 \cdot \alpha_j \cdot \left( \; \sum_{i \in F^*_{D'}} x^*_{i,j} + \sum_{i \in G} x^{(\text{II})}_{i,j} \; \right) \; 
& = \;\; \sum_{j \in H} \; 2\cdot \left( \; \alpha_{p(j)} + c_{w(j),p(j)} \; \right) \cdot d_j  \notag \\[6pt]
& \hspace{-180pt} \le \;\; \sum_{j \in \MC{D}} \; 2\cdot \alpha_j \cdot r'_j \; + \; \sum_{j \in H} \; 2\cdot c_{w(j),p(j)} \cdot x'^{(0)}_{w(j),p(j)} \; \le \;\; \sum_{j \in \MC{D}} \; 2\cdot \alpha_j \cdot r'_j \; + \; 2\cdot \sum_{i \in U, \; j \in \MC{D}} c_{i,j} \cdot x'^{(0)}_{i,j},
\label{ieq-bound-overall-2-1-2}
\\[-26pt] \notag
\end{align}
where in the second last inequality we use the fact that 
$$\sum_{k \in H(j)} \alpha_{p(j)} \cdot d_k \;\; = \;\; \alpha_j \cdot \sum_{k \in H(j)} d_k \;\; = \;\; \alpha_j \cdot \sum_{k \in H(j)} r'_j \cdot \frac{x'^{(0)}_{w(j),p(j)}}{\sum_{i \in U} x'^{(0)}_{i,p(j)}} \;\; = \;\; \alpha_j \cdot r'_j$$ for all $j \in \MC{D}$ and the fact that $d_j \le x'^{(0)}_{w(j),p(j)}$ for any $j \in H$ by the definition of $d_j$.
%

\medskip

%
%
%
%

\smallskip

%
Combining Equality~(\ref{ieq-bound-overall-2-1}), Equality~(\ref{ieq-bound-overall-2-1-1}), Inequality~(\ref{ieq-bound-overall-2-1-2}) with Inequality~(\ref{ieq-bound-overall-1}), we obtain
\begin{align}
\psi( \bm{x}^\circ, \bm{y}^* ) \; 
& \le \hspace{4pt} 4\cdot \sum_{i \in U} y'^{(0)}_i \; + \; 3 \cdot \sum_{i \in U, \; j \in \MC{D}} c_{i,j} \cdot x'^{(0)}_{i,j}  \notag \\
& \hspace{-1pt} \; + \; 2\cdot \sum_{i \in I} y'^{(0)}_i \; + \; 2\cdot \sum_{i \in I, \; j \in \MC{D}} c_{i,j} \cdot x'^{(0)}_{i,j} \; + \; \sum_{j \in \MC{D}} \; 2 \cdot \left( \; 1- \sum_{i \in U} x'^{(0)}_{i,j} \; \right) \cdot \alpha_j .
\label{ieq-bound-overall}
\end{align}

\begin{figure*}[h]
\centering
\fbox{
\begin{minipage}{.8\textwidth}
\begin{align}
& \text{min} & & \sum_{i \in \MC{F}} \; y_i \; + \; \sum_{i \in \MC{F}, \hspace{2pt} j \in \MC{D}} c_{i,j}\cdot x_{i,j} & &  
\tag*{LP-(N)} \\[3pt]
& \text{s.t.} & & \sum_{i\in \MC{F}} \; x_{i,j} \; \ge \; 1, & & \forall j\in \MC{D} \label{cons-P-1} \tag*{(N-1)} \\[3pt]
& & & \sum_{j \in \MC{D}} \; x_{i,j} \; \le \; u_i \cdot y_i, & & \forall i \in \MC{F} \label{cons-P-2} \tag*{(N-2)} \\[3pt]
& & & 0 \; \le \; x_{i,j} \; \le \; y_i, & & \forall i \in \MC{F}, j \in \MC{D} \quad \label{cons-P-3} \tag*{(N-3)} \\[5pt]
& & & 0 \; \le \; y_i \; \le \; 1, & & \forall i \in \MC{F}. \label{cons-P-4} \tag*{(N-4)}
\end{align}
\vspace{-6pt}
\end{minipage}\quad
}
\fbox{
\begin{minipage}{.8\textwidth}
\begin{align}
& \text{max} & & \sum_{j \in \MC{D}} \; \alpha_j \; - \; \sum_{i \in \MC{F}} \; \eta_i & &  
\tag*{LP-(DN)} \\[6pt]
& \text{s.t.} & & \alpha_j \; \le \; \beta_i \; + \; \Gamma_{i,j} \; + \; c_{i,j}, \enskip & & \forall i \in \MC{F}, j\in \MC{D}, & & \label{cons-D-1} \tag*{(D-1)} \\[6pt]
& & & u_i \cdot \beta_i \; + \; \sum_{j \in \MC{D}} \Gamma_{i,j} \; \le \; 1 \; + \; \eta_i, & & \forall i \in \MC{F},  \label{cons-D-2} \tag*{(D-2)} \\[1pt]
& & & \alpha_j, \; \beta_i, \; \Gamma_{i,j}, \; \eta_i \; \ge \; 0, & & \forall i \in \MC{F}, j \in \MC{D}. \label{cons-D-3} \tag*{(D-3)}
\end{align}
\vspace{-6pt}
\end{minipage}\quad
}
\caption{(Restate for further reference) The natural LP formulations for CFL-CFC.}
\label{fig-natural-LP}
\end{figure*}

\noindent
The following lemma follows from complementary slackness between 
$(\bm{x}',\bm{y}')$ and $(\bm{\alpha},\bm{\beta},\bm{\Gamma},\bm{\eta})$, and the fact that $0 < y'^{(0)}_i < 1$ for all $i \in I$.

\begin{restatable}{lemma}{lemmaboundprimaldualforI}
\label{lemma-bound-primal-dual-for-I}
$$\sum_{j \in \MC{D}} \left( \; 1 - \sum_{i \in U} x'^{(0)}_{i,j} \; \right) \cdot \alpha_j \;\; \le \;\; \sum_{i \in I} \; y'^{(0)}_i \; + \; \sum_{i \in I, \; j \in \MC{D}} \; c_{i,j} \cdot x'^{(0)}_{i,j}.$$
%
\end{restatable}

\begin{proof}
Consider any $i \in I$ and the cost incurred.
From the fact that $(\bm{x}',\bm{y}')$ and $(\bm{\alpha},\bm{\beta},\bm{\Gamma},\bm{\eta})$ are optimal primal and dual solutions for~\ref{LP-natural-CFL} and~\ref{LP-natural-dual}, by complementary slackness conditions, we have
$$y'_i \; + \; \sum_{j\in \MC{D}} \; c_{i,j} \cdot x'_{i,j} \;\; = \;\; y'_i \cdot \left( \; u_i \cdot \beta_i \; + \; \sum_{j \in \MC{D}} \Gamma_{i,j} \; \right) \; + \; \sum_{j \in \MC{D}} c_{i,j} \cdot x'_{i,j}$$
since $y'_i > 0$ implies that constraint~\ref{cons-D-2} is tight, and $y'_i < 1$ implies that the dual variable $\eta_i$ must be zero.
By the fact that $\beta_i > 0$ implies that constraint~\ref{cons-P-2} is tight and $\Gamma_{i,j} > 0$ implies that constraint~\ref{cons-P-3} is tight, the above becomes 
$$\beta_i \cdot \sum_{j\in \MC{D}} \; x'_{i,j} \; + \; \sum_{j \in \MC{D}} \; \Gamma_{i,j}\cdot x'_{i,j} \; + \; \sum_{j \in \MC{D}} \; c_{i,j} \cdot x'_{i,j}.$$
Finally, applying the fact that $x'_{i,j} > 0$ implies that constraint~\ref{cons-D-1} is tight, we obtain
$$y'_i \; + \; \sum_{j\in \MC{D}} \; c_{i,j} \cdot x'_{i,j} \;\; = \;\; \sum_{j \in \MC{D}} \; \alpha_j \cdot x'_{i,j}.$$
%

\smallskip

Taking summation over all facilities in $I$, we obtain
$$\sum_{i \in I} \; y'_i \; + \; \sum_{i \in I, \; j \in \MC{D}} c_{i,j}\cdot x'_{i,j} \;\; = \;\; \sum_{j \in \MC{D}, \; i \in I} \; \alpha_j \cdot x'_{i,j} \;\; \ge \;\; \sum_{j \in \MC{D}} \left( \; 1 - \sum_{i \in U} x'_{i,j} \; \right) \cdot \alpha_j.$$
where in the last equality we apply constraint~\ref{cons-P-1} 
for each $j \in \MC{D}$.
\end{proof}

\smallskip

\noindent
Applying Lemma~\ref{lemma-bound-primal-dual-for-I} on Inequality~(\ref{ieq-bound-overall}), we obtain
\begin{align*}
\psi( \bm{x}^\circ, \bm{y}^* ) \; 
& \; \le \; 4 \cdot \sum_{i \in \MC{F}} y'^{(0)}_i \; + \; 4 \cdot \sum_{i \in \MC{F}, \; j \in \MC{D}} c_{i,j} \cdot x'^{(0)}_{i,j}, 
\end{align*}
and Theorem~\ref{theorem-cfl-ufc-approx} is proved.

\bigskip
\medskip

%



\bibliographystyle{plain}
\bibliography{approx_cfl}

\end{document}